\documentclass[12pt]{article}

\usepackage[margin=1in]{geometry}

\usepackage{setspace}

\usepackage{graphicx}

\usepackage{enumitem}

\usepackage{microtype}
\usepackage{subfigure}
\usepackage{booktabs}
\usepackage{comment}
\usepackage{wrapfig}
\usepackage{arydshln}
\usepackage{enumitem}
\usepackage{multirow}
\usepackage{url}
\usepackage{natbib}
\usepackage{authblk}

\RequirePackage[colorlinks,citecolor=blue,linkcolor=blue,urlcolor=blue,pagebackref]{hyperref}

\usepackage{amsmath}
\usepackage{amssymb}
\usepackage{mathtools}
\usepackage{amsfonts}
\usepackage{bm}
\usepackage{bbm}

\usepackage{algorithm}
\usepackage{algorithmic}

\newcommand{\ip}[2]{\left\langle #1,#2\right\rangle}

\def\1{\bm{1}}

\def\rmd{{\mathrm{d}}}

\DeclareMathAlphabet{\mathsfit}{\encodingdefault}{\sfdefault}{m}{sl}
\SetMathAlphabet{\mathsfit}{bold}{\encodingdefault}{\sfdefault}{bx}{n}

\def\calA{{\mathcal{A}}}

\def\calE{{\mathcal{E}}}

\def\calN{{\mathcal{N}}}

\def\calP{{\mathcal{P}}}

\def\calR{{\mathcal{R}}}
\def\calS{{\mathcal{S}}}
\def\calT{{\mathcal{T}}}

\def\calW{{\mathcal{W}}}

\def\bbR{{\mathbb{R}}}

\newcommand{\E}{\mathbb{E}}

\newcommand{\KL}{D_{\mathrm{KL}}}
\newcommand{\Var}{\mathrm{Var}}

\DeclareMathOperator*{\argmax}{arg\,max}

\newcommand{\abs}[1]{\left|#1\right|}

\newcommand{\norm}[1]{\left\|#1\right\|}

\usepackage{amsthm}

\theoremstyle{plain}

\newtheorem{theorem}{Theorem}[section]
\newtheorem{lemma}[theorem]{Lemma}

\newtheorem{corollary}[theorem]{Corollary}
\newtheorem{proposition}[theorem]{Proposition}

\newtheorem{assumption}{Assumption}[section]
\newtheorem*{example}{Example}
\newtheorem*{remark}{Remark}

\usepackage[textsize=tiny]{todonotes}
\usepackage{multirow}
\usepackage{wrapfig}
\usepackage{subfigure}

\usepackage{xcolor}
\newcount\Comments  
\newcommand{\kibitz}[2]{\ifnum\Comments=1\textcolor{#1}{#2}\fi}

\usepackage{listings}

\usepackage{tabularx}

\usepackage[capitalize,noabbrev]{cleveref}

\allowdisplaybreaks

\title{Expected Utility Regret Rule: Minimax and Bayes Optimal Portfolio Choice}

\author{Masahiro Kato\thanks{Email: \texttt{mkato-csecon@g.ecc.u-tokyo.ac.jp}}$\,$}

\affil{The University of Tokyo}

\date{\today}

\begin{document}

\maketitle 

\begin{abstract}
This study considers the problem of portfolio choice, where we recommend a portfolio to an investor to maximize the expected utility of their wealth. Our goal is to construct an asymptotically optimal portfolio choice rule in terms of expected utility regret, the difference between the expected utility of an oracle investor and that achieved by a portfolio chosen from data. We propose the Expected Utility Regret (EUR) rule, which jointly selects a portfolio class and estimates its weights. In a regular parametric return model, a single EUR rule attains both the minimax and the Bayes lower bounds, including their leading constants, without using the evaluation prior. We then derive the mean--variance and risk-parity portfolios as special cases of this framework. First, under smooth increasing and concave utility, the EUR rule and the sample mean--variance portfolio attain the same leading expected regret when expected excess returns approach zero sufficiently fast. Second, when the returns divided by their volatilities have a joint distribution that does not depend on the order of the assets, the EUR rule and the risk-parity portfolio coincide.
\end{abstract}


\section{Introduction}
\label{sec:introduction}
This study considers the portfolio choice problem of an investor who uses observed returns to decide both which assets to hold and how much to invest in them. Our goal is to establish a portfolio choice rule with an optimality guarantee. In the portfolio choice problem, our interest lies in the result of the portfolio choice. For this problem, existing approaches often start from the estimation of distributional information of asset returns, such as means and volatilities, and then construct portfolios such as mean--variance portfolios. However, such approaches do not consider the portfolio choice problem itself from the standpoint of decision-making, and do not by themselves establish which portfolio is appropriate. Some studies are exceptions: \citet{Kan2007optimalportfolio} derive the expected utility loss of mean--variance portfolios built from sample means and covariances of normally distributed returns, and \citet{Brandt2009parametricportfolio} model portfolio weights as a function of asset characteristics and maximize the investor's average utility over the sample period. These studies, however, compare rules within a given family, and they establish neither an optimal rule among all portfolio choice rules nor the conditions under which a familiar portfolio, such as the sample mean--variance portfolio, is optimal for an investor with general expected utility. In this study, we discuss the problem in a more direct manner, by considering a performance metric for the portfolio choice.

In the estimate-then-optimize approach, accurate estimates alone do not determine how good the resulting portfolio is, because the first step is judged in the units of the parameters rather than in utility, and the second step treats the estimates as the true parameters. Comparing portfolio choices therefore requires a criterion defined in terms of the chosen portfolio. Such a criterion applies equally to rules that plug in estimates, as the rule proposed in this study does within each class.

\subsection{Contributions}
Our contribution is to provide a general framework and a method for the portfolio choice problem. We first propose expected utility regret as the framework for evaluating a portfolio choice: the difference between the expected utility of an oracle who knows the return distribution, under the investor's investment restrictions, and that of the portfolio chosen from data. It extends the utility loss that \citet{Kan2007optimalportfolio} measure for mean--variance portfolios to general expected utility and to a choice among portfolio classes.

We then propose the Expected Utility Regret (EUR) rule, which attains both minimax and Bayes optimality in terms of expected utility regret. When multiple classes are available, it uses a preliminary sample to select the classes worth comparing and the remaining observations to choose a class and estimate its weights. In a regular parametric return model with a fixed number of assets and finitely many portfolio classes, its worst-case and prior-integrated expected regrets attain, including their leading constants, lower bounds over all measurable randomized portfolio choice rules. The same EUR rule is used for every evaluation prior with a continuous density that does not depend on the sample size $n$.

In the regret analysis, an exact decomposition relates expected utility regret to two estimation problems: estimating the weights within a portfolio class, a closed convex set of feasible portfolios, and comparing the optimized expected utilities of different classes.

The evaluation criterion determines which term carries the leading constant. A pair of classes whose optimized expected utilities differ by $g>0$ has a selection contribution of order at most $g\exp(-cng^2)$ for a constant $c>0$. When classes can tie, pairs separated by gaps of order $n^{-1/2}$ therefore carry the leading worst-case regret, while weight estimation contributes at order $n^{-1}$ and carries the worst-case constant only when no tie occurs. Under a prior whose continuous density is positive where two classes tie, the parameters with such gaps have prior probability of order $n^{-1/2}$, so both terms contribute at order $n^{-1}$. In either evaluation, the variance of the difference of the estimated optimized utilities of two classes determines how hard they are to separate, and movements common to the two portfolios cancel in this difference.

We derive familiar portfolios as special cases through two routes. The first is a limit: when expected excess returns range over a neighborhood of zero whose radius shrinks faster than $n^{-1/4}$ but more slowly than $n^{-1/2}$, and the return covariance remains nondegenerate, the EUR rule and the sample mean--variance portfolio attain the same minimax and Bayes constants, because the higher return moments affect expected utility below the order of the regret from estimating the mean. When expected excess returns are of order $n^{-1/4}$, a third-moment term displaces the optimal weights by as much as their estimation error, and the sample mean--variance portfolio pays an excess regret that the EUR rule avoids. The second route is an exact symmetry: when standardized returns and feasible standardized holdings are invariant under a group of asset permutations that can move any asset to any other, such as the cyclic permutations, concavity makes equal standardized holdings optimal, with equal contributions to volatility. This risk-parity composition is an exact representation of the EUR rule, so the minimax and Bayes conclusions apply to it, and when relative volatilities are unknown, the rule estimates the composition together with the amount invested. Once these conditions fail, imposing a simple composition near a symmetric model improves precision but leaves a loss quadratic in the departure along the omitted investment direction, and a binding money budget changes the composition when assets differ in relative volatility.

In practice, an investor who may hold only a few assets can lose more by choosing which assets to hold than by estimating their weights, when multiple asset sets can be equally good. Although an investor's utility function is often not known, in both routes the utility affects the optimal risky holdings only through the amount invested in risky assets, to first order in the expected excess returns in the first route and exactly in the second. Once the conditions of either route are established, the composition can therefore be recommended without knowing the investor's utility function.

\subsection{Related Literature}
Our analysis connects portfolio choice to the literature on the effects of estimating return parameters.
\paragraph{Estimation error in portfolio choice.}
Mean--variance theory specifies a population objective \citep{Markowitz1952portfolioselection}, and subsequent studies examine the utility lost when its inputs are estimated \citep{Klein1976theeffect,Jorion1986bayessteinestimation,Kan2007optimalportfolio,DeMiguel2009optimalversus,Tu2011markowitzmeets}. Bayesian and robust portfolio choice build prior beliefs or ambiguity into the investment decision \citep{Kandel1996onthe,Barberis2000investingfor,Pastor2000portfolioselection,Goldfarb2003robustportfolio,Garlappi2007portfolioselection}, whereas we use a prior only to evaluate a rule.

\paragraph{Statistical decision theory.}
Expected utility regret places portfolio choice in statistical decision theory, where a rule maps data to an action and is judged by a prespecified loss \citep{Wald1949statisticaldecision,Haavelmo1944theprobability}. Econometrics has developed minimax-regret treatment choice \citep{Manski2004statisticaltreatment,Manski2000identificationproblems,Manski2002treatmentchoice,Manski2010partialidentification,Dehejia2005programevaluation,Stoye2009minimaxregret,Tetenov2012statisticaltreatment,Kitagawa2018whoshould,Athey2021policylearning,Zhou2023offlinemultiaction,Manski2021econometricsfor}. Our portfolio classes play the role of actions in treatment choice, except that each class is a continuum of portfolios whose weights must be estimated.

\paragraph{Local asymptotics.}
Such rules are usually analyzed locally by replacing the experiment near a reference parameter with its Gaussian limit \citep{LeCam1972theoryofstatisics,LeCam1986asymptoticmethods,VanderVaart1991anasymptotic}, including treatment rules and sequential experiments \citep{Hirano2009asymptotics,Hirano2020asymptoticanalysis,Hirano2025asymptoticrepresentations,Adusumilli2023risk,Adusumilli2025samplestopsampling}. Our lower bounds are of this type, but localization alone does not determine performance over the original parameter set, where classes with fixed positive utility gaps require large-deviation control \citep{Glynn2004largedeviations,Imbens2025admissibilitycompletely}.

\paragraph{Best-arm identification.}
Best-arm identification chooses the action with the highest expected reward from a finite set \citep{Bechhofer1954asingle,Audibert2010bestarm,Bubeck2011pureexploration}, with information bounds for fixed confidence and fixed budget \citep{Lai1987adaptivetreatment,Kaufmann2014complexity,Kaufmann2016complexity,Garivier2016optimalbest,Komiyama2022minimaxoptimal,Degenne2023existence,Wang2024uniformlyoptimal}. Bayes, simple-regret and policy-choice formulations evaluate recommendations by regret \citep{Russo2020simplebayesian,Komiyama2023rateoptimal,Kasy2021adaptivetreatment,Ariu2021policychoice,Armstrong2022asymptoticefficiency,Adusumilli2022neymanallocation,Kato2024adaptiveexperimental}, and \citet{Kato2025minimaxbayesoptimalbestarm,Kato2025minimaxbayesadaptiveexperimental} obtain simultaneous minimax and Bayes optimality with exact constants. Portfolio choice differs in two respects: the investor does not allocate observations across actions, and the recommendation is a point in a continuum whose estimated weights contribute their own constant.

\paragraph{Moments and expected utility.}
The approximation of expected utility by return moments is developed by \citet{Samuelson1970fundamental}, and \citet{Markov2023portfolio} derive allocations with a skewness contribution for asymmetric Laplace returns. Section~\ref{sec:distribution} compares this approximation error with the statistical error and identifies the range of expected excess returns in which the sample mean--variance portfolio attains the smallest leading expected regret among all portfolio choice rules.

\paragraph{Risk parity.}
Earlier studies characterize equal volatility contributions and give robust or utility-based interpretations of risk parity \citep{Maillard2010properties,Fisher2015riskparity,Gava2022alpha}, and \citet{Noguer2026heuristic} shows that, for a mean--variance objective, exchangeable beliefs about returns make equal weighting the Bayes portfolio, and exchangeable beliefs about standardized returns with known volatilities make inverse-volatility weighting the Bayes portfolio, without estimation from a sample. We derive the risk-parity form from full expected utility, with exact estimation constants, and identify both the precision gained from investment restrictions \citep{Jagannathan2003riskreduction} and the utility lost.

\paragraph{Predictors and states.}
When portfolios depend on predictors, direct methods estimate allocations from the information at the investment date \citep{AitSahalia2001variableselection,Brandt2009parametricportfolio}, while high-dimensional methods control errors in estimated return inputs \citep{Fan2008highdimensional}.

\paragraph{Performance bounds from statistical learning.}
Performance bounds from statistical learning relate estimation error directly to portfolio utility for the empirical maximizer of an expected relative utility \citep{Rokhlin2020relativeutility} and for a regularized allocation rule using market predictors \citep{BazierMatte2018generalization}. Universal portfolios \citep{Cover1991universalportfolios,Luo2018efficientonlineportfoliowithlogarithmicregret} instead measure regret against the best rule in hindsight over an arbitrary return path. Beyond such bounds, we identify when the lower and upper bounds agree, including their constants, under worst-case and Bayes evaluation.

Section~\ref{sec:model} defines the decision problem, Section~\ref{sec:complexity} analyzes how estimation error produces regret, Section~\ref{sec:eurr} gives the EUR rule, and Sections~\ref{sec:minimax} and~\ref{sec:bayes} establish its minimax and Bayes optimality. Section~\ref{sec:existing-portfolios} relates the EUR rule to mean--variance and risk-parity portfolios, and Section~\ref{sec:dynamic} considers observed states. Section~\ref{sec:experiments} reports the experiments, Section~\ref{sec:conclusion} concludes, and the appendices collect the proofs.
\paragraph{Notation.}
For an integer $K\geq1$, $w\in\bbR^K$, and $p\in[1,\infty)$, we write
\begin{align}
  \norm{w}_p
  =
  \left(\sum_{a=1}^K\abs{w_a}^p\right)^{1/p},
  \qquad
  \norm{w}_\infty
  =
  \max_{1\leq a\leq K}\abs{w_a},
  \label{eq:lp-norm}
\end{align}
so that $\norm{w}_1$ is the sum of the absolute weights and $\norm{w}_2$ is the Euclidean norm. The support of $w$ is $\operatorname{supp}(w)=\{a\in\{1,\ldots,K\}:w_a\neq0\}$, and we write $\norm{w}_0=\abs{\operatorname{supp}(w)}$ for the number of nonzero weights. This $\ell_0$ quantity is not a norm, because multiplying $w$ by a nonzero scalar leaves it unchanged. Sections~\ref{sec:complexity} and~\ref{sec:dynamic} use a general norm on $\bbR^K$, which is introduced there together with its dual norm. For a matrix $A$, we write $A^\top$ for its transpose, $\operatorname{tr}(A)$ for its trace, and $\norm{A}$ for the operator norm induced by the Euclidean norm; for symmetric matrices $A$ and $B$ of the same size, $A\preceq B$ means that $B-A$ is positive semidefinite. The $q\times q$ identity matrix is $I_q$. We use $\abs{\cdot}$ for the absolute value of a real number and for the number of elements of a finite set, and $\lfloor x\rfloor$ for the largest integer at most $x$. In bounds, unsubscripted $c$ and $C$ denote positive constants that do not depend on $n$ or on the return distribution and may change between displays. Expectation and probability under a distribution $P$ are $\E_P$ and $\Pr_P$, and $N(\mu,\Sigma)$ is the Gaussian distribution with mean $\mu$ and covariance matrix $\Sigma$. We write $X_n\xrightarrow{p}X$ and $X_n\xrightarrow{d}X$ for convergence in probability and in distribution as $n\to\infty$, with $\rightsquigarrow$ also denoting convergence in distribution. We write $x_n\asymp y_n$ when $x_n/y_n$ is bounded away from zero and from infinity, and $x_n\lesssim y_n$ when $x_n\leq Cy_n$ for a constant independent of $n$. Derivatives of the utility are $U'$ and $U''$, and $\nabla$ denotes the gradient with respect to the argument of the function being differentiated; a subscript specifies that argument when needed.

\section{Setup}
\label{sec:model}
The investor observes returns and chooses a portfolio for a new investment period. We specify the returns, the portfolio choice problem and its loss, and the two ways of evaluating a rule under uncertainty about the return distribution, and then represent the feasible portfolios as a union of portfolio classes.

\subsection{Return}
Let $K$ be the number of assets and $[K]=\{1,\ldots,K\}$ the set of assets. In this study, we treat the problem in general terms without restricting the assets to a specific type, such as stocks, bonds, or gold.

Let $R\in\bbR^K$ be a one-period vector of asset returns. In the static problem, the observations
\begin{align}
  R_1,\ldots,R_n\sim P
  \label{eq:iid}
\end{align}
are independent with a common unknown distribution $P$ in a set $\calP_K$ of return distributions. The new return vector used to evaluate a chosen portfolio is independent of the training observations and has the same distribution. Section~\ref{sec:dynamic} replaces this independent-observation model with a stationary time series and an observed state.

Unless a result states otherwise, the number of assets $K$ and the set of return distributions $\calP_K$ do not change with $n$. A statement that is uniform over $P\in\calP_K$ uses a bound whose constants and remainder do not depend on $P$. Although $\calP_K$ is fixed, the return distribution at which a supremum over $\calP_K$ is evaluated may differ with $n$.

\subsection{Portfolio Choice Problem}
Let $\calW_K\subset\bbR^K$ be a nonempty compact set of feasible portfolios $w=(w_a)_{a\in[K]}$, where $w_a$ is the investment in asset $a\in[K]$. Let $x_0$ be the initial wealth and let $U$ be a utility function that is continuous, increasing, and concave on an interval containing all attainable values of $x_0+w^\top R$.

Under an unknown distribution $P$, the investor evaluates a portfolio by
\begin{align}
  V_P(w)
  =
  \E_P\left[U\left(x_0+w^\top R\right)\right].
  \label{eq:value}
\end{align}
The largest expected utility and the set of portfolios attaining it are
\begin{align}
  V_P^\star
  =
  \sup_{w\in\calW_K}V_P(w),
  \qquad
  \calA(P)
  =
  \argmax_{w\in\calW_K}V_P(w).
  \label{eq:oracle-set}
\end{align}
We call a member of $\calA(P)$ an oracle portfolio because choosing it uses knowledge of $P$.

A portfolio choice rule $\delta_n$ is a measurable probability kernel from $(\bbR^K)^n$ to $\calW_K$, so its output $\widehat w_n$ may be deterministic or randomized conditional on the observations. Let $\mathfrak D_{n,K}$ be the set of these rules. When an optimization has multiple solutions, we use a specified measurable selection; for the finite-dimensional compact optimizations in the EUR rule, this is the lexicographically smallest solution.

\subsection{Expected Utility Regret}
The utility lost by using a feasible portfolio $w$ instead of an oracle portfolio is
\begin{align}
  r_P(w)
  =
  V_P^\star-V_P(w).
  \label{eq:regret}
\end{align}
The expected regret of $\delta_n$ is obtained by also averaging over its training sample and any randomization:
\begin{align}
  \calR_n(\delta_n,P)
  =
  \E_P\left[r_P(\widehat w_n)\right].
  \label{eq:expected-regret}
\end{align}
The expectation defining $V_P$ evaluates a portfolio on a new return, whereas the expectation in \eqref{eq:expected-regret} evaluates the error in choosing that portfolio from data. We keep the cardinal normalization of $U$ fixed: adding a constant to $U$ does not affect regret, and multiplying $U$ by a positive constant multiplies regret and all its leading constants by the same factor.

\subsection{Uncertainty}
The model $\calP_K$ describes the distributions that the investor cannot distinguish before observing data, so we evaluate a rule either by its largest expected regret over that model or by its expected regret integrated under a specified prior.

\paragraph{Minimax optimality.}
The smallest achievable worst-case expected regret at sample size $n$ is
\begin{align}
  \mathcal V_{n,K}^{\mathrm{mm}}
  =
  \inf_{\delta_n\in\mathfrak D_{n,K}}
  \sup_{P\in\calP_K}
  \calR_n(\delta_n,P).
  \label{eq:minimax}
\end{align}
For a positive deterministic normalization $a_n$, asymptotic minimax optimality with constant $C$ means that both $a_n\sup_{P\in\calP_K}\calR_n(\delta_n,P)$ and $a_n\mathcal V_{n,K}^{\mathrm{mm}}$ converge to $C$ as $n\to\infty$. The supremum is taken anew at each $n$, so this criterion includes return distributions approaching a tie between optimal classes.

\paragraph{Bayes optimality.}
For Bayes evaluation, let $\calP_K=\{P_\theta:\theta\in\Theta\}$, where the parameter set $\Theta\subset\bbR^p$ is compact and $p$ is the number of return parameters. Let $\mathfrak H$ be the set of probability measures on $\Theta$ with continuous densities relative to Lebesgue measure. For $H\in\mathfrak H$, define
\begin{align}
  \mathcal V_{n,K}^{H}
  =
  \inf_{\delta_n\in\mathfrak D_{n,K}}
  \int_\Theta
  \calR_n(\delta_n,P_\theta)
  \rmd H(\theta).
  \label{eq:bayes-value}
\end{align}
The infimum allows a competing rule to use $H$, but the EUR rule uses the same construction for every $H$: with $\mathfrak D_K=\prod_{n\geq1}\mathfrak D_{n,K}$ denoting rules specified at all sample sizes, the order of quantifiers is
\begin{align}
  \exists(\delta_n^\star)\in\mathfrak D_K
  \quad
  \forall H\in\mathfrak H.
  \label{eq:bayes-quantifier-setup}
\end{align}
For each $H\in\mathfrak H$ that does not depend on $n$, asymptotic Bayes optimality means that the integrated expected regret of this rule and the infimum in \eqref{eq:bayes-value}, both multiplied by the same positive deterministic normalization $a_n$, have the same limit as $n\to\infty$.

\subsection{Portfolio Class}
In this study, we represent the feasible set $\calW_K$ as a union of sets of portfolios, each of a simple form. Investment restrictions used in practice, such as a limit on the number of assets held, often make $\calW_K$ nonconvex, but the analysis of estimated portfolio weights relies on convexity. When a smooth, strictly concave expected utility is maximized over a convex set at a portfolio in its relative interior, the optimal portfolio changes continuously with the return distribution, a first-order condition characterizes it, and a small error in the weights costs expected utility of the order of the squared error. Writing $\calW_K$ as a union of convex sets keeps these properties within each set, so that the nonconvexity affects only the choice among the sets. This representation therefore lets us treat complicated feasible sets, in particular nonconvex ones, without losing generality.

To allow a choice among supports or other investment restrictions, we write the feasible set as
\begin{align}
  \calW_K
  =
  \bigcup_{m\in\mathfrak M_K}\calW_{K,m},
  \label{eq:model-decomposition}
\end{align}
where each $\calW_{K,m}$ is nonempty, closed, and convex. We call each $\calW_{K,m}$ a portfolio class, or simply a class, and $\mathfrak M_K$ the set of class indices. A class consists of the portfolios that satisfy one investment restriction, such as holding only a given set of assets. The decomposition separates two decisions: which class to invest in, a choice of the index $m$, and which weights to hold inside that class, an optimization over the convex set $\calW_{K,m}$. Each class is convex, but the union $\calW_K$ need not be. We assume that $\mathfrak M_K$ is a separable metric space, that the graph of $m\mapsto\calW_{K,m}$ is measurable, and that the optimizations below admit measurable selections. We keep this representation fixed when decomposing regret.

\paragraph{The role of convex portfolio classes.}
The next example shows what fails when a nonconvex feasible set is handled as a single set rather than as a union of convex classes.

\begin{example}[A nonconvex feasible set treated as one set]
Let $K=2$ and let $e_a$ be the portfolio that invests all wealth in asset $a$. An investor who must hold exactly one of the two assets chooses from $\calW_2=\{e_1,e_2\}$, which is compact but not convex. Write $g_P=V_P(e_1)-V_P(e_2)$. Then, we have $\calA(P)=\{e_1\}$ if $g_P>0$, $\calA(P)=\{e_2\}$ if $g_P<0$, and $\calA(P)=\calW_2$ if $g_P=0$, and a portfolio outside $\calA(P)$ has regret $\abs{g_P}$. Three properties of a concave objective on a convex set, on which the analysis of estimated weights rests, fail here.
\begin{itemize}
\item The oracle portfolio is not continuous in $P$: it jumps between $e_1$ and $e_2$ as $g_P$ changes sign, however smoothly $P$ varies.
\item The first-order condition no longer characterizes the oracle portfolio. On a convex set, a maximizer $v$ of a differentiable concave $V_P$ satisfies $\langle\nabla V_P(v),w-v\rangle\leq0$ for every feasible $w$, and this condition is also sufficient. Here, the segment from $e_1$ to $e_2$ leaves $\calW_2$, so $e_1$ can be optimal while $\langle\nabla V_P(e_1),e_2-e_1\rangle>0$, and the comparison is decided by the sign of $g_P$ rather than by derivatives at the oracle portfolio.
\item The loss is no longer quadratic in the error. On a convex set with smooth $V_P$ and a maximizer $v$ in its relative interior, the derivative of $V_P$ at $v$ is zero in every direction that stays in the set. Hence, for $w$ in the set near $v$, $V_P(v)-V_P(w)$ is at most a constant times $\norm{w-v}_2^2$, and weights with estimation error of order $n^{-1/2}$ lose expected utility of only order $n^{-1}$. Here, by contrast, holding the wrong portfolio loses the whole difference $\abs{g_P}$. If the portfolio with the larger estimated expected utility is held and $\sqrt n$ times the error in the estimate of $g_P$ has a nondegenerate normal limit, then for $\abs{g_P}$ of order $n^{-1/2}$ the wrong portfolio is held with probability bounded away from zero, so the expected regret is of order $n^{-1/2}$ (Section~\ref{sec:complexity}).
\end{itemize}
Writing $\calW_2=\{e_1\}\cup\{e_2\}$ as a union of two singleton classes isolates these failures: nothing is estimated within a class, and the whole difficulty is the comparison of the two numbers $V_P(e_1)$ and $V_P(e_2)$.
\end{example}

In general, under the decomposition \eqref{eq:model-decomposition}, the nonconvexity of $\calW_K$ affects only the choice of class. Within a class, expected utility is concave on a closed convex set, so the first-order condition characterizes its maximizers, and under the regularity conditions of Section~\ref{sec:complexity} the expected utility lost by estimated weights is of the order of the squared estimation error. The choice of class, in contrast, is made by comparing the highest expected utilities attainable in the classes. This choice is easy when one class is clearly better, because a class whose best portfolio has a clearly lower estimated expected utility can be discarded, as the EUR rule does in its preliminary comparison (Section~\ref{sec:eurr}). It is difficult when the best portfolios of two classes attain nearly the same expected utility, and Sections~\ref{sec:complexity} and~\ref{sec:minimax} show that, when classes can tie, the choice of class is the dominant source of worst-case regret, because a wrong class loses the whole utility gap rather than an amount of the order of a squared error.

\paragraph{Nonconvex feasible sets defined as unions of classes.}
Many investment restrictions used in practice define nonconvex feasible sets of the form \eqref{eq:model-decomposition}.

\begin{example}[Nonconvex portfolios]
Let $\Delta_K$ be the simplex of fully invested portfolios without short positions in \eqref{eq:simplex}. The following feasible sets are not convex in general, and each is a union of nonempty closed convex classes.
\begin{itemize}
\item \emph{A limit on the number of assets held.} Every position has to be monitored and traded, and portfolio selection under such a cardinality constraint is a standard problem \citep{Chang2000heuristics,Jobst2001computational,Bertsimas2009algorithm}; portfolios with few active positions are also advocated for their stability \citep{Brodie2009sparse}. For $s\in\{1,\ldots,K-1\}$,
$\{w\in\Delta_K:\norm{w}_0\leq s\}=\bigcup_{m\subseteq[K],\,1\leq\abs{m}\leq s}\{w\in\Delta_K:w_a=0\text{ for }a\notin m\}$,
and each class holds only the assets in $m$.
\item \emph{A minimum position size.} Each asset is either not held or held with a weight of at least $\ell\in(0,1]$ \citep{Jobst2001computational}:
$\{w\in\Delta_K:w_a=0\text{ or }w_a\geq\ell\text{ for every }a\in[K]\}=\bigcup_{m}\{w\in\Delta_K:w_a=0\text{ for }a\notin m,\ w_a\geq\ell\text{ for }a\in m\}$,
where $m$ ranges over the nonempty sets of assets with $\ell\abs{m}\leq1$.
\item \emph{The diversification rule for European investment funds (UCITS).} The securities of one issuer may make up at most $5\%$ of the fund's assets, a limit that may be raised to $10\%$ provided that the holdings above $5\%$ total at most $40\%$ \citep[Article~52]{EU2009ucits}. With one issuer per asset,
$\{w\in\Delta_K:w_a\leq0.10\text{ for every }a\in[K],\ \sum_{a:w_a>0.05}w_a\leq0.40\}=\bigcup_{S}\{w\in\Delta_K:w_a\leq0.05\text{ for }a\notin S,\ w_a\leq0.10\text{ for }a\in S,\ \sum_{a\in S}w_a\leq0.40\}$,
where $S\subseteq[K]$ ranges over the sets of issuers allowed above $5\%$ for which the class is nonempty.
\item \emph{Round lots.} Holdings must be multiples of a lot $1/L$ of wealth for an integer $L\geq1$ \citep{Jobst2001computational}:
$\{w\in\Delta_K:Lw_a\in\mathbb Z\text{ for every }a\in[K]\}=\bigcup_{k}\{k/L\}$,
where $k$ ranges over the vectors of nonnegative integers with $\sum_{a=1}^Kk_a=L$. The classes are singletons.
\end{itemize}
Each equality follows by taking $m=\operatorname{supp}(w)$ in the first two cases, $S=\{a:w_a>0.05\}$ in the third, and $k=Lw$ in the fourth. In each case, the investor decides which class to invest in as well as the weights within it.
\end{example}

When the feasible set is one convex set, as for the mean--variance and risk-parity portfolios in Section~\ref{sec:existing-portfolios}, there is one class and no comparison occurs. The next two paragraphs give the forms of $\calW_{K,m}$ used throughout this study.

\paragraph{Portfolios without short positions.}
The portfolios that invest all wealth and hold no short position form the simplex
\begin{align}
  \Delta_K
  &:={}
  \left\{w\in\bbR^K:w_a\geq0,\ \sum_{a=1}^Kw_a=1\right\}.
  \label{eq:simplex}
\end{align}
A limit on the number of assets held restricts the simplex further to
\begin{align}
  \calW_{K,s}^{\ell_0}
  &:={}
  \left\{w\in\Delta_K:\norm{w}_0\leq s\right\},
  \label{eq:sparse}
\end{align}
where $s\in\{1,\ldots,K\}$ is that limit and $\norm{w}_0$ counts the nonzero weights of $w$. For $s<K$, this set is not convex, since the average of two portfolios that hold different assets can hold more than $s$ assets. As in the first case of the example of nonconvex portfolios, its classes are indexed by the sets of assets $m\subseteq[K]$ with $1\leq\abs{m}\leq s$, and $\calW_{K,m}=\{w\in\Delta_K:\operatorname{supp}(w)\subseteq m\}$ is the simplex of portfolios that hold only the assets in $m$. Choosing a class then means choosing which assets to hold.

\paragraph{Portfolios with short positions.}
Allowing short positions instead enlarges the simplex, and a bound on the sum of the absolute weights keeps the enlarged set compact:
\begin{align}
  \calW_K^{\ell_1}(L)
  &:={}
  \left\{w\in\bbR^K:\sum_{a=1}^Kw_a=1,\ \norm{w}_1\leq L\right\},
  \label{eq:gross}
\end{align}
where $L\geq1$. The quantity $\norm{w}_1$ is the gross exposure of $w$, the total size of its long and short positions, and $L=1$ recovers $\Delta_K$. Both $\Delta_K$ and $\calW_K^{\ell_1}(L)$ are convex, so each forms a single class, and $\mathfrak M_K$ then has one index.

\section{Analyses of Estimation Error and Expected Utility Regret}
\label{sec:complexity}
\label{sec:upper}
We show how the analysis of the estimation problem relates to that of the decision problem. Because the feasible set is the union of the portfolio classes $\calW_{K,m}$ (Section~\ref{sec:model}), the investor must both estimate the weights within a class and compare classes.

\subsection{Return Model}
We impose two assumptions on the return model: integrability of utility and a regular parametric model.

\paragraph{Integrability of utility.}
To compare sample and population expected utilities and their gradients uniformly over portfolios and distributions, the next assumption bounds expected utility and, where needed, its derivatives.

\begin{assumption}[Utility derivatives and integrability]
\label{ass:utility}
The function $U$ is continuous, increasing, and concave, and the following integrability bound holds:
\begin{align}
  \sup_{P\in\calP_K}
  \E_P\left[
    \sup_{w\in\calW_K}
    \abs{U\left(x_0+w^\top R\right)}
  \right]
  <
  \infty.
  \label{eq:portfolio-utility-integrability}
\end{align}
For results involving the gradient, $U$ is continuously differentiable and
\begin{align}
  \sup_{P\in\calP_K}
  \E_P\left[
    \sup_{w\in\calW_K}
    \abs{U'\left(x_0+w^\top R\right)}\norm{R}_2
  \right]
  <\infty.
  \label{eq:first-derivative-integrability}
\end{align}
For results involving derivatives through order $k_U\geq2$, $U$ is $k_U$ times continuously differentiable and, for every $k=0,\ldots,k_U$,
\begin{align}
  \sup_{P\in\calP_K}
  \E_P\left[
    \sup_{w\in\calW_K}
    \abs{U^{(k)}\left(x_0+w^\top R\right)}\norm{R}_2^k
  \right]
  <\infty,
  \label{eq:higher-derivative-integrability}
\end{align}
where $U^{(0)}=U$ and the factor $\norm{R}_2^0$ equals one.
\end{assumption}

\paragraph{A regular parametric model.}
The conditions in this paragraph are used from Section~\ref{sec:eurr} onward. For the parametric model $\{P_\theta:\theta\in\Theta\}$, define
\begin{align}
  B_m(\theta)
  =
  \sup_{w\in\calW_{K,m}}V_{P_\theta}(w).
  \label{eq:component-value-theta}
\end{align}
Let $w_m(\theta)\in\argmax_{w\in\calW_{K,m}}V_{P_\theta}(w)$ attain $B_m(\theta)$, and let $A(\theta)=\argmax_{m\in\mathfrak M_K}B_m(\theta)$ be the set of optimal classes. Let $p_\theta$ denote the density of $P_\theta$, $S_\theta(R)=\nabla\log p_\theta(R)$ its score, and $I_\theta=\E_\theta[S_\theta(R)S_\theta(R)^\top]$ the Fisher information, and let $\Theta^+$ be a compact estimation set containing $\Theta$ in its interior. For a sample $D=\{R_i:i\in I_D\}$ indexed by a nonempty set $I_D\subseteq\{1,\ldots,n\}$, the maximum likelihood estimator of $\theta$ is
\begin{align}
  \widehat\theta(D)
  \in
  \argmax_{\theta\in\Theta^+}
  \sum_{i\in I_D}\log p_\theta(R_i).
  \label{eq:mle}
\end{align}
Either assumption below ensures that this maximum is attained: Assumption~\ref{ass:regular-return-model} makes the log-likelihood continuous on the compact set $\Theta^+$, and Assumption~\ref{ass:regular-one-class} requires it directly. When multiple maximizers exist, $\widehat\theta(D)$ denotes the lexicographically smallest one, a measurable function of the sample (Section~\ref{sec:model}), and $\widehat\theta_n=\widehat\theta(\{R_1,\ldots,R_n\})$ uses all $n$ observations. The exact constants in Sections~\ref{sec:minimax} and~\ref{sec:bayes} depend on how the error of this estimator propagates to $B_m$ and $w_m$, and the next assumption makes this dependence smooth.

\begin{assumption}[Regular return model]
\label{ass:regular-return-model}
The returns are independent with a common return distribution $P_\theta$, which has a density $p_\theta$ relative to a known finite measure $\nu$ on a bounded subset of $\bbR^K$. The evaluation set $\Theta$ is the closure of a bounded open set in $\bbR^p$ with piecewise $C^2$ boundary and lies strictly inside a compact estimation set $\Theta^+$. The following conditions hold.
\begin{enumerate}[label=(\roman*),leftmargin=2.1em]
  \item On a neighborhood of $\Theta^+$, the densities are bounded above and away from zero, uniformly in the return vector. Their parameter derivatives through order four exist, are continuous in the parameter, and are uniformly bounded. The model is identifiable on $\Theta^+$, and $I_\theta$ is uniformly positive definite on $\Theta^+$.
  \item There are finitely many compact convex portfolio classes. Utility is increasing, concave, and $C^4$ on a neighborhood of the attainable wealth interval. In each positive-dimensional class, the negative tangent Hessian of expected utility is uniformly positive definite throughout the class and $\Theta^+$, and the class optimizer $w_m(\theta)$ remains a positive distance from the relative boundary. Singleton classes are also allowed.
  \item Whenever $\theta\in\Theta$ satisfies $\abs{A(\theta)}\geq2$, the optimizers indexed by $A(\theta)$ are distinct. If $m_0$ is the smallest index in $A(\theta)$, the matrix whose rows are $\nabla(B_m-B_{m_0})(\theta)^\top$, for $m\in A(\theta)\setminus\{m_0\}$, has full row rank.
  \item Whenever $\theta\in\partial\Theta$ satisfies $\abs{A(\theta)}\geq2$, the boundary $\partial\Theta$ is smooth near $\theta$, and the matrix obtained by appending the normal of $\partial\Theta$ at $\theta$ to the rows in part (iii) also has full row rank.
\end{enumerate}
\end{assumption}

\begin{remark}[Scope of the regular return model]
Taking $\nu$ to be counting measure gives finite-support models, and continuous models are covered when they have a common bounded support and a smooth positive density. For $\calW_{K,s}^{\ell_0}$ with $K$ and $s$ fixed, part (ii) requires the optimizer over each support simplex to hold every asset in that support, so a class cannot tie with a class whose support contains its own.
\end{remark}

\paragraph{One class with unbounded returns.}
With a single convex portfolio class, no comparison between classes occurs. The next assumption states the properties of the maximum likelihood estimator used by the one-class results in Sections~\ref{sec:minimax} and~\ref{sec:bayes}, and allows unbounded returns and unknown relative asset volatilities.

\begin{assumption}[Regular maximum likelihood estimation in one class]
\label{ass:regular-one-class}
There is one full-dimensional compact convex feasible set $\calW_K$, independent of $\theta$ and $n$. The evaluation and estimation sets $\Theta\subset\operatorname{int}(\Theta^+)$ have the properties stated in Assumption~\ref{ass:regular-return-model}. The observations are independent and have density $p_\theta$ relative to a sigma-finite measure. The following conditions hold.
\begin{enumerate}[label=(\roman*),leftmargin=2.1em]
\item The model is identifiable and differentiable in quadratic mean on a neighborhood of $\Theta^+$. Its score $S_\theta$ has mean zero, continuous positive definite Fisher information $I_\theta$, and fourth moments bounded over that neighborhood. Differentiation under the integral is valid for bounded functions of a sample. The likelihood attains its maximum on $\Theta^+$, and its lexicographically selected maximizer is measurable.
\item The map $(\theta,w)\mapsto V_{P_\theta}(w)$ has continuous mixed derivatives through order three on a neighborhood of $\Theta^+\times\calW_K$. Its negative Hessian in $w$ is bounded above and below by positive multiples of $I_K$. The unique optimizer $w_1(\theta)$ stays a positive distance from the boundary of $\calW_K$.
\item The estimator $\widehat\theta_n$ in \eqref{eq:mle} has a remainder $\eta_n(\theta)$ such that
\begin{align}
 \widehat\theta_n-\theta
 &=I_\theta^{-1}\frac1n\sum_{i=1}^n S_\theta(R_i)+\eta_n(\theta),
 &\sup_{\theta\in\Theta}n\E_\theta[\norm{\eta_n(\theta)}_2^2]&\longrightarrow0,
 \label{eq:regular-estimator-ltwo}\\
 \sup_{n\geq1}\sup_{\theta\in\Theta}
 n^2\E_\theta[\norm{\widehat\theta_n-\theta}_2^4]&<\infty.
 \label{eq:regular-estimator-fourth}
\end{align}
\end{enumerate}
All sets, the utility, and both dimensions $K$ and $p$ are independent of $n$.
\end{assumption}

Appendix~\ref{app:regular-one-class} derives the expected-regret constants from these conditions, which hold for the bounded model in Assumption~\ref{ass:regular-return-model} with one full-dimensional class (Appendix~\ref{app:likelihood-global}) and for a non-Gaussian model with unbounded returns and estimated asset scales (Section~\ref{sec:estimated-volatility}).

\subsection{Regret Decomposition}
The investor loses utility by estimating weights within the chosen class, and selecting a suboptimal class causes an additional loss even with perfectly estimated weights.

For $m\in\mathfrak M_K$, define
\begin{align}
  V_{P,m}^\star
  =
  \sup_{w\in\calW_{K,m}}V_P(w),
  \qquad
  \calA_m(P)
  =
  \argmax_{w\in\calW_{K,m}}V_P(w).
  \label{eq:component-values}
\end{align}
Then, it holds that
\begin{align}
  V_P^\star
  =
  \sup_{m\in\mathfrak M_K}V_{P,m}^\star.
  \label{eq:component-supremum}
\end{align}
If $\widehat m$ is the selected class and $\widehat w\in\calW_{K,\widehat m}$, then the following identity holds:
\begin{align}
  r_P(\widehat w)
  =
  \underbrace{V_P^\star-V_{P,\widehat m}^\star}_{\text{regret from class selection}}
  +
  \underbrace{V_{P,\widehat m}^\star-V_P(\widehat w)}_{\text{regret from allocation within a class}}.
  \label{eq:two-layer-decomposition}
\end{align}
The two components in \eqref{eq:two-layer-decomposition}, but not their sum, depend on the collection $(\calW_{K,m})_{m\in\mathfrak M_K}$, which we therefore fix.

\subsection{Estimation Error and Regret}
We use derivatives to connect a portfolio's estimation error to its utility loss.

\paragraph{Estimating weights within a portfolio class.}
Fix a portfolio class $\calW_{K,m}$. Because the investor does not know $V_P$, the investor maximizes its sample counterpart,
\begin{align}
  \widehat V_n(w)
  =
  \frac1n\sum_{i=1}^nU\left(x_0+w^\top R_i\right),
  \label{eq:empirical-value}
\end{align}
the average utility of $w$ over the observed returns. Holding its maximizer $\widehat w_m$ in \eqref{eq:component-erm} below, the investor loses $V_{P,m}^\star-V_P(\widehat w_m)$, and we bound this loss in two steps. First, $V_P$ does not increase along the segment from an optimal portfolio $v$ to any portfolio $w$ of the class. If $\widehat V_n$ changes at nearly the same rate along every such segment, it cannot increase much either, and the utility-loss condition below turns this into a bound on the distance of $\widehat w_m$ from the optimal portfolios. Here, the rate is the derivative in the direction $w-v$, and $Z_{n,m}(P)$ below is the largest difference between these rates for $\widehat V_n$ and $V_P$. Second, the loss is at most this difference times the distance moved, and Proposition~\ref{prop:within-class} combines the two steps.

For each $m\in\mathfrak M_K$, let $\norm{\cdot}_m$ be a norm on $\bbR^K$ with dual norm $\norm{\cdot}_{m,*}$. When $V_P$ is differentiable on a neighborhood of $\calW_{K,m}$, define
\begin{align}
  Z_{n,m}(P)
  =
  \sup_{\substack{v\in\calA_m(P),\ w\in\calW_{K,m}\\w\neq v}}
  \left|
    \int_0^1
    \left\langle
      \nabla\widehat V_n(v+t(w-v))-\nabla V_P(v+t(w-v)),
      \frac{w-v}{\norm{w-v}_m}
    \right\rangle
    \rmd t
  \right|.
  \label{eq:directional-estimation-error}
\end{align}
For a singleton class, set $Z_{n,m}(P)=0$. A possibly looser bound is
\begin{align}
  Z_{n,m}(P)
  \leq
  \sup_{u\in\calW_{K,m}}
  \norm{\nabla\widehat V_n(u)-\nabla V_P(u)}_{m,*}.
  \label{eq:global-gradient-sufficient-bound}
\end{align}
Both steps require $Z_{n,m}(P)$ to be well defined and to have moments.

\begin{assumption}[Accuracy of the sample-average gradient]
\label{ass:gradient-process}
For every $m\in\mathfrak M_K$, the map $w\mapsto V_P(w)$ is differentiable on a neighborhood of $\calW_{K,m}$, and $Z_{n,m}(P)$ is measurable and has the moments used in each result.
\end{assumption}

The first step also depends on how quickly expected utility falls away from its maximizing set. For a nonempty set $A\subset\calW_{K,m}$, write
\begin{align}
  d_m(w,A)
  =
  \inf_{v\in A}\norm{w-v}_m.
  \label{eq:model-distance}
\end{align}

\begin{assumption}[Utility loss away from an optimal portfolio]
\label{ass:utility-loss}
There are constants $q>1$ and $c_0>0$ such that, for every $P\in\calP_K$, every $m\in\mathfrak M_K$, and every $w\in\calW_{K,m}$,
\begin{align}
  V_{P,m}^\star-V_P(w)
  \geq
  c_0 d_m\left(w,\calA_m(P)\right)^q.
  \label{eq:utility-loss-condition}
\end{align}
\end{assumption}

With $q=2$, the loss grows quadratically in the distance, and a larger $q$ allows expected utility to be flatter near its maximizing set. Consider a maximizer of the sample average utility in class $m$,
\begin{align}
  \widehat w_m
  \in
  \argmax_{w\in\calW_{K,m}}\widehat V_n(w).
  \label{eq:component-erm}
\end{align}

\begin{proposition}[Within-class regret bound]
\label{prop:within-class}
Suppose that Assumptions \ref{ass:gradient-process} and \ref{ass:utility-loss} hold. Then, for every $P\in\calP_K$ and $m\in\mathfrak M_K$, the portfolio $\widehat w_m$ in \eqref{eq:component-erm} satisfies
\begin{align}
  d_m\left(\widehat w_m,\calA_m(P)\right)
  &\leq
  \left(\frac{Z_{n,m}(P)}{c_0}\right)^{1/(q-1)},
  \label{eq:within-class-distance}
  \\
  V_{P,m}^\star-V_P(\widehat w_m)
  &\leq
  c_0^{-1/(q-1)}Z_{n,m}(P)^{q/(q-1)}.
  \label{eq:within-class-bound}
\end{align}
\end{proposition}

Taking expectations in \eqref{eq:within-class-bound} shows that if $\E_P[Z_{n,m}(P)^{q/(q-1)}]=O(n^{-q/(2(q-1))})$ holds, as for a sample average, then the expected allocation regret is $O(n^{-q/(2(q-1))})$, which is $O(n^{-1})$ when $q=2$.

\paragraph{Comparing fitted portfolios.}
To choose a class, the investor fits a portfolio in every class on one part of the data, $\mathcal D_1$, and compares the fitted portfolios by their average utility on another part, as in \eqref{eq:split-rule} below. Since the selection regret depends on how accurately the fitted portfolios can be compared, we measure complexity through utility differences. Let $(\widetilde w_m)_{m\in\mathfrak M_K}$ be these fitted portfolios, defined in \eqref{eq:training-component-rule}. Conditional on $\mathcal D_1$, define the random semimetric
\begin{align}
  d_{P,\mathcal D_1}(m,m')
  =
  \left\|
    U\left(x_0+\widetilde w_m^\top R\right)
    -
    U\left(x_0+\widetilde w_{m'}^\top R\right)
  \right\|_{L^2(P)},
  \label{eq:model-metric}
\end{align}
where $R\sim P$ is independent of $\mathcal D_1$ and $\norm{f}_{L^2(P)}=(\E_P[f(R)^2])^{1/2}$. Because adding the same quantity to the average utility of every class does not change their ranking, we subtract the utility of one fitted portfolio. Given an independent sample $R_1',\ldots,R_{n_2}'$, define
\begin{align}
  \mathfrak C_{\mathrm{sel},n_2}(P\mid\mathcal D_1)
  =
  \E_{R',\varepsilon}
  \left[
    \sup_{m\in\mathfrak M_K}
    \left|
      \frac1{n_2}\sum_{i=1}^{n_2}\varepsilon_i
      \left(
      \ell_{\widetilde w_m}(R_i')
      -
      \ell_{\widetilde w_{m_0}}(R_i')
      \right)
    \right|
    \mathrel{\Big|}\mathcal D_1
  \right],
  \label{eq:selection-complexity}
\end{align}
where $\ell_w(r)=U(x_0+w^\top r)$, $(\varepsilon_i)$ are independent Rademacher signs, and $m_0$ is a fixed reference class.

Suppose that, conditional on $\mathcal D_1$, the process indexed by $m\in\mathfrak M_K$ with coordinates $n_2^{-1/2}\sum_{i=1}^{n_2}\varepsilon_i(\ell_{\widetilde w_m}(R_i')-\ell_{\widetilde w_{m_0}}(R_i'))$ has sub-Gaussian increments with variance proxy at most $C d_{P,\mathcal D_1}(m,m')^2$, with probability taken over both the evaluation sample and the Rademacher signs. Write $\calN(\varepsilon,\mathfrak M_K,d_{P,\mathcal D_1})$ and $\operatorname{diam}(\mathfrak M_K,d_{P,\mathcal D_1})$ for the $\varepsilon$-covering number and the diameter of the set of class indices under this semimetric. Dudley's inequality then gives
\begin{align}
  \mathfrak C_{\mathrm{sel},n_2}(P\mid\mathcal D_1)
  \leq
  \frac{C}{\sqrt{n_2}}
  \int_0^{\operatorname{diam}(\mathfrak M_K,d_{P,\mathcal D_1})}
  \sqrt{\log\calN(\varepsilon,\mathfrak M_K,d_{P,\mathcal D_1})}
  \rmd\varepsilon.
  \label{eq:dudley}
\end{align}

\paragraph{The split-sample rule and its regret bound.}
Split the observations into independent samples $\mathcal D_1$ and $\mathcal D_2=\{R_1',\ldots,R_{n_2}'\}$ with sizes $n_1$ and $n_2$. On $\mathcal D_1$, compute
\begin{align}
  \widetilde w_m
  \in
  \argmax_{w\in\calW_{K,m}}
  \frac1{n_1}\sum_{R_i\in\mathcal D_1}
  U\left(x_0+w^\top R_i\right)
  \label{eq:training-component-rule}
\end{align}
for every $m\in\mathfrak M_K$. On $\mathcal D_2$, choose
\begin{align}
  \widehat m
  \in
  \argmax_{m\in\mathfrak M_K}
  \frac1{n_2}\sum_{i=1}^{n_2}
  U\left(x_0+\widetilde w_m^\top R_i'\right),
  \qquad
  \widehat w_n=\widetilde w_{\widehat m}.
  \label{eq:split-rule}
\end{align}
Together with deterministic tie-breaking, the measurability conditions in Section~\ref{sec:model} make this rule measurable.

Let
\begin{align}
  \mathfrak M_P^\star
  =
  \argmax_{m\in\mathfrak M_K}V_{P,m}^\star
  \label{eq:optimal-components}
\end{align}
denote the set of population-optimal portfolio classes.

\begin{theorem}[Selection and allocation upper bound]
\label{thm:selection-allocation-upper}
Suppose that Assumptions \ref{ass:utility}, \ref{ass:gradient-process}, and \ref{ass:utility-loss} hold. For the rule in \eqref{eq:split-rule}, we have
\begin{align}
  \E_P\left[r_P(\widehat w_n)\right]
  &\leq
  \inf_{m\in\mathfrak M_P^\star}
  \E_P\left[V_{P,m}^\star-V_P(\widetilde w_m)\right]
  +
  4\E_P\left[
    \mathfrak C_{\mathrm{sel},n_2}(P\mid\mathcal D_1)
  \right]
  \nonumber\\
  &\leq
  c_0^{-1/(q-1)}
  \inf_{m\in\mathfrak M_P^\star}
  \E_P\left[Z_{n_1,m}(P)^{q/(q-1)}\right]
  +
  4\E_P\left[
    \mathfrak C_{\mathrm{sel},n_2}(P\mid\mathcal D_1)
  \right],
  \label{eq:selection-allocation-upper}
\end{align}
where $Z_{n_1,m}(P)$ is computed from $\mathcal D_1$.
\end{theorem}

The proof is in Appendix~\ref{app:selection-upper}. Because the second term in \eqref{eq:selection-allocation-upper} uses the process of utility differences, the bound also applies to infinitely many classes.

\paragraph{Utility gaps that determine the largest regret.}
Which gaps matter can be seen from two classes whose optimized expected utilities differ by $g\geq0$, where $g$ is estimated by $\widehat g_n$. Suppose that, for constants $c>0$ and $\sigma>0$, where $\sigma$ measures the scale of the deviations of $\widehat g_n$, the probability of choosing the inferior portfolio satisfies
\begin{align}
  \Pr_P\left(\widehat g_n-g\leq-g\right)
  \leq
  \exp\left(-cng^2/\sigma^2\right).
  \label{eq:scalar-tail}
\end{align}
The selection loss is $g$ when the inferior class is chosen and zero otherwise, so its expected value is bounded by
\begin{align}
  g\exp\left(-cng^2/\sigma^2\right).
  \label{eq:gap-tail-contribution}
\end{align}
Now let the gap $g_n$ depend on $n$. If $\sqrt n g_n\to0$, the selection regret multiplied by $\sqrt n$ is at most $\sqrt n g_n$ and tends to zero. If $\sqrt n g_n\to\infty$, it also tends to zero by \eqref{eq:gap-tail-contribution}, provided that the deviation bound holds uniformly with $c$ bounded away from zero and $\sigma$ bounded above. Therefore, only gaps comparable to $n^{-1/2}$ can contribute a positive leading constant, and because the bounds are uniform, the same conclusion holds for the supremum over the original model.

For these gaps, by a uniform central limit theorem, if $\sigma>0$ is the limiting standard deviation of $\sqrt n(\widehat g_n-g_n)$ and $g_n=\sigma t/\sqrt n$, then $\sqrt n$ times the expected selection regret of the sign rule, which selects the class with the larger estimate, converges to $\sigma t\Phi(-t)$ for bounded $t\geq0$, where $\Phi$ is the standard normal distribution function. Taking the largest value over $t$, the comparison of the two classes contributes $\sigma\sup_{t\in[0,\infty)}t\Phi(-t)$ when the model contains the corresponding alternatives.

\paragraph{Complexity of the feasible set.}
The directions along which the weights can vary determine the relevant dimension. Suppose first that $\mathfrak M_K$ contains one class and that the oracle portfolio $w_P^\star$ is unique, and define its tangent cone by
\begin{align}
  T_P
  =
  \overline{\left\{t(w-w_P^\star):t\geq0,\ w\in\calW_K\right\}},
  \label{eq:tangent-cone}
\end{align}
and let
\begin{align}
  \omega(T_P)
  =
  \E\left[
    \sup_{u\in T_P\cap\mathbb S^{K-1}}\abs{G^\top u}
  \right],
  \qquad
  G\sim N(0,I_K),
  \label{eq:tangent-width}
\end{align}
be the Gaussian width of its symmetrized unit directions, where $\mathbb S^{K-1}=\{u\in\bbR^K:\norm{u}_2=1\}$ and $\omega(\{0\})=0$. Put $p_q=q/(q-1)$, and suppose that the directional error in Assumption~\ref{ass:gradient-process} satisfies the moment bound
\begin{align}
  \sup_{P\in\calP_K}
  \E_P\left[Z_{n,1}(P)^{p_q}\right]^{1/p_q}
  \leq
  C_q\sup_{P\in\calP_K}\frac{\omega(T_P)}{\sqrt n}.
  \label{eq:width-gradient}
\end{align}
Proposition~\ref{prop:within-class} then gives
\begin{align}
  \mathcal V_{n,K}^{\mathrm{mm}}
  \lesssim
  \sup_{P\in\calP_K}
  \left(
    \frac{\omega(T_P)^2}{n}
  \right)^{q/(2(q-1))}.
  \label{eq:convex-rate}
\end{align}
For a dense interior optimizer, $\omega(T_P)^2$ can be of order $K$. If optimization is instead restricted to a support $S$ with $\abs{S}=s$, the tangent cone lies in an $(s-1)$-dimensional linear space and its squared Gaussian width is at most of order $s$.

For $\calW_{K,s}^{\ell_0}$, whose portfolio classes are support simplices, Theorem~\ref{thm:selection-allocation-upper} gives
\begin{align}
  \E_P\left[r_P(\widehat w_n)\right]
  \lesssim
  \E_P\left[\mathfrak C_{\mathrm{sel},n_2}(P\mid\mathcal D_1)\right]
  +
  \inf_{S\in\mathfrak M_P^\star}
  \E_P\left[Z_{n_1,S}(P)^{q/(q-1)}\right].
  \label{eq:sparse-upper}
\end{align}
These bounds hold at each sample size, so they also apply when the number of assets or classes changes with $n$, with $q$ and $c_0$ bounded as stated. Appendix~\ref{app:minimax-lower} gives finite-support lower bounds for the allocation and class-selection terms.

\paragraph{The cost of simplifying a portfolio.}
A simpler allocation is useful when the utility lost by imposing its form is small relative to the loss from estimation. Consider one convex class $\mathcal W$ and a matrix $N$ whose columns form an orthonormal basis of its affine tangent space, of dimension $r\geq1$. For a feasible $w$, define
\begin{align}
 \zeta_P(w)&=N^\top\nabla V_P(w),
 &J_P(w)&=-N^\top\nabla^2V_P(w)N.
 \label{eq:marginal-utility-diagnostic}
\end{align}
Suppose that $V_P$ is three times continuously differentiable with bounded third derivatives, and $\lambda I_r\preceq J_P(w)\preceq\Lambda I_r$ throughout the class, for constants $0<\lambda\leq\Lambda<\infty$. Let the maximizer $w_P^\star$ of $V_P$ over $\mathcal W$ lie in the relative interior of $\mathcal W$, and write $\ell_P(w)=V_P(w_P^\star)-V_P(w)$ for the within-class loss of a feasible $w$, which the next proposition relates to $\zeta_P(w)$.

\begin{proposition}[Utility loss and portfolio approximation]
\label{prop:portfolio-approximation}
Under the conditions just stated, the within-class loss satisfies
\begin{align}
 0\leq\ell_P(w)&\leq\frac{\norm{\zeta_P(w)}_2^2}{2\lambda},
 &\ell_P(w)&=\frac12\zeta_P(w)^\top J_P(w)^{-1}\zeta_P(w)
       +O(\norm{\zeta_P(w)}_2^3)
 \label{eq:utility-loss-gradient}
\end{align}
where the inequalities hold for every feasible $w$ and the expansion holds as $w\to w_P^\star$. For any two feasible random portfolios $\widehat w_n$ and $\widetilde w_n$, put $e_n=\widetilde w_n-\widehat w_n$. Then, we have
\begin{align}
 \left|\E_P[\ell_P(\widetilde w_n)]-\E_P[\ell_P(\widehat w_n)]\right|
 &\leq \Lambda\left(\E_P[\norm{\widehat w_n-w_P^\star}_2^2]
                         \E_P[\norm{e_n}_2^2]\right)^{1/2}
       +\frac\Lambda2\E_P[\norm{e_n}_2^2].
 \label{eq:portfolio-approximation-transfer}
\end{align}
Consequently, if $\widehat w_n$ has mean squared error $O(n^{-1})$ and $\E_P[\norm{e_n}_2^2]=o(n^{-1})$, then the two expected regrets differ by $o(n^{-1})$. The same conclusion holds after a supremum or integration over distributions when the stated bounds and remainders hold over that entire set.
\end{proposition}

The proof is in Appendix~\ref{app:portfolio-approximation}. By \eqref{eq:utility-loss-gradient}, the relevant test for a mean--variance approximation is the utility lost at the allocation it produces, not normality of returns. Equation~\eqref{eq:portfolio-approximation-transfer} also gives a sufficient accuracy requirement for a numerical or analytic approximation to the EUR rule within a class.

\section{The Expected Utility Regret Rule}
\label{sec:eurr}
\label{subsec:likelihood-implementation}
This section states our proposed rule, the EUR rule, which is based on maximum likelihood estimation of the return parameter. When only one portfolio class is available, the rule is a single-stage procedure, while with multiple classes it has two stages: a preliminary sample removes clearly inferior classes, and an independent second sample then estimates the portfolios of the remaining classes and determines which one to hold.

\paragraph{Maximum likelihood estimates.}
Both procedures use the maximum likelihood estimator $\widehat\theta(D)$ over $\Theta^+$ in \eqref{eq:mle}, computed from a sample $D$ of return vectors, and take the lexicographically smallest maximizer when there are multiple maximizers.

\paragraph{Estimated optimized expected utilities and weights.}
For a sample $D$ and a class $m\in\mathfrak M_K$, the rule computes two quantities from $\widehat\theta(D)$. First, it evaluates the expected utility of a portfolio $w$ under the fitted distribution, $V_{P_{\widehat\theta(D)}}(w)=\int U(x_0+w^\top r)\,P_{\widehat\theta(D)}(\rmd r)$, using a finite sum over the support points when $\nu$ is counting measure and an integral of the fitted density otherwise. Second, it maximizes this function over the class:
\begin{align}
  \widehat w_m(D)
  \in
  \argmax_{w\in\calW_{K,m}}V_{P_{\widehat\theta(D)}}(w),
  \qquad
  \widehat B_m(D)
  =
  V_{P_{\widehat\theta(D)}}(\widehat w_m(D)).
  \label{eq:plugin-class-optimization}
\end{align}
In the notation of Section~\ref{sec:complexity}, these are the estimated optimal weights $\widehat w_m(D)=w_m(\widehat\theta(D))$ and the estimated optimized expected utility $\widehat B_m(D)=B_m(\widehat\theta(D))$ of the class. Because $U$ is concave, \eqref{eq:plugin-class-optimization} is the maximization of a concave function over a compact convex set, so standard convex optimization methods solve it. Since $\widehat\theta(D)\in\Theta^+$, Assumption~\ref{ass:regular-return-model}(ii), or Assumption~\ref{ass:regular-one-class}(ii) with one class, makes the maximizer unique and places it in the relative interior of the class. The maximizer therefore solves the first-order condition that the directional derivative of $V_{P_{\widehat\theta(D)}}$ at the maximizer is zero in every direction in which one can move within the class. For a singleton class, $\widehat w_m(D)$ is its only portfolio. When only one class is available, the rule computes $\widehat w_1(D)$ from all $n$ observations (Section~\ref{subsec:eurr-one}). With multiple classes, it computes $\widehat B_m$ on the preliminary sample for every $m\in\mathfrak M_K$, but computes $\widehat B_m$ and $\widehat w_m$ on the second sample only for the candidate classes (Section~\ref{subsec:eurr-multiple}).

\subsection{EUR Rule with One Portfolio Class}
\label{subsec:eurr-one}
When only one class is available, there is no comparison between classes. For $n\geq2$, the rule therefore computes $\widehat\theta(D)$ from all $n$ observations $D=\{R_1,\ldots,R_n\}$ and holds $\widehat w_1(D)=w_1(\widehat\theta(D))$, and for $n=1$ it holds a fixed portfolio in $\calW_K$.

\subsection{EUR Rule with Multiple Portfolio Classes}
\label{subsec:eurr-multiple}
With multiple classes, the rule proceeds in two stages: the first determines which classes are worth comparing, and the second chooses among them and estimates the weights to hold.

\paragraph{Sample split and screening.}
For $n\geq2$, the first $k_n$ observations form the preliminary sample $D_0$, and the remaining $N_n=n-k_n$ form the second sample $D_1$. Let $\sigma_U>0$ be a utility scale supplied with the model and homogeneous of degree one in $U$: replacing $U$ by $\lambda U+c$ with $\lambda>0$ replaces $\sigma_U$ by $\lambda\sigma_U$. For a screening threshold $b_n>0$, the rule forms the candidate set
\begin{align}
  \widehat{\mathfrak M}_n
  =\left\{m\in\mathfrak M_K:\widehat B_m(D_0)\geq\max_{\ell\in\mathfrak M_K}\widehat B_\ell(D_0)-\sigma_Ub_n\right\},
  \label{eq:likelihood-retained-classes}
\end{align}
which contains the classes for which the optimized expected utility estimated from $D_0$ lies within $\sigma_Ub_n$ of the largest. The experiments in Section~\ref{sec:experiments} use $\sigma_U=1$.

\paragraph{Choice among the candidate classes.}
All later estimates use $D_1$. If the candidate set contains only one class, the rule holds the estimated optimizer of that class, and if it contains two, it holds the estimated optimizer of the class with the larger estimated optimized expected utility, choosing the smaller index at equality. If it contains three or more classes, the rule holds the estimated optimizer of the class selected by the finite Gaussian comparison of Section~\ref{subsec:gaussian-comparison}, which also uses a covariance estimated from $D_0$.

\paragraph{Tuning sequences.}
The sequences $k_n$ and $b_n$ must satisfy
\begin{align}
  k_n\to\infty,
  \qquad
  \frac{k_n}n\to0,
  \qquad
  k_nb_n^2\to\infty,
  \qquad
  \sqrt n\,b_n^2\to0,
  \label{eq:tuning-conditions}
\end{align}
together with $k_nb_n^2/\log n\to\infty$ and $N_nb_n^2/\log n\to\infty$. Every pair $k_n=\lfloor n^\alpha\rfloor$ and $b_n=k_n^{-\beta}$ with $\alpha<1$, $0<\beta<1/2$, and $\alpha\beta>1/4$, which force $\alpha>1/2$ and $\alpha\beta<1/2$, satisfies these conditions and is compatible with $\tau_n$ and $\epsilon_n$ in Section~\ref{subsec:gaussian-comparison}, and we use $\alpha=4/5$ and $\beta=3/8$. Section~\ref{subsec:eurr-tuning} explains what each condition controls in the proofs.

\subsection{Statement of the Rule}
\label{subsec:gaussian-comparison}
\label{subsec:eurr-statement}
This subsection completes the rule for three or more candidate classes and collects all steps into Algorithm~\ref{rule:eurr}.

\paragraph{Observed utility differences.}
When the candidate set has three or more classes, the rule selects a class at random, with selection probabilities computed from a finite linear program. Let $C=\widehat{\mathfrak M}_n$ with $\abs{C}\geq3$, let $m_0=\min C$ be the reference class, and put $q=\abs{C}-1$. To first order, the covariance of $\sqrt{N_n}$ times the error in the estimate of the vector of utility differences $d_C(\theta)=(B_m(\theta)-B_{m_0}(\theta))_{m\in C\setminus\{m_0\}}$ is
\begin{align}
  \Omega_C(\theta)
  =\nabla d_C(\theta)I_\theta^{-1}\nabla d_C(\theta)^\top.
  \label{eq:likelihood-contrast-covariance}
\end{align}
With $\tau_n=n^{-1/32}$, the rule computes the covariance $\widehat\Omega_C=\Omega_C(\widehat\theta(D_0))+\tau_n^2I_q$ from $D_0$ and the observation $Y_C=\sqrt{N_n}\,d_C(\widehat\theta(D_1))+\tau_nZ_C$ from $D_1$, where $Z_C$ is a standard normal vector independent of the data.

\paragraph{The linear program.}
Two bounds are supplied with the model: $H_U\geq1$ bounds $\abs{U(x_0+w^\top R)}$ on the common return support and feasible set, and $\Lambda\geq1$ bounds the largest eigenvalue of $\Omega_C(\vartheta)+I_q$ for every $C\subset\mathfrak M_K$ with $\abs{C}\geq3$ and every $\vartheta\in\Theta^+$. Put $\rho_n=2H_U\sqrt{N_n}$, $t_n=(1+\sqrt\Lambda)\log(n+1)$, and $\epsilon_n=n^{-4}$. The means $\mu$ considered by the program lie on an equally spaced Cartesian grid $\mathcal G_n$ on $[-\rho_n,\rho_n]^q$ that has mesh at most $\epsilon_n$ and includes the endpoints. The finite partition $\mathcal Q_n$ divides $[-\rho_n-t_n,\rho_n+t_n]^q$ into half-open Cartesian cells of side at most $\epsilon_n$, with the outer boundary assigned to the adjacent cells, and $O_n=\bbR^q\setminus\bigcup_{Q\in\mathcal Q_n}Q$ is its exterior. For $\mu\in\bbR^q$, append $\mu_{m_0}=0$ and let $L_C(\mu,j)=\max_{i\in C}\mu_i-\mu_j$ be the utility-gap loss of selecting $j\in C$. With $\Phi_{\widehat\Omega_C}$ the centered Gaussian probability measure with covariance $\widehat\Omega_C$, the rule takes the lexicographically selected solution of the finite linear program
\begin{align}
  \min_{\substack{v\in\bbR,\ \pi\in\bbR^{\mathcal Q_n\times C}}}\quad &v
  \nonumber\\
  \text{subject to}\quad
  &\pi_{Qj}\geq0,\qquad \sum_{j\in C}\pi_{Qj}=1\quad(Q\in\mathcal Q_n),
  \nonumber\\
  &\sum_{Q\in\mathcal Q_n}\Phi_{\widehat\Omega_C}(Q-\mu)
       \sum_{j\in C}\pi_{Qj}L_C(\mu,j)
    \nonumber\\
  &\qquad+\Phi_{\widehat\Omega_C}(O_n-\mu)L_C(\mu,m_0)
    \leq v\quad(\mu\in\mathcal G_n).
  \label{eq:likelihood-gaussian-program}
\end{align}
For each mean $\mu$ on the grid, the last constraint requires that the expected utility gap be at most $v$ for the rule that draws $j$ with probabilities $(\pi_{Qj})_{j\in C}$ when $Y\sim N(\mu,\widehat\Omega_C)$ falls in $Q$ and selects $m_0$ otherwise.

\paragraph{Selection.}
If $Y_C$ lies in a cell $Q\in\mathcal Q_n$, the rule draws $j\in C$ with probabilities $(\pi_{Qj})_{j\in C}$; otherwise it sets $j=m_0$. After appending $Y_{m_0}=0$, it replaces $j$ by an index attaining $\max_{i\in C}Y_i$ whenever $\max_{i\in C}Y_i-Y_j>\log(n+1)$, and it then holds $\widehat w_j(D_1)=w_j(\widehat\theta(D_1))$.

\paragraph{The complete rule.}
Algorithm~\ref{rule:eurr} collects these steps into the EUR rule $\delta_n^\star$, which uses neither the true parameter nor an evaluation prior. Figure~\ref{fig:eurr-diagram} shows which subsample each step uses when at least two classes are available.

\begin{algorithm}[t]
\caption{EUR rule}
\label{rule:eurr}
\begin{algorithmic}[1]
\REQUIRE Returns $R_1,\ldots,R_n$; likelihood family on $\Theta^+$; utility $U$; classes $(\calW_{K,m})_{m\in\mathfrak M_K}$; utility scale $\sigma_U$ when at least two classes are available; bounds $H_U$ and $\Lambda$ when the candidate set can have three or more classes.
\IF{$n=1$}
  \STATE Return a fixed portfolio in $\calW_K$.
\ENDIF
\IF{$|\mathfrak M_K|=1$}
  \STATE Let $j$ be its only index, compute the maximum likelihood estimate $\widehat\theta$ of $\theta$ from all $n$ observations, and return $w_j(\widehat\theta)$.
\ENDIF
\STATE Set $k_n=\lfloor n^{4/5}\rfloor$, $N_n=n-k_n$, and $b_n=k_n^{-3/8}$.
\STATE Compute the maximum likelihood estimates $\widehat\theta(D_0)$ of $\theta$ from the first $k_n$ observations and $\widehat\theta(D_1)$ from the remaining observations.
\STATE Form $C=\widehat{\mathfrak M}_n$ using \eqref{eq:likelihood-retained-classes}.
\IF{$|C|=1$}
  \STATE Let $j$ be its only index.
\ELSIF{$|C|=2$}
  \STATE Let $j$ maximize $B_j(\widehat\theta(D_1))$ over $j\in C$, using the smaller index at equality.
\ELSE
  \STATE Set $m_0=\min C$, $q=|C|-1$, and $\tau_n=n^{-1/32}$.
  \STATE Calculate $\widehat\Omega_C=\Omega_C(\widehat\theta(D_0))+\tau_n^2I_q$.
  \STATE Draw $Z_C\sim N(0,I_q)$ independently and calculate $Y_C=\sqrt{N_n}\,d_C(\widehat\theta(D_1))+\tau_nZ_C$.
  \STATE Construct $\mathcal G_n$ and $\mathcal Q_n$ and solve \eqref{eq:likelihood-gaussian-program} for $\pi$.
  \STATE Draw $j$ using $\pi_{Qj}$ if $Y_C\in Q$ holds for some $Q\in\mathcal Q_n$; otherwise let $j=m_0$.
  \STATE Append $Y_{m_0}=0$. If $\max_{i\in C}Y_i-Y_j>\log(n+1)$ holds, replace $j$ by an index attaining $\max_{i\in C}Y_i$.
\ENDIF
\STATE Return $\widehat w_j(D_1)=w_j(\widehat\theta(D_1))$.
\end{algorithmic}
\end{algorithm}

\begin{figure}[t]
\centering
\includegraphics[width=0.78\textwidth]{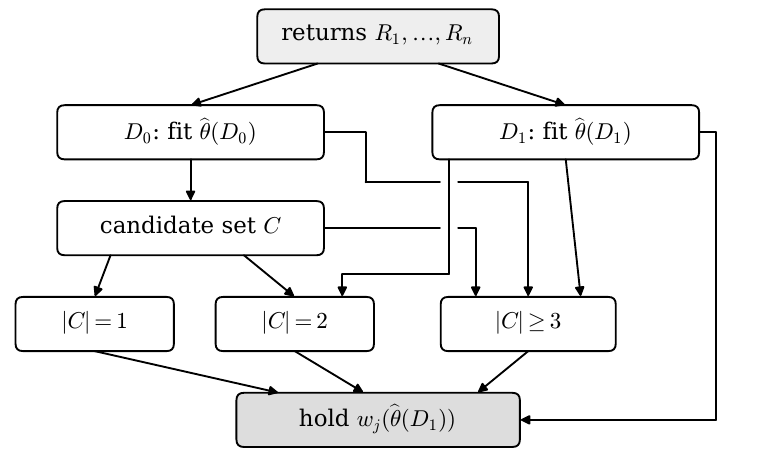}
\caption{The EUR rule with at least two portfolio classes. Each arrow points from a sample or quantity to a step that uses it. The preliminary sample $D_0$ determines the candidate set $C$ and, when $\abs{C}\geq3$, the covariance $\widehat\Omega_C$ of the finite Gaussian comparison, whereas every other estimate is computed from the second sample $D_1$ alone. The index $j$ is the class selected in the branch that applies.}
\label{fig:eurr-diagram}
\end{figure}

\section{Minimax Optimality}
\label{sec:minimax}
We compare the EUR rule with every measurable randomized portfolio choice rule over the fixed evaluation set $\Theta$ under Assumption~\ref{ass:regular-return-model}, or under Assumption~\ref{ass:regular-one-class} when only one class is available. Throughout, $K$, $p$, the portfolio classes, the utility, $\Theta$, and $\Theta^+$ stay fixed as $n\to\infty$.

\subsection{Minimax Constant}
The leading constant depends on whether different classes can attain the same largest expected utility. Near a tie, a wrong choice between distinct class-specific portfolios loses utility of the same order as the estimation error. Away from ties, by contrast, the leading loss is quadratic in the error of the estimated weights.

For a finite set $C$ of classes and a positive definite $(\abs{C}-1)\times(\abs{C}-1)$ matrix $\Omega$, which plays the role of $\Omega_C(\theta)$ in \eqref{eq:likelihood-contrast-covariance}, let $\mathfrak D_C$ be the set of measurable probability kernels from $\bbR^{|C|-1}$ to $C$. With the utility-gap loss $L_C$ of Section~\ref{subsec:gaussian-comparison}, put
\begin{align}
  \gamma_C(\Omega)
  =\inf_{\delta\in\mathfrak D_C}\sup_{\mu\in\bbR^{|C|-1}}
      \E[L_C(\mu,\delta(\mu+G))],\qquad G\sim N(0,\Omega).
  \label{eq:likelihood-gaussian-value}
\end{align}
For two classes the value has a closed form.

\begin{lemma}[Value of the two-class comparison]
\label{lem:two-class-gaussian-value}
Let $\abs{C}=2$, so that $q=1$ and $\Omega>0$ is a scalar. Then, we have $\gamma_C(\Omega)=\kappa\sqrt\Omega$, where $\kappa=\sup_{u\geq0}u\Phi(-u)$ and $\Phi$ is the standard normal distribution function.
\end{lemma}

The proof is in Appendix~\ref{app:likelihood-global}. Because the sign rule, which holds the class with the larger estimated optimized expected utility (Section~\ref{subsec:eurr-multiple}), attains this value, a two-class candidate set needs neither a grid nor randomization. Let $\mathcal T=\{\theta\in\Theta:|A(\theta)|\geq2\}$ be the set of parameters at which at least two classes are optimal, and when $\mathcal T$ is nonempty, define $\Gamma_{\mathrm{lik}}=\sup_{\theta\in\mathcal T}\gamma_{A(\theta)}(\Omega_{A(\theta)}(\theta))$.

For the allocation term, let the columns of $N_m$ be an orthonormal basis for the affine tangent space of class $m$, write $\dot w_m(\theta)\in\bbR^{K\times p}$ for the Jacobian of its optimal weights, and set $D_m(\theta)=N_m^\top\dot w_m(\theta)$ and $J_m(\theta)=-N_m^\top\nabla_w^2V_{P_\theta}(w_m(\theta))N_m$. We use the constants
\begin{align}
  c_m(\theta)&=\frac12\operatorname{tr}\left(J_m(\theta)D_m(\theta)I_\theta^{-1}D_m(\theta)^\top\right),
  \nonumber\\
  s_{m\ell}^2(\theta)&=\nabla(B_m-B_\ell)(\theta)^\top I_\theta^{-1}\nabla(B_m-B_\ell)(\theta).
  \label{eq:likelihood-constants}
\end{align}
Here, $c_m(\theta)=0$ for a singleton class, which has no weights to estimate, and $s_{m\ell}^2(\theta)$ is the variance of the limiting normalized estimation error of $B_m-B_\ell$.

Let $m(\theta)$ denote the optimal class wherever it is unique. Define
\begin{align}
  C_{\mathrm{mm}}
  =\begin{cases}
    \Gamma_{\mathrm{lik}},&\mathcal T\ne\varnothing,\\
    \sup_{\theta\in\Theta}c_{m(\theta)}(\theta),&\mathcal T=\varnothing.
  \end{cases}
  \label{eq:minimax-constant-definition}
\end{align}
The normalization is $a_n=\sqrt n$ in the first case and $a_n=n$ in the second.

\subsection{Minimax Lower Bound}
Near an optimal tie, the lower bound comes from the Gaussian comparison in \eqref{eq:likelihood-gaussian-value}, which no rule can resolve more accurately, and without a tie it comes from the information available for estimating the optimal weights.

\begin{proposition}[Lower bound for every portfolio choice rule]
\label{prop:regular-minimax-lower}
Under Assumption~\ref{ass:regular-return-model}, or Assumption~\ref{ass:regular-one-class} with one class, with $a_n$ and $C_{\mathrm{mm}}$ defined above, we have
\begin{align}
  \liminf_{n\to\infty}a_n\mathcal V_{n,K}^{\mathrm{mm}}
  \geq C_{\mathrm{mm}}.
  \label{eq:regular-minimax-lower}
\end{align}
The infimum defining $\mathcal V_{n,K}^{\mathrm{mm}}$ ranges over all rules based on the original $n$ return observations.
\end{proposition}

For a fixed tie point $\theta_0$ and finitely many parameters displaced from it by order $n^{-1/2}$, a likelihood expansion yields a Gaussian comparison with covariance $\Omega_{A(\theta_0)}(\theta_0)$. Because the class-specific optimizers are distinct, any portfolio choice can be converted into a class choice whose leading loss is no larger. If $\mathcal T$ is empty, the allocation lower bound follows instead from a prior concentrated near a parameter at which $c_{m(\theta)}(\theta)$ is largest.

\subsection{Worst-Case Upper Bound}
For the EUR rule, we bound the worst-case expected regret over all of $\Theta$, including parameters that vary with $n$.

\begin{proposition}[Worst-case expected regret of the EUR rule]
\label{prop:regular-minimax-upper}
Under Assumption~\ref{ass:regular-return-model}, or Assumption~\ref{ass:regular-one-class} with one class, the EUR rule satisfies
\begin{align}
  \limsup_{n\to\infty}
  a_n\sup_{\theta\in\Theta}\calR_n(\delta_n^\star,P_\theta)
  \leq C_{\mathrm{mm}}.
  \label{eq:regular-minimax-upper}
\end{align}
\end{proposition}

The candidate set misses a required class with probability at most $C\exp(-c k_nb_n^2)$, where $c$ and $C$ do not depend on $\theta\in\Theta$, and $n$ times this probability also vanishes. After a successful preliminary comparison, every candidate class has utility gap at most $3\sigma_Ub_n/2$, so that for $\theta_n\in\Theta$ converging to $\theta_0$, the candidate classes eventually belong to $A(\theta_0)$. Because Appendix~\ref{app:likelihood-global} bounds the errors from covariance estimation, the perturbation, and the grid, the selection contribution multiplied by $\sqrt n$ is asymptotically at most $\gamma_{A(\theta_0)}(\Omega_{A(\theta_0)}(\theta_0))$. In addition, within-class estimation contributes $c_m(\theta)/n+o(n^{-1})$ uniformly over $\Theta$. If $\Theta$ contains no tie, compactness gives each inferior class a positive minimum gap, so only the allocation term of order $n^{-1}$ remains.

\subsection{Minimax Optimality}
Because the lower and upper bounds share the same constant, the EUR rule has the smallest possible leading worst-case expected regret in the specified return model.

\begin{theorem}[Minimax optimality of the EUR rule]
\label{thm:likelihood-global}
Under Assumption~\ref{ass:regular-return-model}, if $\mathcal T\ne\varnothing$ holds, then we have
\begin{align}
  \lim_{n\to\infty}\sqrt n\sup_{\theta\in\Theta}\calR_n(\delta_n^\star,P_\theta)
  =\lim_{n\to\infty}\sqrt n\mathcal V_{n,K}^{\mathrm{mm}}
  =\Gamma_{\mathrm{lik}}.
  \label{eq:likelihood-minimax}
\end{align}
If $\mathcal T=\varnothing$ holds, then the following limits hold. They also hold for one class under Assumption~\ref{ass:regular-one-class}:
\begin{align}
  \lim_{n\to\infty}n\sup_{\theta\in\Theta}\calR_n(\delta_n^\star,P_\theta)
  =\lim_{n\to\infty}n\mathcal V_{n,K}^{\mathrm{mm}}
  =\sup_{\theta\in\Theta}c_{m(\theta)}(\theta).
  \label{eq:likelihood-unique-minimax}
\end{align}
\end{theorem}

These limits follow from Propositions~\ref{prop:regular-minimax-lower} and~\ref{prop:regular-minimax-upper}, since $\mathcal V_{n,K}^{\mathrm{mm}}$ is at most the worst-case expected regret of the EUR rule. A positive tie contribution, which is of order $n^{-1/2}$, dominates the allocation term of order $n^{-1}$.

\subsection{Tuning Parameters and Approximations}
\label{subsec:eurr-tuning}
In the proofs of the upper bounds for the EUR rule, Propositions~\ref{prop:regular-minimax-upper} and~\ref{prop:regular-bayes-upper}, the tuning sequences and the finite approximations of Algorithm~\ref{rule:eurr} play the following roles.

\paragraph{Screening.}
The first two conditions in \eqref{eq:tuning-conditions} make the preliminary estimate consistent while keeping $N_n/n\to1$, so that the second sample contains almost all observations. The third ensures that the preliminary comparison removes classes with a sufficiently large utility gap with probability tending to one, exponentially fast in $k_nb_n^2$. The fourth makes parameters near a tie of three classes negligible for the Bayes criterion (Section~\ref{sec:bayes}). The logarithmic conditions make the probabilities of an incorrect removal on $D_0$ and of selecting on $D_1$ a class with utility gap above $\sigma_Ub_n/2$ negligible after multiplication by $n$. Scaling the threshold by $\sigma_U$ makes the candidate set invariant to positive affine transformations of $U$. Our choice $\alpha=4/5$ and $\beta=3/8$ gives $\alpha\beta=3/10$, $k_nb_n^2\asymp n^{1/5}$, and $\sqrt n\,b_n^2\asymp n^{-1/10}$.

\paragraph{Comparison of the candidate classes.}
Two candidate classes are compared by the sign rule, which attains the value in Lemma~\ref{lem:two-class-gaussian-value}. With three or more candidate classes, however, the estimated utility differences are correlated through $\Omega_C$, so selecting the class with the largest estimate need not minimize the largest expected utility gap. The program \eqref{eq:likelihood-gaussian-program} is a finite version of the comparison whose value $\gamma_C(\Omega)$ is defined in \eqref{eq:likelihood-gaussian-value}, and Appendix~\ref{app:likelihood-global} bounds the difference. The dimension $q=\abs{C}-1$ here is unrelated to the curvature exponent $q$ in Assumption~\ref{ass:utility-loss}.

\paragraph{Approximations in the Gaussian comparison.}
The perturbation $\tau_nZ_C$ makes $Y_C$ comparable in total variation to a Gaussian vector even when the returns have finite support, and its vanishing variance leaves the limiting constant unchanged. The mesh $\epsilon_n$ controls the discretization error. The half-width $\rho_n$ covers every mean of $\sqrt{N_n}\,d_C$, since each utility difference is at most $2H_U$ in absolute value, and the margin $t_n$ makes the probability of an observation falling outside the cells negligible. The threshold $\log(n+1)$ bounds the loss from a class with an unusually poor observed utility difference. The grid has on the order of $n^{9q/2}$ points, that is, $n^9$ with three candidate classes.

\subsection{Discussion}
The two normalizations imply different costs of deciding from data. Provided the constant is positive, halving the leading worst-case regret requires about four times as many observations when class selection dominates, but only twice as many when weight estimation alone matters.

\paragraph{A return model with both contributions.}
A bounded return model makes both contributions explicit: the investor chooses between two opposite positions and estimates a continuous position in a third asset. Take independent signs $X,Y\in\{-1,1\}$ with means $\theta$ and $\eta$, where $(\theta,\eta)\in[-1/4,1/4]\times[-1/20,1/20]$, and let $R=(\varepsilon X,-\varepsilon X,\lambda Y)$ with $\varepsilon=1/10$ and $\lambda=1/5$. The two classes are
\begin{align}
  \calW_+
  &=
  \{(1,0,t):\abs{t}\leq1/2\},
  &
  \calW_-
  &=
  \{(0,1,t):\abs{t}\leq1/2\}.
  \label{eq:joint-estimation-classes}
\end{align}
Let $U(x)=-\exp(-x)$. For $s\in\{+1,-1\}$, the portfolio of $\calW_s$ that holds $t$ in the third asset has the expected utility $V_s$ given below. Since both classes have the same interior optimizer $t^\star(\eta)$ in the third asset, substituting it gives the optimized values $B_s$ and their difference. On the decision boundary $\theta=0$, the standard deviation of the limiting normalized error of an efficient estimate of this difference (Appendix~\ref{app:influence}) is $s_{+-}(0,\eta)$, and the allocation contribution within class $s$ is $c_s$:
\begin{align}
  V_s(t;\theta,\eta)
  &=
  -e^{-x_0}
  \left(\cosh(\varepsilon)-s\theta\sinh(\varepsilon)\right)
  \left(\cosh(\lambda t)-\eta\sinh(\lambda t)\right),
  \label{eq:joint-estimation-objective}\\
  t^\star(\eta)
  &=
  \frac{\operatorname{arctanh}(\eta)}\lambda,
  \label{eq:joint-estimation-weight}\\
  B_s(\theta,\eta)
  &=
  -e^{-x_0}
  \left(\cosh(\varepsilon)-s\theta\sinh(\varepsilon)\right)
  \sqrt{1-\eta^2},
  \label{eq:joint-estimation-example}\\
  B_+(\theta,\eta)-B_-(\theta,\eta)
  &=
  2e^{-x_0}\theta\sinh(\varepsilon)\sqrt{1-\eta^2},
  \label{eq:joint-estimation-gap}\\
  s_{+-}(0,\eta)
  &=
  2e^{-x_0}\sinh(\varepsilon)\sqrt{1-\eta^2},
  \label{eq:joint-estimation-contrast-sd}\\
  c_s(\theta,\eta)
  &=
  \frac{e^{-x_0}
  \left(\cosh(\varepsilon)-s\theta\sinh(\varepsilon)\right)}
  {2\sqrt{1-\eta^2}}.
  \label{eq:joint-estimation-allocation-constant}
\end{align}
The rectangle $\Theta^+=[-3/10,3/10]\times[-3/50,3/50]$ satisfies Assumption~\ref{ass:regular-return-model}. For alternatives $\theta_n=h/\sqrt n$ with nonzero $h$, choosing the class with the wrong sign costs utility of order $n^{-1/2}$ with a probability that need not vanish, and this gives the leading worst-case selection term. At $\theta=0$, by contrast, selection causes no regret because both classes attain the same optimized value, while estimating the third-asset position still contributes regret of order $n^{-1}$.

\section{Bayes Optimality}
\label{sec:bayes}
In the setting of Section~\ref{sec:minimax}, we now integrate the expected regret of the same EUR rule under a prior $H\in\mathfrak H$ on $\Theta$ that does not depend on $n$, and we let $n\to\infty$ for each such prior.

\subsection{Bayes Constant}
The Bayes constant is the sum of a contribution from class selection near pairwise ties and a contribution from estimating weights. Let $h$ be the continuous density of $H$, and write $V_m^{\mathrm{eff}}(\theta)=D_m(\theta)I_\theta^{-1}D_m(\theta)^\top$ for the covariance of the limiting normalized estimation error of the optimizer in tangent coordinates. The plug-in estimator $w_m(\widehat\theta)$ based on maximum likelihood attains this covariance, which equals the information bound for regular estimates of the optimizer (Appendix~\ref{app:influence}).

The regions with a unique optimal class are indexed by $m(\theta)$, and the boundaries at which exactly two classes attain the largest expected utility are
\begin{align}
    \calT_{m\ell}
    =
    \left\{
      \theta\in\operatorname{int}(\Theta):
      B_m(\theta)=B_\ell(\theta)>B_j(\theta)
      \text{ for every }j\in\mathfrak M_K\setminus\{m,\ell\}
    \right\}.
    \label{eq:pairwise-boundaries}
  \end{align}
For this prior, define
\begin{align}
  C_{\mathrm{sel}}(H)
  &={}
  \frac12
  \sum_{\{m,\ell\}}
  \int_{\calT_{m\ell}}
  \frac{s_{m\ell}^2(\theta)h(\theta)}
  {\norm{\nabla(B_m-B_\ell)(\theta)}_2}
  \rmd\mathcal H^{p-1}(\theta),
  \label{eq:bayes-selection-constant}
  \\
  C_{\mathrm{alloc}}(H)
  &={}
  \int_\Theta
  \frac12\operatorname{tr}\left(
    J_{m(\theta)}(\theta)V_{m(\theta)}^{\mathrm{eff}}(\theta)
  \right)
  h(\theta)\rmd\theta.
  \label{eq:bayes-allocation-constant}
\end{align}
The sum is over unordered pairs, and $\mathcal H^{p-1}$ is $(p-1)$-dimensional Hausdorff measure. In \eqref{eq:bayes-allocation-constant}, the value of $m(\theta)$ on the lower-dimensional tie sets does not affect the integral, and the integrand is $c_{m(\theta)}(\theta)h(\theta)$ by \eqref{eq:likelihood-constants}.

Where $h$ is positive on a tie surface, a neighborhood of width $n^{-1/2}$ has prior probability of order $n^{-1/2}$, and the utility gap within it is also of order $n^{-1/2}$, so selection there contributes regret of order $n^{-1}$. The coefficient $1/2$ in \eqref{eq:bayes-selection-constant} follows from $\int_{\bbR}|v|\Phi(-|v|/\sigma)\,\rmd v=\sigma^2/2$. Define $C_H=C_{\mathrm{sel}}(H)+C_{\mathrm{alloc}}(H)$.

\subsection{Bayes Lower Bound}
The lower bound below applies to every rule in the infimum \eqref{eq:bayes-value}, including rules that use $H$, because knowing the evaluation prior cannot remove the estimation error in local comparisons or in the weights.

\begin{proposition}[Integrated lower bound]
\label{prop:regular-bayes-lower}
Under Assumption~\ref{ass:regular-return-model}, or Assumption~\ref{ass:regular-one-class} with one class, every $H\in\mathfrak H$ specified above satisfies
\begin{align}
  \liminf_{n\to\infty} n\mathcal V_{n,K}^{H}
  \geq C_H.
  \label{eq:regular-bayes-lower}
\end{align}
\end{proposition}

The proof in Appendices~\ref{app:likelihood-global} and~\ref{app:regular-one-class} rests on two facts. First, near a regular pairwise tie, the error of the posterior comparison is, to leading order, the Gaussian error of the efficient estimate of the utility difference. Second, where the optimal class is unique, the posterior-optimal weights have the efficient leading covariance $V_{m(\theta)}^{\mathrm{eff}}(\theta)$.

\subsection{Worst-Case Upper Bound}
To integrate the local approximations over the prior, we need a bound that holds for every $\theta\in\Theta$.

Let $B_{(1)}(\theta)\geq B_{(2)}(\theta)\geq B_{(3)}(\theta)$ be the three largest optimized expected utilities. If fewer than three classes are available, set $E_n=\varnothing$; otherwise put $E_n=\{\theta\in\Theta:B_{(1)}(\theta)-B_{(3)}(\theta)\leq3\sigma_Ub_n/2\}$. Because the gradients of the two utility differences have full row rank at a triple tie, $\operatorname{Leb}(E_n)\leq Cb_n^2$, and the worst-case bound of Section~\ref{sec:minimax} then yields
\begin{align}
  n\int_{E_n}\calR_n(\delta_n^\star,P_\theta)h(\theta)\,\rmd\theta
  \leq C\sqrt n\,b_n^2+o(1)\longrightarrow0.
  \label{eq:bayes-multiple-neighborhood}
\end{align}

Outside $E_n$, a successful preliminary comparison leaves at most two candidate classes, and when their utility gap $g$ is positive, the expected selection loss is at most $Cg\exp(-cN_ng^2)$ up to negligible terms. After the change of variables $g=v/\sqrt{N_n}$, the remaining integrand is dominated by $C|v|\exp(-cv^2)$ and has the local limit $|v|\Phi(-|v|/s_{m\ell})$, and the coarea formula converts the integral over the gap into the surface integral in \eqref{eq:bayes-selection-constant}.

\begin{proposition}[Integrated upper bound for the EUR rule]
\label{prop:regular-bayes-upper}
Under Assumption~\ref{ass:regular-return-model}, or Assumption~\ref{ass:regular-one-class} with one class, for each $H\in\mathfrak H$ specified above, we have
\begin{align}
  \limsup_{n\to\infty}n\int_\Theta
  \calR_n(\delta_n^\star,P_\theta)\,\rmd H(\theta)
  \leq C_H.
  \label{eq:regular-bayes-upper}
\end{align}
\end{proposition}

After multiplication by $n$, the within-class term converges to $c_{m(\theta)}(\theta)$ wherever the optimal class is unique, and a bounded second moment of the normalized allocation loss justifies integration, giving \eqref{eq:bayes-allocation-constant}. Appendices~\ref{app:likelihood-global} and~\ref{app:regular-one-class} give the details.

\subsection{Bayes Optimality}
The integrated lower bound and the upper bound for the EUR rule agree, so using the evaluation prior cannot improve the leading constant.

\begin{theorem}[Bayes optimality of the EUR rule]
\label{thm:likelihood-bayes}
Under Assumption~\ref{ass:regular-return-model}, or Assumption~\ref{ass:regular-one-class} with one class, every $H\in\mathfrak H$ specified above satisfies
\begin{align}
  \lim_{n\to\infty}n\int_\Theta\calR_n(\delta_n^\star,P_\theta)\,\rmd H(\theta)
  =\lim_{n\to\infty}n\mathcal V_{n,K}^{H}
  =C_{\mathrm{sel}}(H)+C_{\mathrm{alloc}}(H).
  \label{eq:likelihood-bayes}
\end{align}
The rule $\delta_n^\star$ is exactly the EUR rule of Section~\ref{sec:eurr} and is independent of $H$. In the one-class case, $C_{\mathrm{sel}}(H)=0$ and the allocation constant is the integral of $c_1(\theta)$ with respect to $H$, as proved in Appendix~\ref{app:regular-one-class}.
\end{theorem}

Theorem~\ref{thm:likelihood-bayes} follows because Proposition~\ref{prop:regular-bayes-lower} bounds the infimum in \eqref{eq:bayes-value} from below and Proposition~\ref{prop:regular-bayes-upper} bounds it from above through the EUR rule. Together with Theorem~\ref{thm:likelihood-global}, it establishes simultaneous asymptotic minimax and Bayes optimality for a single rule defined at every sample size.

\subsection{Discussion}
The two criteria weight difficult comparisons differently. A parameter close to a tie can determine the minimax constant although its shrinking neighborhood has small probability under a continuous prior that does not depend on $n$, and that probability supplies the extra factor $n^{-1/2}$ in the Bayes selection term.

\citet[Theorem~3]{Lai1987adaptivetreatment} derives matching leading bounds for cumulative regret integrated against a prior on the arm parameters, and, as in this section, the prior does not depend on the horizon. \citet[Sections~2.2.1 and~5]{Adusumilli2023risk} instead places a prior that does not depend on $n$ on $\sqrt n(\theta-\theta_0)$. The priors of Section~\ref{sec:distribution} concentrate at zero more slowly than $n^{-1/2}$.

The allocation term remains even without ties, because it is positive whenever estimating the return parameter changes the optimal weights in a direction with positive utility curvature. For instance, in the return model of Section~\ref{sec:minimax} with both contributions, a prior with positive density near $\theta=0$ assigns a nonzero leading cost both to the choice between the two opposite positions and to the position in the third asset.

\section{Relationship to the Existing Portfolios and Their Extensions}
\label{sec:existing-portfolios}
Two properties link the EUR rule to the mean--variance and risk-parity portfolios. First, when expected excess returns lie in neighborhoods of zero that shrink faster than $n^{-1/4}$ but more slowly than $n^{-1/2}$, the EUR rule and the sample mean--variance portfolio attain the same constants. Second, when standardized returns and feasible standardized holdings are invariant under a transitive group of asset permutations, the EUR rule holds equal standardized positions, which is the risk-parity form. We then quantify the expected utility lost when these conditions fail.

\subsection{Portfolios from a Small Optimal Risky Investment}
\label{sec:distribution}
We give conditions under which a mean--variance portfolio, which uses only the first two moments of returns, preserves the leading constant in expected regret.

\paragraph{Conditions on the utility.}
In this section, the investor can also hold a riskless cash asset besides the $K$ risky assets, and we measure wealth in units of the cash account, so that the cash return is zero. The vector $R$ then contains the returns of the risky assets in excess of the cash return, $w_a$ is the amount invested in risky asset $a$, and the remaining wealth $x_0-\sum_{a=1}^Kw_a$ is held in cash. Terminal wealth is therefore $x_0+w^\top R$, as in Section~\ref{sec:model}, and the portfolio $w=0$ holds all wealth in cash. The feasible set $\calW_K$ is compact and convex with zero in its interior, so small long and short positions in each risky asset are feasible.

Let $P_0$ be a bounded baseline return distribution satisfying $\E_0[R]=0$ and $\Sigma_0=\E_0[RR^\top]\succ0$. Define an exponential family by
\begin{align}
  \frac{\rmd P_\mu}{\rmd P_0}(r)
  &=\exp\left(\eta(\mu)^\top r-\psi(\eta(\mu))\right),
  &\psi(\eta)&=\log\E_0[\exp(\eta^\top R)],
  \label{eq:mv-mean-family}
\end{align}
where, for $\mu$ in a sufficiently small neighborhood of zero that does not depend on $n$, $\eta(\mu)$ solves $\nabla\psi(\eta)=\mu$, and this solution exists and is smooth because $\nabla^2\psi(0)=\Sigma_0\succ0$. Under $P_\mu$, the mean is $\mu$ and the covariance is $\Sigma_\mu=\nabla^2\psi(\eta(\mu))$. The observations are independent draws from $P_\mu$, with $P_0$ known and $\mu$ unknown. Proposition~\ref{prop:mv-localization} and Theorem~\ref{thm:local-mean-variance} use the following smoothness conditions on the utility near initial wealth.

\begin{assumption}[Smooth utility near a cash allocation]
\label{ass:mv-local-utility}
The utility function is increasing and concave and belongs to $C^3(J)$ on an open interval $J$ containing the compact set of attainable values of terminal wealth $x_0+w^\top R$. At initial wealth, $u_1=U'(x_0)>0$ and $b=-U''(x_0)>0$. The number of assets, the baseline distribution, the feasible set, and the utility do not depend on the sample size.
\end{assumption}

Write $A=b/u_1$ for the absolute risk aversion of $U$ at $x_0$. Since cash is the unique optimum at zero expected excess returns, the next result shows that optimal positions remain small when the mean is small.

\begin{proposition}[Optimal allocations near cash]
\label{prop:mv-localization}
Under \eqref{eq:mv-mean-family} and Assumption~\ref{ass:mv-local-utility}, the oracle $w^\star(\mu)$ is unique and interior for all sufficiently small $\norm{\mu}_2$. As $\mu\to0$, we have
\begin{align}
  \norm{w^\star(\mu)}_2&=O(\norm{\mu}_2),
  &w^\star(\mu)&=A^{-1}\Sigma_\mu^{-1}\mu+O(\norm{\mu}_2^2).
  \label{eq:mv-oracle-localization}
\end{align}
Moreover, for feasible $w$ near zero, we have
\begin{align}
  V_{P_\mu}(w)
  &=U(x_0)+u_1w^\top\mu
    -\frac b2 w^\top\Sigma_\mu w
    -\frac b2(w^\top\mu)^2+O(\norm{w}_2^3),
  \label{eq:mv-local-objective}
\end{align}
with constants independent of $\mu$ throughout a sufficiently small neighborhood of zero.
\end{proposition}

The first two nonconstant terms in \eqref{eq:mv-local-objective} form the usual mean--variance objective, and at $w=O(\norm{\mu}_2)$ the squared-mean term is $O(\norm{\mu}_2^4)$. Regret is controlled by the expansion of the optimal weights in \eqref{eq:mv-oracle-localization}.

\paragraph{Attainment of the same constants.}
We let the means shrink to zero with $n$ at a rate that keeps this approximation accurate. Let $r_{\mu,n}>0$ satisfy
\begin{align}
  r_{\mu,n}\longrightarrow0,\qquad
  \sqrt n\,r_{\mu,n}\longrightarrow\infty,\qquad
  \sqrt n\,r_{\mu,n}^2\longrightarrow0,
  \label{eq:mv-joint-limit}
\end{align}
and set $\Theta_n=\{\mu\in\bbR^K:\norm{\mu}_2\leq r_{\mu,n}\}$. By the second condition, the means in $\Theta_n$ range over a region much larger than the $n^{-1/2}$ error of the estimated mean, and by the third, the $O(r_{\mu,n}^2)$ gap between $w^\star(\mu)$ and $A^{-1}\Sigma_\mu^{-1}\mu$ is smaller than that error.

Let $\bar R_n=n^{-1}\sum_{i=1}^nR_i$ and $\widehat\Sigma_n=n^{-1}\sum_{i=1}^n(R_i-\bar R_n)(R_i-\bar R_n)^\top$. Define the sample mean--variance portfolio by
\begin{align}
  \widehat w_n^{\mathrm{MV}}
  &\in\argmax_{w\in\calW_K}\widehat M_n(w),
  &\widehat M_n(w)
  &=U(x_0)+u_1w^\top\bar R_n-\frac b2w^\top\widehat\Sigma_n w.
  \label{eq:sample-mean-variance-portfolio}
\end{align}
With probability tending to one, the maximizer is the interior portfolio $A^{-1}\widehat\Sigma_n^{-1}\bar R_n$, whose composition $\widehat\Sigma_n^{-1}\bar R_n$ does not involve the utility, which affects the portfolio only through the scale $A^{-1}$.

The one-class EUR rule computes the maximum likelihood estimate $\widehat\mu(D)$ over a sufficiently small compact mean-parameter set that contains zero in its interior and does not depend on $n$, and outputs $w^\star(\widehat\mu(D))$. Because the family is indexed by its mean, $\widehat\mu(D)=\bar R_n$ whenever $\bar R_n$ lies in the interior of that set (Appendix~\ref{app:local-mv}).

Both rules attain the constant
\begin{align}
  C_0=\frac{Ku_1^2}{2b},
  \label{eq:mv-constant}
\end{align}
as the next theorem shows.

For Bayes evaluation, let $f$ be a $C^1$ probability density on the unit ball in $\bbR^K$ that does not depend on $n$ and is positive in the interior and zero on the boundary. We require the following matrix to be finite:
\begin{align}
  \mathcal I(f)
  =\int_{\norm{v}_2<1}
       \frac{\nabla f(v)\nabla f(v)^\top}{f(v)}\,\rmd v.
  \label{eq:mv-prior-information}
\end{align}
The prior $H_n$ has density $h_n(\mu)=r_{\mu,n}^{-K}f(\mu/r_{\mu,n})$ on $\Theta_n$, so $\mu/r_{\mu,n}$ has density $f$ at every sample size while $H_n$ concentrates at zero. Define $\mathcal V_{n,\mathrm{loc}}^{\mathrm{mm}}=\inf_{\delta_n\in\mathfrak D_{n,K}}\sup_{\mu\in\Theta_n}\calR_n(\delta_n,P_\mu)$ and $\mathcal V_{n,\mathrm{loc}}^{H_n}=\inf_{\delta_n\in\mathfrak D_{n,K}}\int_{\Theta_n}\calR_n(\delta_n,P_\mu)\,\rmd H_n(\mu)$, with infima over all measurable randomized portfolio choice rules.

\begin{theorem}[Local minimax and Bayes limits of the sample mean--variance portfolio]
\label{thm:local-mean-variance}
Under \eqref{eq:mv-mean-family}, Assumption~\ref{ass:mv-local-utility}, and \eqref{eq:mv-joint-limit}, let $\delta_n^{\mathrm{MV}}$ be the rule in \eqref{eq:sample-mean-variance-portfolio}. Then, for every density $f$ satisfying the conditions above, the constant $C_0$ in \eqref{eq:mv-constant} satisfies
\begin{align}
  \lim_{n\to\infty}n\sup_{\mu\in\Theta_n}
        \calR_n(\delta_n^{\mathrm{MV}},P_\mu)
  &=\lim_{n\to\infty}n\mathcal V_{n,\mathrm{loc}}^{\mathrm{mm}}=C_0,
  \label{eq:local-mv-minimax}\\
  \lim_{n\to\infty}n\int_{\Theta_n}
        \calR_n(\delta_n^{\mathrm{MV}},P_\mu)\,\rmd H_n(\mu)
  &=\lim_{n\to\infty}n\mathcal V_{n,\mathrm{loc}}^{H_n}=C_0.
  \label{eq:local-mv-bayes}
\end{align}
The one-class EUR rule has the same two limits. Its output $\widehat w_n^\star$ also satisfies
\begin{align}
  \sup_{\mu\in\Theta_n}
      n\E_\mu\left[\norm{\widehat w_n^\star-\widehat w_n^{\mathrm{MV}}}_2^2\right]
  \longrightarrow0.
  \label{eq:eurr-mv-equivalence}
\end{align}
Neither rule uses $f$ or $H_n$.
\end{theorem}

In the proof (Appendix~\ref{app:local-mv}), the mean--variance approximation changes expected regret by $O(r_{\mu,n}^4)=o(n^{-1})$, and an integration-by-parts inequality for the joint likelihood and prior gives the lower bound $C_0$ for every rule. Since $\sqrt n\,r_{\mu,n}\to\infty$, the prior contribution $r_{\mu,n}^{-2}\mathcal I(f)$ to the second-moment matrix $\mathcal I_n$ of the joint score in \eqref{eq:mv-joint-information} vanishes after division by $n$, so $C_0$ is the same for every admissible $f$, unlike the local problem of \citet[Sections~2.2.1 and~5]{Adusumilli2023risk}.

\subsection{Higher Moments in the Portfolio Rule}
Whether a rule must use a moment beyond the second depends on how fast the expected excess returns shrink to zero with $n$. We show this for the third moment and then consider return distributions whose lower moments match exactly.
\paragraph{The third moment at a slower rate.}
Continue to use the exponential family \eqref{eq:mv-mean-family} and Assumption~\ref{ass:mv-local-utility}, and suppose additionally that $U\in C^4(J)$. For $v\in\bbR^K$, define
\begin{align}
  \beta(v)&=\frac{U'''(x_0)}{2b}\Sigma_0^{-1}
       \E_0\left[R\left(A^{-1}R^\top\Sigma_0^{-1}v\right)^2\right],
  &M_3(v)&=\frac b2 \beta(v)^\top\Sigma_0\beta(v).
  \label{eq:third-moment-displacement}
\end{align}
The vector $\beta(v)$ depends on the third moments of the baseline returns, and the population-optimal weights satisfy
\begin{align}
  w^\star(\mu)
  =A^{-1}\Sigma_\mu^{-1}\mu+\beta(\mu)+O(\norm{\mu}_2^3)
  \quad\text{as }\mu\to0.
  \label{eq:third-moment-oracle}
\end{align}
Since $\beta$ is quadratic, a mean of order $n^{-1/4}$ makes this displacement of order $n^{-1/2}$, the same order as the error in estimated weights.

\begin{theorem}[Third moments in the leading expected regret]
\label{thm:third-moment-regret}
Under the conditions just stated, let $\Theta_n=\{\mu\in\bbR^K:\norm{\mu}_2\leq n^{-1/4}\}$ and use the two rules defined in Section~\ref{sec:distribution}: the sample mean--variance portfolio and the one-class EUR rule. Then, as $n\to\infty$, we have
\begin{align}
  \sup_{\norm{v}_2\leq1}
  \left|n\calR_n(\delta_n^{\mathrm{MV}},P_{n^{-1/4}v})-C_0-M_3(v)\right|
  &\longrightarrow0,
  \label{eq:third-moment-mv-regret}\\
  \sup_{\norm{v}_2\leq1}
  \left|n\calR_n(\delta_n^\star,P_{n^{-1/4}v})-C_0\right|
  &\longrightarrow0.
  \label{eq:third-moment-eurr-regret}
\end{align}
Let $f$ satisfy the density conditions in Section~\ref{sec:distribution}, and let $H_n$ have density $h_n(\mu)=n^{K/4}f(n^{1/4}\mu)$ on $\Theta_n$. For this evaluation set and prior, both $n\mathcal V_{n,\mathrm{loc}}^{\mathrm{mm}}$ and $n\mathcal V_{n,\mathrm{loc}}^{H_n}$ converge to $C_0$. Consequently, the sample mean--variance portfolio has the additional worst-case constant $\max_{\norm{v}_2\leq1}M_3(v)$ and the additional Bayes constant $\int_{\norm{v}_2\leq1}M_3(v)f(v)\,\rmd v$.
\end{theorem}

The proof is in Appendix~\ref{app:critical-moments}. If $\beta$ is nonzero in some direction, both additional constants are positive because $f$ is positive inside the unit ball.

For example, take $K=1$, $x_0=1$, $U(x)=\log x$, and baseline returns $R=-1$ and $R=2$ with probabilities $2/3$ and $1/3$. Then $P_\mu(R=2)=(1+\mu)/3$ and $\Sigma_\mu=2+\mu-\mu^2$. On a feasible interval containing a neighborhood of cash, the oracle is $w^\star(\mu)=\mu/2$, while the population mean--variance choice is $\mu/(2+\mu-\mu^2)$ for small $\mu$. Thus, $\beta(v)=v^2/4$, $M_3(v)=v^4/16$, and $C_0=1/2$, although the return family has only two support points.

\paragraph{Portfolio preferences beyond a local approximation.}
Beyond special cases such as quadratic utility, or exponential utility with Gaussian returns, higher moments can change which portfolio is optimal even when means and covariances agree.

\begin{proposition}[Finite-support moment separation]
\label{prop:moment-separation}
Let $L\geq1$, and suppose that $U$ is continuous on an interval $J$ but is not a polynomial of degree at most $L$ on $J$. Then, there are probability distributions $Q_+$ and $Q_-$ supported on the same finite subset of $J$ such that
\begin{align}
  \E_{Q_+}[X^k]
  =
  \E_{Q_-}[X^k]
  \qquad
  \text{for }k=0,1,\ldots,L,
  \label{eq:matching-moments}
\end{align}
but
\begin{align}
  \E_{Q_+}[U(X)]
  \neq
  \E_{Q_-}[U(X)].
  \label{eq:utility-separation}
\end{align}
For $L=2$, the two distributions have the same mean and variance.
\end{proposition}

Taking $L=2$ shows that, for such an investor, the mean and covariance of returns do not suffice to rank portfolios.

\begin{corollary}[Means and covariances do not determine the oracle]
\label{cor:moment-insufficiency}
Suppose that $U$ is not a polynomial of degree at most two on the wealth interval under consideration. There are two bounded finite-support static return models and a two-asset feasible set $\calW_2=\calW_{2,1}\cup\calW_{2,2}$, where
\begin{align}
  \calW_{2,a}
  =
  \left\{
    w\in\Delta_2:\norm{w-e_a}_2\leq\epsilon
  \right\}
  \label{eq:moment-separation-feasible-set}
\end{align}
for the pure-asset portfolios $e_1$ and $e_2$ and a sufficiently small $\epsilon>0$. The two models have the same mean vector and covariance matrix, but the oracle belongs to $\calW_{2,1}$ under one model and to $\calW_{2,2}$ under the other. Call a map from the mean vector and covariance matrix of returns to a probability distribution on $\calW_2$ a population decision functional of those two moments. Any such functional makes the same randomized decision under both models and therefore has a positive expected utility gap under at least one. Consequently, a data-dependent rule that converges to such a functional has nonvanishing asymptotic regret.
\end{corollary}

Section~\ref{sec:experiments} gives a four-point example with constant relative risk aversion (CRRA) utility in which the empirical second-order criterion of Appendix~\ref{app:taylor} has a positive limiting selection loss.
\subsection{Portfolios from Symmetry}
\label{sec:risk-parity}
Unlike the route through a small optimal risky investment, this route does not require expected excess returns to shrink with $n$. When standardized returns and feasible standardized holdings are invariant under a transitive group of asset permutations, concavity makes equal standardized holdings optimal, and these holdings contribute equally to portfolio volatility.

\paragraph{Standardized returns and feasible investments.}
As in Section~\ref{sec:distribution}, $R$ collects excess returns over the cash asset and $w_a$ is the amount invested in risky asset $a$. Let $D_s=\operatorname{diag}(s_1,\ldots,s_K)$ for known scales $s_a>0$ that measure the relative volatilities of the assets, and write $Z=D_s^{-1}R$. The distribution of $Z$ is indexed by the unknown parameter $\theta$.

We consider one portfolio class under Assumption~\ref{ass:regular-return-model}. Let $\mathcal G$ be a group of permutation matrices acting transitively on the assets: for every pair of assets $a,a'\in[K]$, some matrix in $\mathcal G$ maps coordinate $a$ to coordinate $a'$. Suppose that, for every $\theta\in\Theta^+$ and $\Pi\in\mathcal G$, $\Pi Z$ and $Z$ have the same distribution under $P_\theta$. For example, the cyclic permutations form such a group. Because the group is transitive, $\E_\theta[Z]=m_\theta\mathbf1$ for a scalar $m_\theta$, and the standardized variances are equal.

Let $\mathcal C=D_s\calW_K$ be the set of feasible standardized holdings. We assume that $\mathcal C\subset\bbR_+^K$ is a full-dimensional compact convex set containing zero and satisfying $\Pi\mathcal C=\mathcal C$ for every $\Pi\in\mathcal G$, such as the box $[0,L]^K$.

Define $\Sigma_\theta=\operatorname{Cov}_\theta(R)$ and $\Sigma_\theta^Z=D_s^{-1}\Sigma_\theta D_s^{-1}$, and assume $\Sigma_\theta^Z\succ0$. For a nonzero portfolio $w$, its volatility and the contribution of asset $a$ are
\begin{align}
  \sigma_\theta(w)&=(w^\top\Sigma_\theta w)^{1/2},
  &\operatorname{RC}_{a,\theta}(w)
  &=w_a\frac{\partial\sigma_\theta(w)}{\partial w_a}
    =\frac{w_a(\Sigma_\theta w)_a}{\sigma_\theta(w)}.
  \label{eq:rp-contributions}
\end{align}
These contributions sum to $\sigma_\theta(w)$, and a positive portfolio has equal risk contributions when each of them equals $\sigma_\theta(w)/K$, the usual volatility-based definition \citep{Maillard2010properties,Cetingoz2024riskbudgeting}.

\paragraph{The risk-parity form of the EUR rule.}
By concavity, averaging feasible standardized holdings over permutations cannot reduce expected utility, so the optimum has equal standardized holdings $\alpha\mathbf1$. Define $T=\mathbf1^\top Z$, $\mathcal A_s=\{\alpha\in\bbR:\alpha\mathbf1\in\mathcal C\}$, and the scalar problem
\begin{align}
  F_\theta(\alpha)&=\E_\theta[U(x_0+\alpha T)],
  &\alpha_\theta&\in\argmax_{\alpha\in\mathcal A_s}F_\theta(\alpha).
  \label{eq:rp-scalar-problem}
\end{align}

\begin{proposition}[Risk-parity representation of the EUR rule]
\label{prop:eurr-risk-parity}
Under the conditions of this section, the maximizer $\alpha_\theta$ in \eqref{eq:rp-scalar-problem} is unique and positive. The population-optimal portfolio and the output of the one-class EUR rule for every $n\geq2$ satisfy
\begin{align}
  w_1(\theta)&=\alpha_\theta D_s^{-1}\mathbf1,
  &\widehat w_n^\star&=\alpha_{\widehat\theta(D)}D_s^{-1}\mathbf1,
  \label{eq:rp-eurr-identity}
\end{align}
where $D=\{R_1,\ldots,R_n\}$ and $\widehat\theta(D)$ are those of Section~\ref{subsec:eurr-one}. Both portfolios have equal contributions to volatility under every distribution in the model. The composition of their investment in risky assets is
\begin{align}
  q_a=\frac{s_a^{-1}}{\sum_{a'=1}^K s_{a'}^{-1}},\qquad a\in[K].
  \label{eq:rp-capital-composition}
\end{align}
\end{proposition}

The proof is in Appendix~\ref{app:portfolio-restrictions}. If $v_\theta$ is the common diagonal entry of $\Sigma_\theta^Z$, the volatility of asset $a$ is $s_a\sqrt{v_\theta}$, so \eqref{eq:rp-capital-composition} is inverse-volatility weighting. Equal contributions follow from the equal row sums of $\Sigma_\theta^Z$, so the pairwise correlations need not be equal. The EUR rule invests $\alpha_{\widehat\theta(D)}\sum_as_a^{-1}$ in risky assets and holds the rest in cash. The utility affects only this amount, through \eqref{eq:rp-scalar-problem}, whereas the composition \eqref{eq:rp-capital-composition} depends only on the known relative scales.

\paragraph{Implications for optimal portfolio choice.}
The minimax and Bayes conclusions follow from the general results for the EUR rule. Put $j_\theta=-F_\theta''(\alpha_\theta)>0$. Differentiating the first-order condition of \eqref{eq:rp-scalar-problem} gives
\begin{align}
  \nabla_\theta \alpha_\theta
    &=j_\theta^{-1}\E_\theta[U'(x_0+\alpha_\theta T)T S_\theta(R)],
  &c_1(\theta)
    &=\frac{j_\theta}{2}
       (\nabla_\theta \alpha_\theta)^\top I_\theta^{-1}\nabla_\theta \alpha_\theta.
  \label{eq:rp-general-constant}
\end{align}
This is \eqref{eq:likelihood-constants} with $\dot w_1(\theta)=D_s^{-1}\mathbf1(\nabla_\theta \alpha_\theta)^\top$, since the utility curvature along $D_s^{-1}\mathbf1$ is $j_\theta$. With one class, Theorems~\ref{thm:likelihood-global} and~\ref{thm:likelihood-bayes} yield
\begin{align}
  \lim_{n\to\infty}n\sup_{\theta\in\Theta}\calR_n(\delta_n^\star,P_\theta)
    &=\lim_{n\to\infty}n\mathcal V_{n,K}^{\mathrm{mm}}
      =\sup_{\theta\in\Theta}c_1(\theta),
  \label{eq:rp-inherited-minimax}\\
  \lim_{n\to\infty}n\int_\Theta\calR_n(\delta_n^\star,P_\theta)\,\rmd H(\theta)
    &=\lim_{n\to\infty}n\mathcal V_{n,K}^H
      =\int_\Theta c_1(\theta)\,\rmd H(\theta).
  \label{eq:rp-inherited-bayes}
\end{align}
Here $H\in\mathfrak H$ does not depend on $n$, and the infima are over all feasible measurable randomized portfolio choice rules, including those with unequal volatility contributions.

The diversification argument is related to \citet{Samuelson1967diversification}. Other rationalizations of risk parity use maximin, ambiguity, or diversification constraints \citep{Fisher2015riskparity,Costa2022datadriven,Gava2022alpha}, or a mean--variance objective \citep[Theorem 11.4]{Noguer2026heuristic}, whereas we obtain it exactly from a sample-based expected-utility rule.

\subsection{The Cost of Imposing a Simple Form}
We quantify the expected utility cost of three departures from the conditions under which the EUR rule takes a simple composition: a local loss of symmetry, unknown relative volatilities, and a binding money budget.

\paragraph{Departures from symmetry.}
\label{sec:symmetry-departures}
To measure what the symmetric form of the EUR rule costs when assets differ, we compare the precision gained by keeping its composition with the utility lost by omitting a direction of investment.

Consider one full-dimensional class satisfying Assumption~\ref{ass:regular-return-model} or~\ref{ass:regular-one-class}. At an interior point $\theta_0$ of the evaluation set, suppose $w_1(\theta_0)=\alpha_0v$ for some $\alpha_0>0$ and a specified positive vector $v$, such as $v=D_s^{-1}\mathbf1$ in the symmetric model above. Choose a compact interval $\mathcal A_v$ on which $\alpha v\in\calW_K$, with $\alpha_0$ in its interior, and define
\begin{align}
 \alpha^v(\theta)&\in\argmax_{\alpha\in\mathcal A_v}V_{P_\theta}(\alpha v),
 &\widehat w_n^v&=\alpha^v(\widehat\theta(D))v.
 \label{eq:fixed-composition-fit}
\end{align}
This allocation fixes the composition at $v$ and re-estimates the amount invested. Write $J_{\theta_0}=-\nabla^2V_{P_{\theta_0}}(\alpha_0v)$ and $D_{\theta_0}=\dot w_1(\theta_0)$, and define the projection in the utility-curvature inner product by
\begin{align}
 \Pi_v&=v(v^\top J_{\theta_0}v)^{-1}v^\top J_{\theta_0},
 &D_\parallel&=\Pi_vD_{\theta_0},
 &D_\perp&=(I_K-\Pi_v)D_{\theta_0}.
 \label{eq:composition-projection}
\end{align}
The subscript $\parallel$ marks the changes in the optimal weights along $v$, and $\perp$ marks the changes that no adjustment of the total investment can reproduce. Put
\begin{align}
 c_\parallel&=\frac12\operatorname{tr}(J_{\theta_0}D_\parallel I_{\theta_0}^{-1}D_\parallel^\top),
 &c_\perp&=\frac12\operatorname{tr}(J_{\theta_0}D_\perp I_{\theta_0}^{-1}D_\perp^\top),
 &L(h)&=\frac12h^\top D_\perp^\top J_{\theta_0}D_\perp h.
 \label{eq:composition-loss-constants}
\end{align}

\begin{proposition}[Precision and departures from a common composition]
\label{prop:local-composition}
For the model and allocations just defined, let $\theta_{n,h}=\theta_0+h/\sqrt n$. For every compact $\mathcal K\subset\bbR^p$, as $n\to\infty$, we have
\begin{align}
 \sup_{h\in\mathcal K}\left|n\calR_n(\delta_n^\star,P_{\theta_{n,h}})
                    -c_\parallel-c_\perp\right|&\longrightarrow0,
 \label{eq:composition-full-regret}\\
 \sup_{h\in\mathcal K}\left|n\E_{\theta_{n,h}}[r_{P_{\theta_{n,h}}}(\widehat w_n^v)]
                    -c_\parallel-L(h)\right|&\longrightarrow0.
 \label{eq:composition-restricted-regret}
\end{align}
The assets, return family, utility, feasible set, and direction $v$ do not change with $n$. For a deterministic perturbation $\theta=\theta_0+t h$ with $t\to0$, the loss of the best portfolio with composition $v$ satisfies
\begin{align}
 V_{P_{\theta_0+th}}(w_1(\theta_0+th))
 -V_{P_{\theta_0+th}}(\alpha^v(\theta_0+th)v)
 =t^2L(h)+o(t^2).
 \label{eq:composition-population-cost}
\end{align}
\end{proposition}

The proof is in Appendix~\ref{app:portfolio-restrictions}. Keeping the composition removes $c_\perp$ from the estimation constant but adds $L(h)$, so $\widehat w_n^v$ has the smaller local expected regret when $L(h)<c_\perp$ and the larger one when $L(h)>c_\perp$. This trade-off is related to the finding of \citet{Jagannathan2003riskreduction} that portfolio restrictions can reduce estimation error.

\paragraph{Estimated relative volatility.}
\label{sec:estimated-volatility}
When the relative volatilities are unknown, the EUR rule estimates the composition as well. Let $s_a(\theta)>0$ be twice continuously differentiable and put $D_s(\theta)=\operatorname{diag}(s_1(\theta),\ldots,s_K(\theta))$. Suppose that, for every $\theta\in\Theta^+$, the distribution of $D_s(\theta)^{-1}R$ under $P_\theta$ is invariant under the same transitive permutation group and has a finite covariance that is continuous in $\theta$ and positive definite throughout $\Theta^+$. The feasible set $\calW_K$ is known and independent of the unknown scales.

We use Assumption~\ref{ass:regular-one-class}. Define $T_\theta=\mathbf1^\top D_s(\theta)^{-1}R$, and suppose the equation $\E_\theta[U'(x_0+\alpha T_\theta)T_\theta]=0$ has a positive solution $\alpha_\theta$ for which $\alpha_\theta D_s(\theta)^{-1}\mathbf1$ lies in the interior of $\calW_K$ throughout $\Theta^+$.

\begin{proposition}[The EUR rule with estimated asset scales]
\label{prop:eurr-estimated-scales}
Under the conditions just stated, the population optimum and the same one-class EUR rule satisfy
\begin{align}
 w_1(\theta)&=\alpha_\theta D_s(\theta)^{-1}\mathbf1,
 &\widehat w_n^\star&=\alpha_{\widehat\theta(D)}
                       D_s(\widehat\theta(D))^{-1}\mathbf1.
 \label{eq:estimated-scale-eurr}
\end{align}
Let $G_s(\theta)$ be the $K\times p$ matrix whose row $a$ is $\nabla_\theta\log s_a(\theta)^\top$. The derivative used in the general allocation constant is
\begin{align}
 \dot w_1(\theta)
 =D_s(\theta)^{-1}
   \left(\mathbf1\nabla_\theta \alpha_\theta^\top-\alpha_\theta G_s(\theta)\right).
 \label{eq:estimated-scale-jacobian}
\end{align}
The minimax and Bayes conclusions of Theorems~\ref{thm:likelihood-global} and~\ref{thm:likelihood-bayes} apply with this derivative in \eqref{eq:likelihood-constants}.
\end{proposition}

At equal standardized holdings, symmetry makes the derivatives of expected utility with respect to standardized holdings equal and the scalar first-order condition makes their sum zero, so strict concavity gives \eqref{eq:estimated-scale-eurr}, and differentiation gives \eqref{eq:estimated-scale-jacobian}. The fitted portfolio has equal contributions under its fitted covariance.

\paragraph{A return family with unknown asset scales.}
An explicit non-Gaussian family shows that these conditions describe an estimable problem. Let $R_a=s_a S_aX_a$, where the signs $S_a\in\{-1,1\}$ are independent with $\Pr(S_a=1)=\pi>1/2$, the magnitudes $X_a$ are independent gamma variables with shape $k_G>0$ and rate one, and signs and magnitudes are independent. The shape parameter is known, and $\theta=(\pi,\log s_1,\ldots,\log s_K)$ is unknown.

For a positive risk-aversion coefficient $\gamma$, use $U(x)=(1-\exp(-\gamma(x-x_0)))/\gamma$ and a known box $\calW_K=[0,L]^K$. Let the estimation ranges be $\pi\in[\pi_-,\pi_+]\subset(1/2,1)$ and $s_a\in[s_-,s_+]\subset(0,\infty)$. Assume $\gamma Ls_+<1$ and $\alpha_{\pi_+}/s_-<L$, where
\begin{align}
 \alpha_\pi&=\frac1\gamma\tanh\left(\frac{\log(\pi/(1-\pi))}{2(k_G+1)}\right),
 &\ell_\pi(\alpha)&=\pi(1+\gamma \alpha)^{-k_G}+(1-\pi)(1-\gamma \alpha)^{-k_G}.
 \label{eq:gamma-investment}
\end{align}
The gamma moment generating function shows that expected utility is finite on this box. Independence gives $V_{P_\theta}(w)=(1-\prod_a\ell_\pi(s_aw_a))/\gamma$, and its unique optimum is $w_a=\alpha_\pi/s_a$.

For $n$ observations, the maximum likelihood estimates over the compact parameter ranges are
\begin{align}
 \widehat\pi_n
 &=\operatorname{proj}_{[\pi_-,\pi_+]}
       \left(\frac1{nK}\sum_{i=1}^n\sum_{a=1}^K\mathbf1\{R_{ia}>0\}\right),
 &\widehat s_{a,n}
 &=\operatorname{proj}_{[s_-,s_+]}
       \left(\frac1{nk_G}\sum_{i=1}^n|R_{ia}|\right),
 \label{eq:gamma-likelihood-estimates}
\end{align}
where $\operatorname{proj}_I$ denotes projection onto an interval $I$. The number of positive returns over all assets and, for each asset, the sum of absolute returns are sufficient statistics, and the information matrix is $\operatorname{diag}(K/(\pi(1-\pi)),k_GI_K)$ in the coordinates of $\theta$. With $j_\pi=\ell_\pi''(\alpha_\pi)\ell_\pi(\alpha_\pi)^{K-1}/\gamma$, substitution into the general constant yields
\begin{align}
 c_1(\theta)
 =\frac{j_\pi}{2}
   \left(\pi(1-\pi)(\alpha_\pi')^2+\frac{K \alpha_\pi^2}{k_G}\right).
 \label{eq:gamma-allocation-constant}
\end{align}
The first term is the cost of estimating the standardized investment amount with known scales. The second, the cost of estimating the scales, splits into $j_\pi \alpha_\pi^2/(2k_G)$ from a common proportional change in holdings and $j_\pi(K-1)\alpha_\pi^2/(2k_G)$ from changes in composition. Appendix~\ref{app:regular-one-class} verifies the likelihood conditions and this calculation, with the evaluation set inside the estimation ranges and a prior independent of $n$.

\paragraph{Money budgets.}
\label{sec:rp-budget}
A money budget can change the composition even when standardized returns are symmetric. Assume again that $D_s$ is known, and write $y=D_sw$, $c=D_s^{-1}\mathbf1\in\bbR^K$, and $\phi_\theta(y)=\E_\theta[U(x_0+y^\top Z)]$. A known upper bound $\mathcal B$ on money invested in risky assets imposes $c^\top y\leq\mathcal B$.

Let $y_0=\alpha_\theta\mathbf1$ be the interior optimum without the money budget, and let $\mathcal B_0=c^\top y_0$. For a budget $\mathcal B$, the feasible set is the original known set intersected with $\{y:c^\top y\leq\mathcal B\}$, and $y_{\mathcal B}$ is the resulting optimum.

\begin{proposition}[The effect of a money budget]
\label{prop:rp-budget}
Under the transitive symmetry conditions of Section~\ref{sec:risk-parity}, suppose the standardized expected utility is three times continuously differentiable and strictly concave near $y_0$, with $J_\phi=-\nabla^2\phi_\theta(y_0)\succ0$. If $\mathcal B\geq\mathcal B_0$, then $y_{\mathcal B}=y_0$. If the budget has a strictly positive Lagrange multiplier and all other constraints are inactive at a positive optimum, equal standardized holdings are possible only when the entries of $c$ are equal.

For $\mathcal B=\mathcal B_0-\Delta$ with $\Delta\downarrow0$, all other constraints remain inactive, and we have
\begin{align}
 y_{\mathcal B}&=y_0-\Delta\frac{J_\phi^{-1}c}{c^\top J_\phi^{-1}c}+O(\Delta^2),
 &y_{\mathcal B}^v&=\left(\alpha_\theta-\frac{\Delta}{c^\top\mathbf1}\right)\mathbf1,
 \label{eq:budget-allocation-expansion}
\end{align}
where $y_{\mathcal B}^v$ is the best feasible allocation that imposes equal standardized holdings. Their expected-utility difference is
\begin{align}
 \phi_\theta(y_{\mathcal B})-\phi_\theta(y_{\mathcal B}^v)
 &=\frac{\Delta^2}{2}
   \left(\frac{\mathbf1^\top J_\phi\mathbf1}{(c^\top\mathbf1)^2}
                   -\frac1{c^\top J_\phi^{-1}c}\right)+O(\Delta^3).
 \label{eq:budget-composition-cost}
\end{align}
The coefficient is positive when relative volatilities differ.
\end{proposition}

Differentiating the constrained first-order conditions at $\mathcal B_0$ gives the expansion (Appendix~\ref{app:portfolio-restrictions}). When the budget binds while the other constraints are inactive, the investor responds by reallocating across assets instead of reducing each standardized holding by the same amount, and \eqref{eq:budget-composition-cost} measures the cost of keeping the symmetric composition.

\section{EUR Rule with State-Space Modeling}
\label{sec:dynamic}
When a portfolio depends on past returns, the fitting step of the EUR rule estimates an allocation function that maps an observed state to a portfolio. We then analyze the resulting problem of maximizing expected utility conditional on that state.

\subsection{State-Space Model}
The state records the information available before the next return is realized. Fix a positive integer history length $L$ and a model $\calP^{\mathrm{state}}$ of strictly stationary distributions of the return process $(R_t)_{t\in\mathbb Z}$, and assume that the sample includes the lags needed to form the state.

Let
\begin{align}
  S_t^{(L)}
  =
  \Psi_L(R_{t-1},\ldots,R_{t-L})
  \in[0,1]^p
  \label{eq:state}
\end{align}
be a state formed from the past $L$ returns, where $\Psi_L$ is a specified measurable map and $p$ is the state dimension. An allocation function is a measurable function
\begin{align}
  \pi:[0,1]^p\to\calW_K.
  \label{eq:allocation-function}
\end{align}
In Theorem~\ref{thm:state-rate}, $\Pi$ is the class of all measurable $\calW_K$-valued allocation functions, and we assume that measurable conditional maximizers exist.

The expected utility of an allocation function is
\begin{align}
  \mathcal V_P(\pi)
  =
  \E_P\left[
    U\left(x_0+\pi(S_t^{(L)})^\top R_t\right)
  \right].
  \label{eq:state-value}
\end{align}
For a class of allocation functions $\Pi$, define
\begin{align}
  \mathcal V_P^\star
  =
  \sup_{\pi\in\Pi}\mathcal V_P(\pi),
  \qquad
  r_P(\pi)
  =
  \mathcal V_P^\star-\mathcal V_P(\pi).
  \label{eq:state-regret}
\end{align}
The expectation in \eqref{eq:state-value} is taken over the stationary joint distribution of the state and the next return, and an estimated allocation function is held fixed in this expectation.

Let $\mathfrak P_n$ be the set of rules that map the observed time series to an allocation function $\widehat\pi_n\in\Pi$, and define
\begin{align}
  \mathcal V_n^{\mathrm{state}}
  =
  \inf_{\widehat\pi_n\in\mathfrak P_n}
  \sup_{P\in\calP^{\mathrm{state}}}
  \E_P\left[r_P(\widehat\pi_n)\right].
  \label{eq:state-minimax-value}
\end{align}

\subsection{Conditional Portfolio Estimation}
Because a state-dependent portfolio is chosen from the expected utility conditional on the observed state, we apply the relation between estimation error and regret from Section~\ref{sec:complexity} to this conditional expected utility. Fix a norm $\norm{\cdot}$ on $\bbR^K$ with dual norm $\norm{\cdot}_*$, and define the conditional expected utility and its gradient by
\begin{align}
  Q_P(s,w)
  &=
  \E_P\left[
    U\left(x_0+w^\top R_t\right)
    \mathrel{\Big|}S_t^{(L)}=s
  \right],
  \nonumber\\
  g_P(s,w)
  &=
  \nabla_wQ_P(s,w)
  =
  \E_P\left[
    U'\left(x_0+w^\top R_t\right)R_t
    \mathrel{\Big|}S_t^{(L)}=s
  \right].
  \label{eq:conditional-gradient}
\end{align}

\begin{assumption}[Smoothness in the observed state and curvature of utility]
\label{ass:state-smoothness}
The state distribution has a density bounded above and away from zero on $[0,1]^p$. Fix $\beta>0$, let $r_\beta=\lceil\beta\rceil-1$, and put $\alpha_\beta=\beta-r_\beta\in(0,1]$. For every multi-index $\nu$ with $\abs{\nu}\leq r_\beta$, the derivative $D_s^\nu g_P(s,w)$ exists and is uniformly bounded. There is $C_\beta<\infty$ such that, whenever $\abs{\nu}=r_\beta$,
\begin{align}
  \norm{D_s^\nu g_P(s,w)-D_s^\nu g_P(s',w)}_*
  \leq
  C_\beta\norm{s-s'}_2^{\alpha_\beta}
  \label{eq:holder-gradient}
\end{align}
for every $P\in\calP^{\mathrm{state}}$, every $w\in\calW_K$, and all $s,s'\in[0,1]^p$. Each $P$ has an almost surely unique optimal allocation function $\pi_P^\star$. For fixed $q>1$ and $c_0>0$, the utility loss satisfies
\begin{align}
  \mathcal V_P^\star-\mathcal V_P(\pi)
  \geq
  c_0
  \E_P\left[
    \norm{\pi(S_t^{(L)})-\pi_P^\star(S_t^{(L)})}^q
  \right]
  \label{eq:state-utility-loss}
\end{align}
for every $\pi\in\Pi$.
\end{assumption}

\paragraph{Estimation of the conditional gradient.}
For a bandwidth $h$, let $\widehat Q_h(s,w)$ be an estimated conditional objective, let $\widehat g_h(s,w)$ be its gradient in $w$, and let $\varepsilon_h(s)$ be the largest error of this gradient over $w\in\calW_K$ at state $s$, that is, $\varepsilon_h(s)=\sup_{w\in\calW_K}\norm{\widehat g_h(s,w)-g_P(s,w)}_*$. Let $P_S$ be the stationary distribution of $S_t^{(L)}$ and $p_q=q/(q-1)$. The next assumption bounds this error after integrating over $P_S$ rather than taking its supremum over states, thereby avoiding the logarithmic factor that a supremum-norm bound adds for H\"older regression classes \citep{Chen2015optimaluniform}.

\begin{assumption}[Accuracy of the estimated conditional gradient]
\label{ass:state-gradient-estimation}
There is a deterministic effective sample size $n_{\mathrm{eff}}^{\mathrm{up}}=n_{\mathrm{eff}}^{\mathrm{up}}(n)\to\infty$ with $n_{\mathrm{eff}}^{\mathrm{up}}\leq n$ and deterministic bandwidth sets $\mathcal H_n\subset(0,1]$. The bandwidth sets satisfy $\sup_{h\in\mathcal H_n}h\to0$ and $\inf_{h\in\mathcal H_n}n_{\mathrm{eff}}^{\mathrm{up}}h^p\to\infty$, and they contain bandwidths $h_n\asymp(n_{\mathrm{eff}}^{\mathrm{up}})^{-1/(2\beta+p)}$. For every $h\in\mathcal H_n$, the estimated conditional objective $\widehat Q_h(s,w)$ is concave in $w$, and, for some moment order $r\geq p_q$, the following bound on its gradient error holds
\begin{align}
  \sup_{P\in\calP^{\mathrm{state}}}
  \E_P\left[
    \left(\int_{[0,1]^p}\varepsilon_h(s)^{p_q}\,P_S(\rmd s)\right)^{r/p_q}
  \right]^{1/r}
  \leq
  C\left(
    h^\beta+
    \left(n_{\mathrm{eff}}^{\mathrm{up}}h^p\right)^{-1/2}
  \right)
  \label{eq:state-gradient-bound}
\end{align}
uniformly over the admissible bandwidths.
\end{assumption}

Given $\widehat Q_h$, the estimated allocation function satisfies $\widehat\pi_h(s)\in\argmax_{w\in\calW_K}\widehat Q_h(s,w)$ at each state $s$, with a measurable choice when there are multiple maximizers.

\paragraph{Lower bound and minimax rate.}
The next assumption specifies the alternative return distributions that give a matching lower bound.

\begin{assumption}[State-dependent lower-bound submodel]
\label{ass:state-lower}
There is a deterministic effective sample size $n_{\mathrm{eff}}^{\mathrm{low}}=n_{\mathrm{eff}}^{\mathrm{low}}(n)\to\infty$ for which $\calP^{\mathrm{state}}$ contains bounded conditional return models with finite conditional support and the following property. Signed perturbations of the probability mass on disjoint cells of width $h$ in the state space remain within the specified return model and change the gradient of the conditional expected utility by an amount $\gamma_h\asymp h^\beta$. For two models that differ in the sign on one cell, the Kullback--Leibler divergence between the distributions of the full observed time series is at most $C n_{\mathrm{eff}}^{\mathrm{low}}\gamma_h^2h^p$. Every decision that corresponds to the wrong sign on that cell incurs an expected utility regret of at least a positive constant times $\gamma_h^{q/(q-1)}h^p$.
\end{assumption}

\begin{theorem}[Minimax rate for a smooth state-dependent allocation function]
\label{thm:state-rate}
Suppose that $\calW_K$ is one of the closed convex classes in \eqref{eq:model-decomposition}. Under Assumptions \ref{ass:utility}, \ref{ass:state-smoothness}, and \ref{ass:state-gradient-estimation}, there is a rule for estimating the allocation function that satisfies
\begin{align}
  \sup_{P\in\calP^{\mathrm{state}}}
  \E_P\left[r_P(\widehat\pi_n)\right]
  \lesssim
  \inf_{h\in\mathcal H_n}
  \left(
    h^\beta+
    \left(n_{\mathrm{eff}}^{\mathrm{up}}h^p\right)^{-1/2}
  \right)^{q/(q-1)}.
  \label{eq:state-upper-rate}
\end{align}
The bandwidth $h_n$ that balances the two terms and the resulting order $\varepsilon_n^\nabla$ of the estimation error satisfy
\begin{align}
  h_n
  &\asymp
  \left(n_{\mathrm{eff}}^{\mathrm{up}}\right)^{-1/(2\beta+p)},
  \qquad
  \varepsilon_n^\nabla
  &\asymp
  \left(n_{\mathrm{eff}}^{\mathrm{up}}\right)^{-\beta/(2\beta+p)}.
  \label{eq:nonparametric-gradient-rate}
\end{align}
If Assumption~\ref{ass:state-lower} also holds and $n_{\mathrm{eff}}^{\mathrm{low}}\asymp n_{\mathrm{eff}}^{\mathrm{up}}$, write $n_{\mathrm{eff}}$ for their common order. Then, we have
\begin{align}
  \mathcal V_n^{\mathrm{state}}
  \asymp
  n_{\mathrm{eff}}^{-\beta q/((2\beta+p)(q-1))}.
  \label{eq:nonparametric-state-rate}
\end{align}
For $q=2$, the rate in \eqref{eq:nonparametric-state-rate} is $n_{\mathrm{eff}}^{-2\beta/(2\beta+p)}$.
\end{theorem}

The limit in Theorem~\ref{thm:state-rate} is taken as $n\to\infty$ with $K$, $p$, $L$, $\beta$, and $q$ fixed, and the rate of the expected regret is the estimation rate raised to the power $q/(q-1)$ because of the utility-loss condition \eqref{eq:state-utility-loss}.

\subsection{History Length and Approximation}
\label{subsec:data-selected}
We now balance the approximation error that comes from a finite history or from finitely many basis functions against the estimation error.

\paragraph{History length.}
Let $A_{\mathrm{mem}}(L)$ be the approximation regret from replacing the full return history with $L$ lags. For a finite set of history lengths $\mathcal L_n$, let $p(L)$, $n_{\mathrm{eff}}^{\mathrm{up}}(L)$, and $\mathcal H_n(L)$ denote the state dimension, effective sample size, and bandwidth set for history length $L$, and suppose that the constants in the estimation bounds are the same for every $L\in\mathcal L_n$. When $L$ and $h$ are specified before the data are observed, the best upper bound that results is
\begin{align}
  \inf_{\substack{L\in\mathcal L_n\\h\in\mathcal H_n(L)}}
  \left\{
    A_{\mathrm{mem}}(L)
    +
    \left(
      h^\beta+
      \left(n_{\mathrm{eff}}^{\mathrm{up}}(L)h^{p(L)}\right)^{-1/2}
    \right)^{q/(q-1)}
  \right\}.
  \label{eq:memory-optimization}
\end{align}
A longer history can reduce the first term but can raise $p(L)$ and hence the estimation term, while choosing $L$ from the data adds a selection error.

\paragraph{Approximation classes.}
The same trade-off arises when allocation functions are approximated with finitely many basis functions. For a fixed state, consider nested convex classes of allocation functions
\begin{align}
  \Pi_1\subseteq\Pi_2\subseteq\cdots\subseteq\Pi_\infty,
  \qquad
  \overline{\bigcup_{m\geq1}\Pi_m}=\Pi_\infty,
  \label{eq:nested-sieves}
\end{align}
where the closure is taken in a metric $d_\Pi$ such that $d_\Pi(\pi_j,\pi)\to0$ implies
\begin{align}
  \sup_{P\in\calP^{\mathrm{state}}}
  \abs{\mathcal V_P(\pi_j)-\mathcal V_P(\pi)}
  \to0.
  \label{eq:sieve-value-continuity}
\end{align}
The index $m$ measures the complexity of $\Pi_m$, for example the number of basis functions \citep{Birge1998minimumcontrast}.

Let
\begin{align}
  \pi_P^\star
  &\in
  \argmax_{\pi\in\Pi_\infty}\mathcal V_P(\pi),
  \qquad
  \pi_{P,m}^\star
  &\in
  \argmax_{\pi\in\Pi_m}\mathcal V_P(\pi),
  \label{eq:sieve-oracles}
\end{align}
and define the approximation regret
\begin{align}
  A_P(m)
  =
  \mathcal V_P(\pi_P^\star)-\mathcal V_P(\pi_{P,m}^\star).
  \label{eq:sieve-approximation-regret}
\end{align}
For each class $\Pi_m$, let $\widetilde\pi_m$ maximize over $\Pi_m$ an estimate of expected utility that is differentiable and concave. Suppose that, for every $\pi\in\Pi_m$, the utility loss relative to the oracle $\pi_{P,m}^\star$ is at least $c\norm{\pi-\pi_{P,m}^\star}_{L^q(P_S)}^q$, and that the error of the estimated directional derivatives is bounded in the dual $L^{q/(q-1)}(P_S)$ norm by a random variable whose $q/(q-1)$ moment is at most $v_{n,m}^{q/(q-1)}$. Applying the argument of Proposition~\ref{prop:within-class} at each state and integrating over the state then gives
\begin{align}
  \E_P\left[r_P(\widetilde\pi_m)\right]
  \leq
  A_P(m)+C v_{n,m}^{q/(q-1)}.
  \label{eq:sieve-upper-bound}
\end{align}
\paragraph{Choosing the complexity by validation.}
Because the best complexity is unknown, we select it among $1\leq m\leq\overline m_n$ by validation. We estimate $\widetilde\pi_m$ on $n_1$ training observations and estimate its expected utility by $\widehat{\mathcal V}^{\mathrm{val}}_m$ on $n_2$ independent validation observations. With a deterministic nonnegative allowance $p_{n_2,m}$ for the error of this estimate, we choose
\begin{align}
  \widehat m
  \in
  \argmax_{1\leq m\leq\overline m_n}
  \left(
    \widehat{\mathcal V}^{\mathrm{val}}_m-2p_{n_2,m}
  \right).
  \label{eq:sieve-selector}
\end{align}

\begin{proposition}[Approximation-class oracle inequality]
\label{prop:sieve-oracle}
Suppose that the following inequalities hold on an event $\mathcal E_n$:
\begin{align}
  \abs{
    \widehat{\mathcal V}^{\mathrm{val}}_m
    -
    \mathcal V_P(\widetilde\pi_m)
  }
  \leq
  p_{n_2,m}
  \qquad
  \text{for every }1\leq m\leq\overline m_n.
  \label{eq:sieve-validation-event}
\end{align}
Then, on $\mathcal E_n$, we have
\begin{align}
  r_P(\widetilde\pi_{\widehat m})
  \leq
  \inf_{1\leq m\leq\overline m_n}
  \left(
    A_P(m)
    +
    \mathcal V_P(\pi_{P,m}^\star)-\mathcal V_P(\widetilde\pi_m)
    +
    3p_{n_2,m}
  \right).
  \label{eq:sieve-oracle-inequality}
\end{align}
If $\underline U\leq U(x)\leq\overline U$ holds on the attainable wealth range, then taking expectations adds at most $(\overline U-\underline U)\Pr_P(\mathcal E_n^c)$ to the right-hand side of \eqref{eq:sieve-oracle-inequality}.
\end{proposition}

With bounded utility, if
\begin{align}
  \E_P\left[
    \mathcal V_P(\pi_{P,m}^\star)-\mathcal V_P(\widetilde\pi_m)
  \right]
  \leq
  C v_{n_1,m}^{q/(q-1)}
  \label{eq:within-sieve-regret-bound}
\end{align}
for every $m\leq\overline m_n$, then Proposition~\ref{prop:sieve-oracle} gives
\begin{align}
  \E_P\left[r_P(\widetilde\pi_{\widehat m})\right]
  \leq
  \inf_{1\leq m\leq\overline m_n}
  \left(
    A_P(m)
    +
    C v_{n_1,m}^{q/(q-1)}
    +
    3p_{n_2,m}
  \right)
  +
  (\overline U-\underline U)\Pr_P(\mathcal E_n^c).
  \label{eq:data-selected-sieve-regret}
\end{align}
In this bound, the approximation loss, the estimation error, and the selection error appear as separate terms.

When the gradient of the conditional expected utility is $\beta$-smooth, suppose uniformly over $P\in\calP^{\mathrm{state}}$ that
\begin{align}
  A_P(m)
  &\lesssim
  m^{-\beta q/(p(q-1))},
  \qquad
  v_{n,m}
  &\asymp
  \sqrt{\frac{m}{n_{\mathrm{eff}}}}.
  \label{eq:sieve-estimation-error}
\end{align}
Here, $n_{\mathrm{eff}}$ is the effective sample size in the estimation bound within each class $\Pi_m$. For $m$ specified in advance, the regret bound is of order
\begin{align}
  m^{-\beta q/(p(q-1))}
  +
  \left(\frac{m}{n_{\mathrm{eff}}}\right)^{q/(2(q-1))}.
  \label{eq:sieve-rate-objective}
\end{align}
Minimizing \eqref{eq:sieve-rate-objective} gives
\begin{align}
  m_n
  &\asymp
  n_{\mathrm{eff}}^{p/(2\beta+p)},
  \qquad
  \inf_{m\in\mathbb N}\sup_{P\in\calP^{\mathrm{state}}}
  \E_P\left[r_P(\widetilde\pi_m)\right]
  &\lesssim
  n_{\mathrm{eff}}^{-\beta q/((2\beta+p)(q-1))}.
  \label{eq:sieve-optimal-size}
\end{align}
Here $n_{\mathrm{eff}}(n)\to\infty$ while $K$, $p$, $\beta$, and $q$ stay fixed. The data-selected rule attains the same order provided that $m_n\leq\overline m_n$ and that $p_{n_2,m_n}$ and the failure term in \eqref{eq:data-selected-sieve-regret} are of no larger order.

\section{Experiments}
\label{sec:experiments}

We first compare the expected regrets of the EUR rule and the empirical utility maximizers of Section~\ref{sec:upper} with the theoretical constants in specified return models, and then compare portfolio objectives on industry returns.

\subsection{Simulation Studies}
Because the return models below have finite support or closed-form expected utility, their expected regret can be evaluated directly.

\paragraph{The EUR rule with class selection and weight estimation.}
We use the four-point model in \eqref{eq:joint-estimation-classes}--\eqref{eq:joint-estimation-allocation-constant}, with $x_0=1$, $k_n=\lfloor n^{4/5}\rfloor$, and the screening threshold $\sigma_Ub_n$ with $\sigma_U=1$. The EUR rule estimates $(\theta,\eta)$ separately from the preliminary sample and the second sample, selects a sign class, and then estimates the position in the third asset from $D_1$.

We take the largest expected regret over 21 equally spaced $\eta\in[-0.05,0.05]$ and 101 equally spaced $\theta\in[-0.25,0.25]$, together with $\theta=u/\sqrt{N_n}$ for 401 equally spaced $u\in[-4,4]$ near the tie between the two classes, keeping only the points inside the parameter set. For the Bayes evaluation, we use two priors: $H_1$ is uniform on the parameter rectangle, and $H_2$ has independent coordinates, with $\eta$ uniform and $\theta$ with density $15(1-(\theta/\bar\theta)^2)^2/(16\bar\theta)$ on $[-\bar\theta,\bar\theta]$, where $\bar\theta=0.25$.

\begin{table}[t]
\centering
\caption{Expected regret of the two-class EUR rule, normalized by $\sqrt n$ (total, selection) or $n$ (allocation, Bayes).}
\label{tab:eurr-binary}
\small
\begin{tabular}{rrrrrr}
\toprule
$n$ & $\sqrt n\,\mathcal R$ & $\sqrt n\,\mathcal R_{\rm sel}$ & $n\mathcal R_{\rm alloc}$ & $n\mathcal R_{H_1}$ & $n\mathcal R_{H_2}$ \\
\midrule
250 & 0.02348 & 0.01684 & 0.10491 & 0.21002 & 0.25757 \\
1,000 & 0.02095 & 0.01447 & 0.20511 & 0.27138 & 0.34676 \\
4,000 & 0.01752 & 0.01392 & 0.22801 & 0.29507 & 0.37310 \\
16,000 & 0.01540 & 0.01369 & 0.21598 & 0.29591 & 0.37166 \\
64,000 & 0.01409 & 0.01327 & 0.20752 & 0.28768 & 0.36089 \\
256,000 & 0.01351 & 0.01312 & 0.20155 & 0.27946 & 0.35065 \\
1,024,000 & 0.01315 & 0.01296 & 0.19724 & 0.27347 & 0.34317 \\
$n\to\infty$ & 0.01253 & 0.01253 & -- & 0.25630 & 0.32162 \\
\bottomrule
\end{tabular}
\begin{minipage}{0.95\textwidth}
\footnotesize
\vspace{0.4em}
\noindent The first three columns are evaluated at the grid parameter with the largest total expected regret, with the selection and allocation terms defined as in \eqref{eq:two-layer-decomposition}. The last two columns integrate the total expected regret under $H_1$ and $H_2$. The last row gives the theoretical constants.
\end{minipage}
\end{table}

As $n$ grows, Table~\ref{tab:eurr-binary} shows that weight estimation first makes a substantial contribution to the loss while class selection eventually dominates, and the largest normalized expected regret approaches $\Gamma_{\mathrm{lik}}=0.01253$. At $n=1{,}024{,}000$, the weights are estimated from the fraction $N_n/n=0.9372$ of the observations. Multiplying the two integrated regrets by $N_n$ instead of $n$ therefore gives $0.25630$ and $0.32162$, the limits in the last row, which shows the finite-sample cost of reserving observations for the preliminary comparison.

\paragraph{Three-class Gaussian comparison.}
We solve the Gaussian selection problem in \eqref{eq:likelihood-gaussian-value} for three classes with optimized expected utilities $(0,\mu_1,\mu_2)$, observation $Y\sim N(\mu,\Omega)$, and the four covariance matrices in Table~\ref{tab:gaussian-comparison}. The linear program \eqref{eq:likelihood-gaussian-program} minimizes the largest expected utility gap on $[-4,4]^2$ with mean and observation grids of spacing $0.5$ or $0.25$, and we compare each resulting selector with selecting the largest of $(0,Y_1,Y_2)$.

\begin{table}[t]
\centering
\caption{Expected regret in the three-class Gaussian selection problem.}
\label{tab:gaussian-comparison}
\small
\begin{tabular}{crrrr}
\toprule
$(\Omega_{11},\Omega_{12},\Omega_{22})$ & Grid $0.5$ & Grid $0.25$ & Largest observation & LP objective \\
\midrule
$(2,1,2)$ & 0.37748 & 0.37327 & 0.37215 & 0.37175 \\
$(1,0.8,1)$ & 0.23691 & 0.23207 & 0.23069 & 0.23024 \\
$(1,0,4)$ & 0.45896 & 0.45324 & 0.46145 & 0.45202 \\
$(1,-0.5,2)$ & 0.38741 & 0.37940 & 0.39103 & 0.37743 \\
\bottomrule
\end{tabular}
\begin{minipage}{0.96\textwidth}
\footnotesize
\vspace{0.4em}
\noindent The first three columns report maximum expected regret on a common evaluation grid with spacing $0.125$, and the last column gives the optimal value of the finite linear program with spacing $0.25$.
\end{minipage}
\end{table}

For uncorrelated utility differences with unequal variances and for negatively correlated ones, the finer-grid selector has smaller maximum regret than selecting the largest observation. For the first two covariance matrices, selecting the largest observation does better, although refining the grids reduces the gap.

\paragraph{Regular allocation and class selection.}
The next model isolates weight-estimation error at a unique regular optimum, with four assets, $x_0=1$, and constant absolute risk aversion (CARA) utility $U(x)=(1-\exp(-4(x-1)))/4$. Returns satisfy
\begin{align}
  R_a
  =
  0.006+0.035C+0.075\varepsilon_a,
  \qquad
  a=1,\ldots,4,
  \label{eq:simulation-regular-model}
\end{align}
where $C,\varepsilon_1,\ldots,\varepsilon_4$ are independent Rademacher variables, so equal weighting is the unique interior optimum on the long-only simplex. The statistical model contains all strictly positive probability vectors near the uniform one on these 32 support points, and because the support is finite, $J_P$ and $V_P^{\mathrm{eff}}$ of Theorem~\ref{thm:unique-oracle-constant} in Appendix~\ref{app:influence} can be computed exactly. At each sample size, we compute the empirical expected utility maximizer in 1,500 samples.

\begin{figure}[t]
\centering
\begin{minipage}{0.49\textwidth}
\centering
\includegraphics[width=\textwidth]{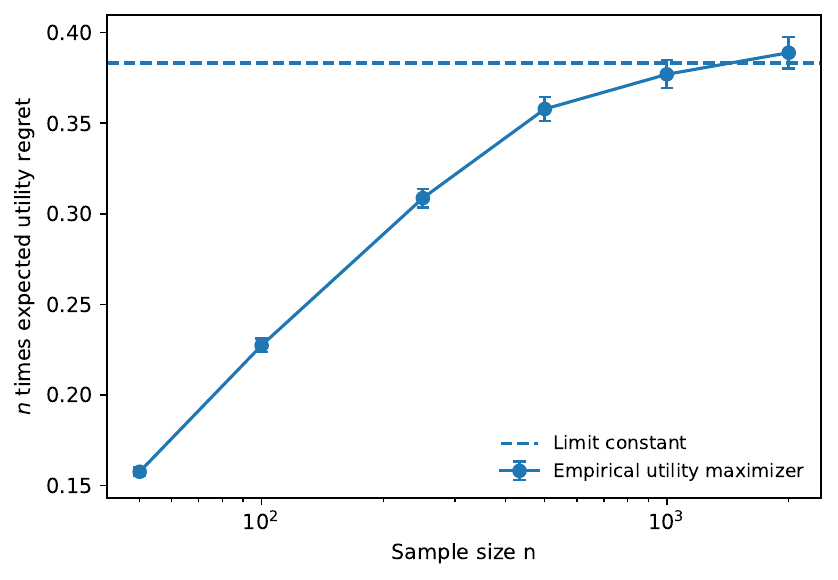}
\end{minipage}\hfill
\begin{minipage}{0.49\textwidth}
\centering
\includegraphics[width=\textwidth]{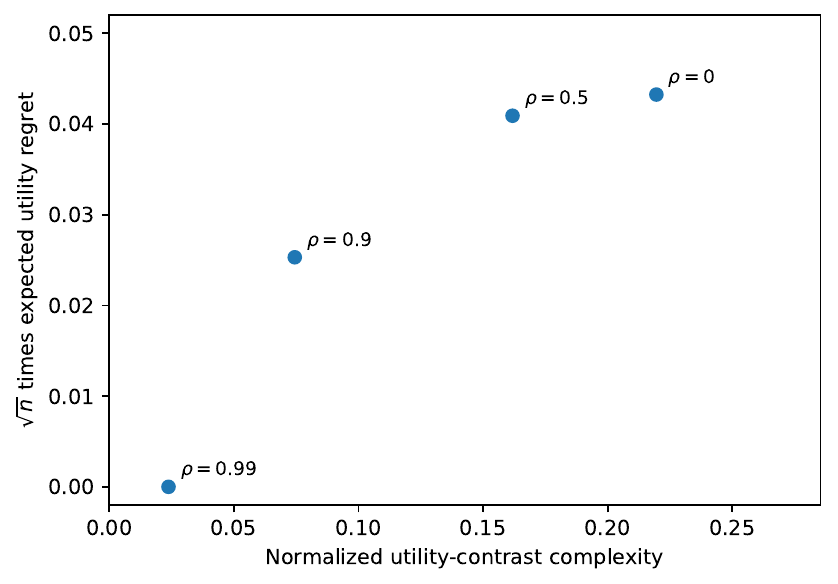}
\end{minipage}
\caption{Regular allocation (left: normalized regret against the exact constant) and correlated class selection (right: normalized selection regret against the complexity of the utility differences at $n=50,100,250,500,1{,}000$). Error bars are Monte Carlo standard errors.}
\label{fig:allocation-selection-experiments}
\end{figure}

The left panel of Figure~\ref{fig:allocation-selection-experiments} compares normalized regret $n\E_P[r_P(\widehat w_n)]$ with the limit constant $\operatorname{tr}(J_PV_P^{\mathrm{eff}})/2=0.3832$. Their ratio rises from $0.411$ at $n=50$ to $0.984$ at $n=1{,}000$ and $1.015$ at $n=2{,}000$ as estimated portfolios on the boundary of the simplex become rarer: at least one estimated weight is zero in $86.7\%$ of samples at $n=50$ and $0.13\%$ at $n=2{,}000$.

To isolate the role of dependence in class comparisons, we fix the number of classes and the normalized utility gap. Each of 20 assets $a=0,\ldots,19$ forms a singleton class, with returns $R_a=\mu_a+0.08(\sqrt\rho C+\sqrt{1-\rho}\varepsilon_a)$, where $C$ and $\varepsilon_a$ are independent Rademacher signs. We set $\mu_a=0$ for $a\ne0$, and under the same CARA utility, $\mu_0$ makes the population expected utility of class zero larger than that of the others by $0.05/\sqrt n$. The rule selects the class with the largest sample expected utility, and we vary $\rho$ over $0$, $0.5$, $0.9$, and $0.99$, with 12,000 replications per design.

At $n=1{,}000$, as $\rho$ rises from $0$ to $0.99$, the utility correlations are $0$, $0.4879$, $0.8870$, and $0.9883$, the normalized complexity of the utility differences declines from $0.2195$ to $0.0238$, and the normalized selection regret falls from $0.0432$ to less than $0.0001$. These declines reflect the role of dependence in the complexity term of Theorem~\ref{thm:selection-allocation-upper}.

\paragraph{Decision boundaries and evaluation priors.}
We next keep the decision rule fixed and vary only how its expected regret is averaged. In a two-class Bernoulli model, the observed utility difference equals $0.08$ with probability $(1+\theta/0.08)/2$ and $-0.08$ otherwise, and the rule selects class one if the estimated difference is strictly positive. With $\bar\theta=0.04$, we average its expected regret under three evaluation priors on $[-\bar\theta,\bar\theta]$ with densities $1/(2\bar\theta)$, $3(1-(\theta/\bar\theta)^2)/(4\bar\theta)$, and $15(1-(\theta/\bar\theta)^2)^2/(16\bar\theta)$, and we also compute its worst-case expected regret.

\begin{figure}[t]
\centering
\begin{minipage}{0.49\textwidth}
\centering
\includegraphics[width=\textwidth]{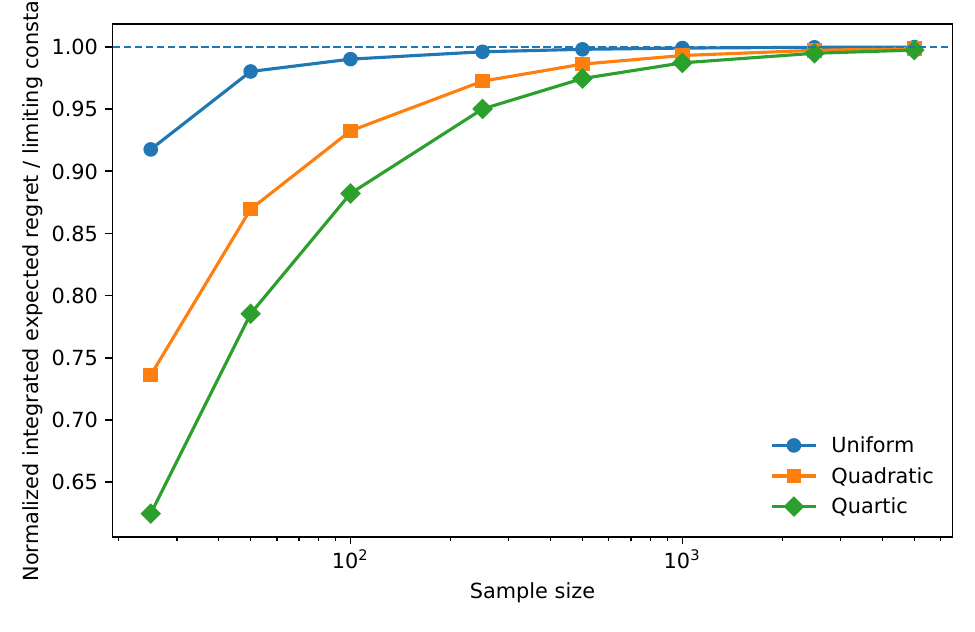}
\end{minipage}\hfill
\begin{minipage}{0.49\textwidth}
\centering
\includegraphics[width=\textwidth]{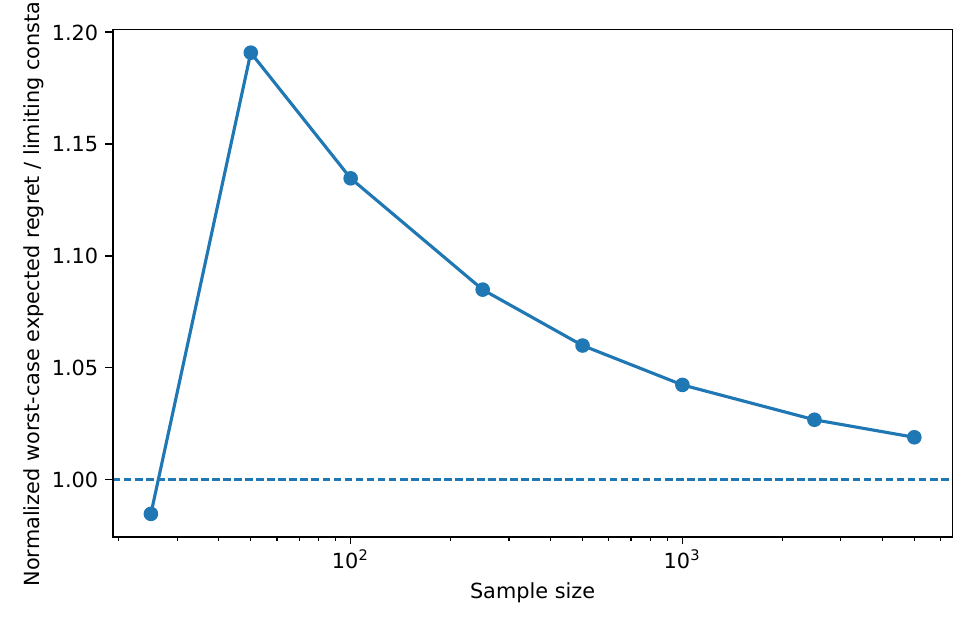}
\end{minipage}
\caption{Regret at a decision boundary, divided by the theoretical constant: integrated expected regret under three evaluation priors (left) and worst-case expected regret (right).}
\label{fig:decision-boundary-experiment}
\end{figure}

The left panel of Figure~\ref{fig:decision-boundary-experiment} shows that the Bayes constant of a single rule varies with the prior density at the decision boundary: at $n=5{,}000$, the normalized integrated expected regrets are $0.03999$, $0.05992$, and $0.07481$ under the uniform, quadratic, and quartic densities, with respective limiting regrets of $0.040$, $0.060$, and $0.075$. In the right panel, $\sqrt n$ times the worst-case expected regret is $0.01385$ at $n=5{,}000$, compared with the Gaussian limit $0.01360$.

\paragraph{Higher moments in class selection.}
We examine whether the first two moments suffice to compare the expected utilities of portfolios. To this end, we use CRRA utility $U(x)=-x^{-4}/4$ with $x_0=1$ and two return distributions on $\{-0.45,\allowbreak-0.10,\allowbreak0.10,\allowbreak0.30\}$ with probability vectors $(461/2100,\allowbreak57/140,\allowbreak1/20,\allowbreak97/300)$ and $(589/2100,\allowbreak13/140,\allowbreak9/20,\allowbreak53/300)$. The two distributions share the mean $-0.0375$, the second moment $0.078125$, and the variance $0.076719$, yet their expected utilities are $-0.791728$ and $-0.893961$. Each criterion, full expected utility or a Taylor criterion of order two, three, or four, selects the class with the larger sample value of the criterion, and we use 12,000 replications at every sample size and both assignments of the distributions to the classes.

\begin{figure}[t]
\centering
\includegraphics[width=0.6\textwidth]{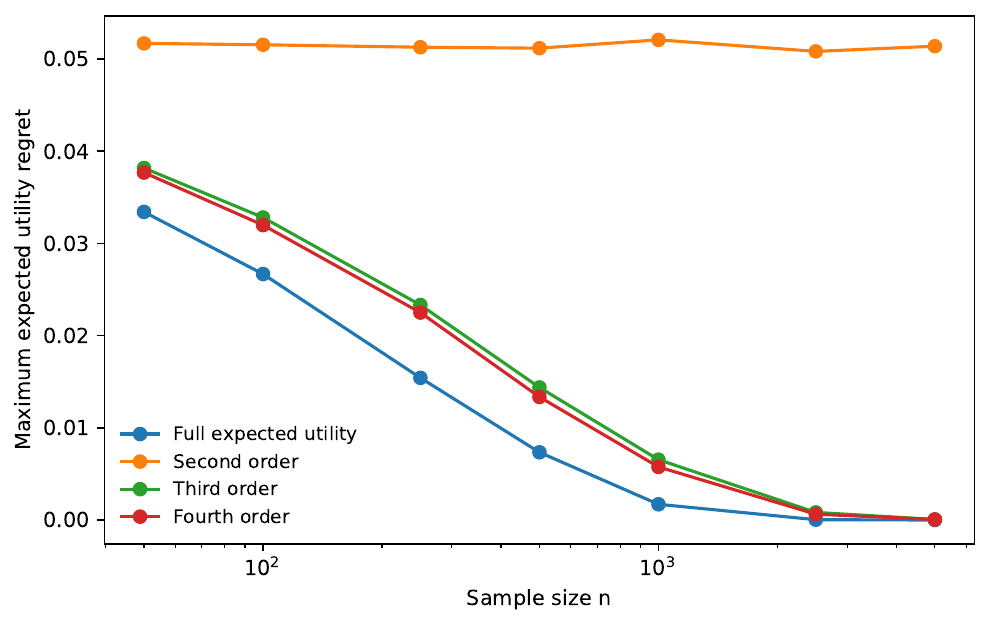}
\caption{Moment separation: worst-case expected utility regret over the two orderings of return distributions with identical means and second moments.}
\label{fig:moment-separation-experiment}
\end{figure}

Because the second-order criterion gives both classes the same population value, its minimum probability of correct selection stays close to one half, and its worst-case expected utility regret is $0.0514$ at $n=5{,}000$, compared with zero for full expected utility and $0.0001$ for the third- and fourth-order criteria. Figure~\ref{fig:moment-separation-experiment} thus illustrates Proposition~\ref{prop:moment-separation}: however accurately the low-order moments are estimated, a criterion that uses only those moments cannot distinguish two distributions whose expected utilities differ only through higher moments.

\paragraph{Third moments in the leading expected regret.}
We use the logarithmic-utility example following Theorem~\ref{thm:third-moment-regret} with $\calW_1=[-0.2,0.2]$ and evaluate the mean $\mu=n^{-1/4}v$ for $v=-1,0,1$. The EUR rule estimates the mean by maximum likelihood over $[-0.4,0.4]$ from all $n$ observations, and the sample mean--variance portfolio uses the same observations.

\begin{table}[t]
\centering
\small
\begin{tabular}{rcccc}
\toprule
$n$ & EUR rule, $v=-1$ & MV, $v=-1$ & EUR rule, $v=1$ & MV, $v=1$ \\
\midrule
4096 & 0.50027 & 0.79945 & 0.50021 & 0.45595 \\
16384 & 0.50006 & 0.70719 & 0.50005 & 0.47903 \\
65536 & 0.50002 & 0.65570 & 0.50001 & 0.49902 \\
262144 & 0.50000 & 0.62444 & 0.50000 & 0.51521 \\
1048576 & 0.50000 & 0.60449 & 0.50000 & 0.52779 \\
\midrule
Limit & 0.50000 & 0.56250 & 0.50000 & 0.56250 \\
\bottomrule
\end{tabular}
\begin{minipage}{0.95\textwidth}
\footnotesize
\vspace{0.4em}
\noindent MV is the sample mean--variance portfolio. Finite-sample entries are exact binomial sums, and the last row gives the limits in Theorem~\ref{thm:third-moment-regret}.
\end{minipage}
\caption{Expected regret multiplied by $n$ in the two-point return model.}
\label{tab:third-moment-regret}
\end{table}

Although both signs of the mean share the limiting excess constant $M_3(v)=1/16$, Table~\ref{tab:third-moment-regret} shows that the sample mean--variance portfolio approaches $9/16$ from above at $v=-1$ and from below at $v=1$. At $v=1$, it therefore has smaller expected regret than the EUR rule up to $n=65{,}536$ and larger expected regret from $n=262{,}144$ on.

\paragraph{Risk parity as the output of the EUR rule.}
To evaluate Proposition~\ref{prop:eurr-risk-parity}, let $K=4$, $x_0=1$, $R=D_sZ$, and $D_s=\operatorname{diag}(0.5,0.65,0.8,0.95)$. The standardized returns have the probability mass function
\begin{align}
 P_\theta(Z=z)
 &=\frac{\exp\left(\theta\sum_{a=1}^4 z_a+
             \tfrac14\sum_{a=1}^4 z_az_{a+1}\right)}{Q(\theta)},
 \qquad z\in\{-1,1\}^4,\quad z_5=z_1,
 \label{eq:rp-numerical-model}
\end{align}
where $Q(\theta)$ is the normalizing constant. The parameter $\theta$ is estimated over $[0.01,0.15]$, and the standardized holdings satisfy $0\leq s_aw_a\leq0.16$. We use $U(x)=\log x$ and $U(x)=(x^{-2}-1)/(-2)$, the latter with relative risk aversion 3, and for either utility the optimum is unique and interior.

This distribution is invariant under cyclic permutations but not under every permutation: at $\theta=0.075$, adjacent and opposite assets have correlations $0.2562$ and $0.1176$. The EUR rule nevertheless assigns $25\%$ of portfolio volatility to each asset, with risky-asset capital shares of approximately $(0.3424,0.2634,0.2140,0.1802)$ for both utilities, while the standardized investment $\alpha_\theta$ is $0.07323$ for log utility and $0.02493$ for relative risk aversion 3.

For each sample size and utility, we use 30,000 repetitions at each of $\theta=0.03,0.075,0.12$ and another 30,000 under a uniform prior on $[0.03,0.12]$, and each repetition estimates $\theta$ by maximum likelihood from all $n$ observations.

\begin{table}[t]
\centering
\small
\begin{tabular}{rcccc}
\toprule
$n$ & Log, $\theta=.075$ & Log, prior & CRRA, $\theta=.075$ & CRRA, prior \\
\midrule
250 & 0.978 (0.008) & 0.911 (0.007) & 0.995 (0.008) & 0.927 (0.007) \\
1000 & 0.992 (0.008) & 0.990 (0.008) & 0.996 (0.008) & 0.992 (0.008) \\
4000 & 0.994 (0.008) & 0.988 (0.008) & 1.013 (0.008) & 1.003 (0.008) \\
16000 & 1.001 (0.008) & 0.987 (0.008) & 1.014 (0.008) & 0.999 (0.008) \\
64000 & 1.009 (0.008) & 1.010 (0.008) & 0.994 (0.008) & 1.004 (0.008) \\
256000 & 0.996 (0.008) & 0.996 (0.008) & 1.009 (0.008) & 0.999 (0.008) \\
\bottomrule
\end{tabular}
\begin{minipage}{0.95\textwidth}
\footnotesize
\vspace{0.4em}
\noindent Parentheses give Monte Carlo standard errors on the same scale. The CRRA columns use coefficient 3.
\end{minipage}
\caption{Expected regret of the risk-parity form of the EUR rule, multiplied by $n$ and divided by the constant in \eqref{eq:rp-general-constant} or its prior integral.}
\label{tab:risk-parity-eurr}
\end{table}

The integrated constants are $0.48019$ for log utility and $0.16756$ for relative risk aversion 3, and Table~\ref{tab:risk-parity-eurr} shows that for both utilities the normalized expected regret approaches one.

\paragraph{Estimating the risk-parity composition.}
In the signed-gamma model of Section~\ref{sec:estimated-volatility}, both the amount invested and the relative scales of the assets are unknown. We use four assets, gamma shape $k_G=3$, CARA coefficient $\gamma=4$, and the feasible box $[0,0.1]^4$, and compute maximum likelihood estimates over $\pi\in[0.53,0.8]$ and $s_a\in[0.7,1.3]$. We evaluate the rule at $\pi=0.58,0.66,0.74$ with scales $(0.85,1,1.1,1.15)$, and under a prior in which $\pi$ is uniform on $[0.58,0.74]$ and the log scales are independent and uniform on $[\log(0.85),\log(1.15)]$.

Each of the twenty settings uses $30{,}000$ repetitions, and in each repetition the estimates from all $n$ observations are used in \eqref{eq:estimated-scale-eurr}. To isolate the loss from estimating the scales, we compare the rule with a comparator that knows the true scales and uses the same estimate of $\pi$, and we compute expected utilities from \eqref{eq:gamma-investment}.

\begin{table}[t]
\centering
\small
\begin{tabular}{rrrrr}
\toprule
$n$ & EUR rule & Known scales & Difference & Contribution RMSE \\
\midrule
250 & 0.10119 (0.00076) & 0.08822 (0.00073) & 0.01297 & 0.01583 \\
1000 & 0.10175 (0.00075) & 0.08912 (0.00072) & 0.01263 & 0.00790 \\
4000 & 0.09936 (0.00073) & 0.08742 (0.00071) & 0.01194 & 0.00393 \\
16000 & 0.10060 (0.00075) & 0.08808 (0.00072) & 0.01252 & 0.00198 \\
64000 & 0.10037 (0.00074) & 0.08796 (0.00071) & 0.01241 & 0.00099 \\
\midrule
Limit & 0.10005 & 0.08781 & 0.01224 & 0 \\
\bottomrule
\end{tabular}
\begin{minipage}{0.95\textwidth}
\footnotesize
\vspace{0.4em}
\noindent Expected-regret columns are multiplied by $n$ and averaged under the joint prior, with Monte Carlo standard errors in parentheses. The last column is the root mean squared deviation from $1/4$ of each asset's true share of portfolio volatility, taken over assets and repetitions.
\end{minipage}
\caption{The EUR rule with estimated relative volatilities.}
\label{tab:estimated-scales}
\end{table}

The integrated leading constant $0.10005$ consists of $0.08781$ from estimating the standardized investment amount and $0.01224$ from estimating the asset scales, and one quarter of the latter reflects the common size of the investment while three quarters reflect its composition. In Table~\ref{tab:estimated-scales}, both the expected regret of the EUR rule and its difference from the known-scale comparator stay close to these limits, and the deviation of the true shares of volatility from $1/4$, which comes from estimating the scales, decreases with $n$.

\paragraph{The cost of departing from symmetry.}
We extend the cyclic sign model by allowing a difference between the first two assets. On $z\in\{-1,1\}^4$, with $z_5=z_1$, let
\begin{align}
 P_{\theta,t}(Z=z)
 &=\frac{\exp\left(\theta\sum_{a=1}^4z_a+t(z_1-z_2)
                   +\tfrac14\sum_{a=1}^4z_az_{a+1}\right)}{Q(\theta,t)}.
 \label{eq:asymmetric-cyclic-model}
\end{align}
Here $Q(\theta,t)$ is the normalizing constant. We use CARA utility with coefficient four, standardized holdings in $[0,0.06]^4$, and estimation ranges $\theta\in[0.05,0.10]$ and $t\in[-0.04,0.04]$. The optimal standardized holdings are then $(\theta\mathbf1+t e)/4$ with $e=(1,-1,0,0)^\top$, but the comparator imposes the equal standardized holdings $\widehat\theta\mathbf1/4$, which are the optimal common amount given the same maximum likelihood estimate of $(\theta,t)$.

We evaluate $\theta=0.075$ and $t=h/\sqrt n$ for $h\in\{0,0.5,1,2\}$ at $n=4{,}096$, $16{,}384$, $65{,}536$, and $262{,}144$, with $30{,}000$ repetitions per setting. Under this model, the constants in \eqref{eq:composition-loss-constants} are $c_\parallel=0.12511$ and $c_\perp=0.12417$, and $L((0,h)^\top)=0.18196h^2$. The limiting expected regrets therefore cross at $|h|\approx0.82607$, so imposing the symmetric composition is better below this value and the EUR rule is better above it.

\begin{table}[t]
\centering
\small
\begin{tabular}{rrrr}
\toprule
$n$ & $h$ & EUR rule & Imposed composition \\
\midrule
4096 & 0 & 0.24760 (0.00142) & 0.12422 (0.00102) \\
4096 & 0.5 & 0.24879 (0.00141) & 0.17266 (0.00103) \\
4096 & 1 & 0.24344 (0.00136) & 0.30786 (0.00103) \\
4096 & 2 & 0.20501 (0.00126) & 0.85311 (0.00102) \\
262144 & 0 & 0.24935 (0.00143) & 0.12504 (0.00102) \\
262144 & 0.5 & 0.25020 (0.00144) & 0.16984 (0.00101) \\
262144 & 1 & 0.25024 (0.00144) & 0.30691 (0.00102) \\
262144 & 2 & 0.24809 (0.00143) & 0.85272 (0.00101) \\
\bottomrule
\end{tabular}
\begin{minipage}{0.95\textwidth}
\footnotesize
\vspace{0.4em}
\noindent Monte Carlo standard errors are in parentheses. The limits from Proposition~\ref{prop:local-composition} are $0.24928$ for the EUR rule and $0.12511+0.18196h^2$ for the imposed composition.
\end{minipage}
\caption{Expected regret multiplied by $n$ near the symmetric model.}
\label{tab:symmetry-departures}
\end{table}

In Table~\ref{tab:symmetry-departures}, the entries at $n=262{,}144$ and $h=0$ are close to their limits, and at $h=2$ the ordering of the two rules reverses.

\paragraph{Money budgets and marginal expected utility.}
In the model \eqref{eq:asymmetric-cyclic-model}, we set $\theta=0.075$ and $t=0$, with known scales $(0.5,0.65,0.8,0.95)$. The unconstrained standardized investment is then $\alpha_\theta=0.01875$, which corresponds to a total risky investment of $\mathcal B_0=0.10952$, and we use 21 equally spaced budgets from $0.6\mathcal B_0$ to $1.1\mathcal B_0$ and the seven shortfalls in Table~\ref{tab:budget-comparison}.

\begin{table}[t]
\centering
\small
\begin{tabular}{rrrr}
\toprule
$\Delta$ & Exact loss & Budget expansion & Marginal-utility approximation \\
\midrule
0.0200000 & 16.352575 & 16.343947 & 16.351785 \\
0.0100000 & 4.086526 & 4.085987 & 4.086477 \\
0.0050000 & 1.021530 & 1.021497 & 1.021527 \\
0.0025000 & 0.255376 & 0.255374 & 0.255376 \\
0.0012500 & 0.063844 & 0.063844 & 0.063844 \\
0.0006250 & 0.015961 & 0.015961 & 0.015961 \\
0.0003125 & 0.003990 & 0.003990 & 0.003990 \\
\bottomrule
\end{tabular}
\begin{minipage}{0.95\textwidth}
\footnotesize
\vspace{0.4em}
\noindent The shortfall is $\Delta=\mathcal B_0-\mathcal B$, and losses are multiplied by $10^6$. The last two columns use \eqref{eq:budget-composition-cost} and \eqref{eq:utility-loss-gradient}.
\end{minipage}
\caption{The cost of imposing equal standardized holdings under a tighter money budget.}
\label{tab:budget-comparison}
\end{table}

The coefficient of $\Delta^2$ in \eqref{eq:budget-composition-cost} is $0.04086$. In Table~\ref{tab:budget-comparison}, the three columns agree more closely as the shortfall decreases.

\subsection{Empirical Studies}
We compare portfolios fitted by full expected utility and by its second-order approximation on monthly value-weighted returns of the ten industry portfolios of the Kenneth R. French Data Library \citep{FrenchDataLibrary} (January 2024 CRSP file) from January 1984 through December 2023, and evaluate the fitted portfolios from January 1999 on.

The primary specification uses CRRA utility with coefficient five and annual rebalancing, and at each decision date a 120-month training block and a 60-month validation block precede a 12-month test block. The classes are all $\binom{10}{3}=120$ three-industry simplices. Within each class, one rule maximizes the sample CRRA utility and the other maximizes the second-order Taylor criterion in \eqref{eq:taylor-criterion}. Each procedure then selects a class by its own objective on the validation block and re-estimates the weights within that class on the combined 180 months before the test block. Classes within $10^{-10}$ of the maximum are treated as tied, and ties are broken in favor of the first class in lexicographic order. Because these procedures maximize sample objectives rather than a likelihood, they are not the EUR rule evaluated in the simulations.

The benchmarks are full expected utility maximization on the ten-asset simplex, the long-only minimum-variance portfolio, and equal weighting. The primary transaction cost is 25 basis points times turnover, where turnover is the $\ell_1$ distance between the weights before rebalancing and the target weights. For each portfolio, we report its certainty equivalent (CE) under the same CRRA utility, paired percentile intervals from 50,000 resamples of the 25 annual test blocks, and Newey--West tests with 12 lags for the monthly utility differences.

\paragraph{Results in the test period.}
We summarize the performance of the fitted rules over the 300 test months.

\begin{table}[t]
\centering
\caption{Out-of-sample performance of the industry portfolios over the test period, net of transaction costs.}
\label{tab:empirical-primary}
\scriptsize
\resizebox{\textwidth}{!}{%
\begin{tabular}{lrrrrrrr}
\toprule
Rule & Mean return & Volatility & Certainty equivalent & Maximum drawdown & Annual turnover & CE difference & $p$-value \\
\midrule
Full expected utility & 8.88 & 16.34 & 1.90 & -50.10 & 0.87 & -2.06 & 0.41 \\
Second order & 8.84 & 16.34 & 1.85 & -49.61 & 0.87 & -2.11 & 0.40 \\
Full simplex & 6.16 & 14.22 & 0.83 & -49.86 & 0.58 & -3.13 & 0.10 \\
Minimum variance & 9.06 & 12.69 & 4.98 & -35.91 & 0.24 & 1.02 & 0.54 \\
Equal weight & 9.85 & 15.16 & 3.96 & -47.61 & 0.10 & 0.00 & -- \\
\bottomrule
\end{tabular}
}
\begin{minipage}{0.96\textwidth}
\footnotesize
\vspace{0.4em}
\noindent All columns except turnover and $p$-value are percentages. CE differences and $p$-values are relative to equal weighting.
\end{minipage}
\end{table}

In Table~\ref{tab:empirical-primary}, the certainty equivalent of the full expected utility rule exceeds that of the rule using the second-order approximation by $0.052$ percentage points, with interval $[-0.320,0.613]$, and that of the full-utility rule on the ten-asset simplex by $1.073$, with interval $[-2.329,4.506]$. Minimum variance gives the highest certainty equivalent, exceeding equal weighting by $1.022$ percentage points with interval $[-2.649,4.708]$. The left panel of Figure~\ref{fig:empirical-results} shows the cumulative net wealth of each rule.

\begin{figure}[t]
\centering
\begin{minipage}{0.49\textwidth}
\centering
\includegraphics[width=\textwidth]{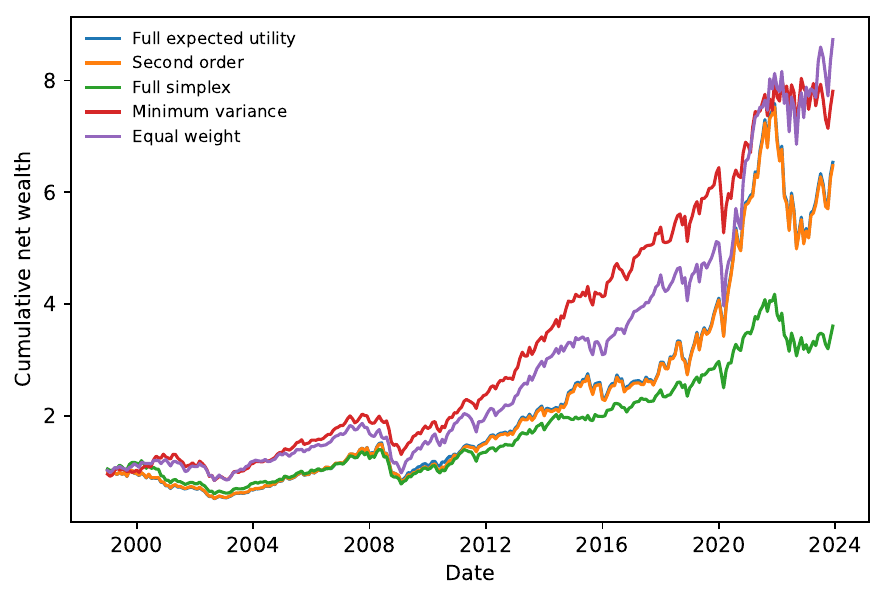}
\end{minipage}\hfill
\begin{minipage}{0.49\textwidth}
\centering
\includegraphics[width=\textwidth]{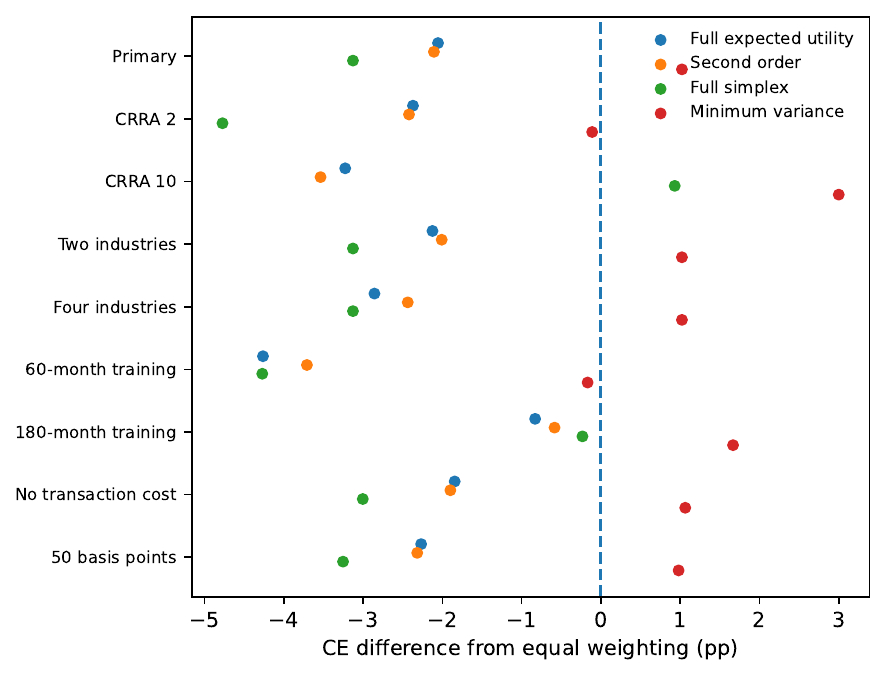}
\end{minipage}
\caption{Industry portfolios: cumulative net wealth in the primary specification (left) and annualized certainty equivalents compared with equal weighting (right).}
\label{fig:empirical-results}
\end{figure}

\paragraph{Comparisons under the same investment constraint.}
Because the ten-industry benchmarks can diversify across more assets than the class-selection procedures, we also compare allocations fitted by full expected utility and by its second-order approximation with minimum-variance and equal-weight allocations within the same 120 three-industry classes. The minimum-variance weights are fitted on the training sample and refitted on all 180 months before the test block.

\begin{table}[t]
\centering
\caption{Portfolio comparisons among classes holding at most three industries.}
\label{tab:empirical-common-constraint}
\small
\begin{tabular}{lrrcr}
\toprule
Procedure & CE (\%) & Difference (pp) & 95\% interval & Turnover \\
\midrule
Full expected utility & 1.903 & 0.000 & -- & 0.874 \\
Second-order objective & 1.851 & -0.052 & [-0.613, 0.320] & 0.868 \\
Minimum variance & 0.475 & -1.427 & [-5.401, 1.954] & 0.920 \\
Equal weights & 1.808 & -0.094 & [-3.175, 3.842] & 0.702 \\
\bottomrule
\end{tabular}
\begin{minipage}{0.96\textwidth}
\footnotesize
\vspace{0.4em}
\noindent The full-utility objective selects the class in every procedure except the second-order one. Differences from full expected utility are in percentage points. Turnover is the average annual sum of absolute weight changes.
\end{minipage}
\end{table}

Under this common restriction, the full-utility procedure has the highest point estimate in Table~\ref{tab:empirical-common-constraint}, although its differences from the added comparators are imprecisely estimated. Moreover, changing the weight estimator within the three-industry classes does not reproduce the performance of the ten-industry minimum-variance portfolio.

\paragraph{The objective used to fit a portfolio.}
Because either objective can be used for class selection and, independently, for the final weight estimation, there are four rules, and comparing them attributes the certainty-equivalent difference between the fitted rules to each step. Let $C_{o_1,o_2}$ denote the annualized certainty equivalent when objective $o_1\in\{F,Q\}$ selects the class and objective $o_2\in\{F,Q\}$ estimates the final weights, where $F$ is full expected utility and $Q$ the second-order criterion. Each contribution averages the paired changes over the two objectives used in the other step:
\begin{align}
  D_{\mathrm{alloc}}
  &={}
  \frac12\left(
    (C_{F,F}-C_{F,Q})
    +(C_{Q,F}-C_{Q,Q})
  \right),
  \label{eq:empirical-allocation-decomposition}
  \\
  D_{\mathrm{select}}
  &={}
  \frac12\left(
    (C_{F,F}-C_{Q,F})
    +(C_{F,Q}-C_{Q,Q})
  \right).
  \label{eq:empirical-selection-decomposition}
\end{align}
By construction, $D_{\mathrm{alloc}}+D_{\mathrm{select}}=C_{F,F}-C_{Q,Q}$.

For each of the four rules, let $U_{j,t}$ be its monthly CRRA utility, $\bar U_j$ its sample mean, and $\lambda_j$ its coefficient in \eqref{eq:empirical-allocation-decomposition} or \eqref{eq:empirical-selection-decomposition}. With the annualized CE transformation $f_\gamma(u)=((1-\gamma)u)^{12/(1-\gamma)}-1$ and $\gamma=5$, we use the centered influence series
\begin{align}
  \xi_t=\sum_{j=1}^4\lambda_j f_\gamma'(\bar U_j)(U_{j,t}-\bar U_j).
  \label{eq:ce-contrast-influence}
\end{align}
We estimate the variance of the CE contribution by dividing the Newey--West long-run variance of $\xi_t$ by the number of test months and use this estimate to compute the statistic in Table~\ref{tab:empirical-decomposition}.

\begin{table}[H]
\centering
\caption{Contributions from class selection and from allocation within a class, in certainty-equivalent terms.}
\label{tab:empirical-decomposition}
\small
\begin{tabular}{lrrrrr}
\toprule
Component & CE contribution & Lower limit & Upper limit & $t$-statistic & $p$-value \\
\midrule
Within-class allocation & -0.122 & -0.283 & 0.038 & -1.302 & 0.193 \\
Class selection & 0.174 & -0.164 & 0.704 & 0.599 & 0.549 \\
Total difference & 0.052 & -0.320 & 0.613 & 0.185 & 0.854 \\
\bottomrule
\end{tabular}
\begin{minipage}{0.94\textwidth}
\footnotesize
\vspace{0.4em}
\noindent Contributions and interval endpoints are percentage points. The $t$-statistics and $p$-values use \eqref{eq:ce-contrast-influence} and Newey--West standard errors.
\end{minipage}
\end{table}

\begin{figure}[t]
\centering
\includegraphics[width=0.62\textwidth]{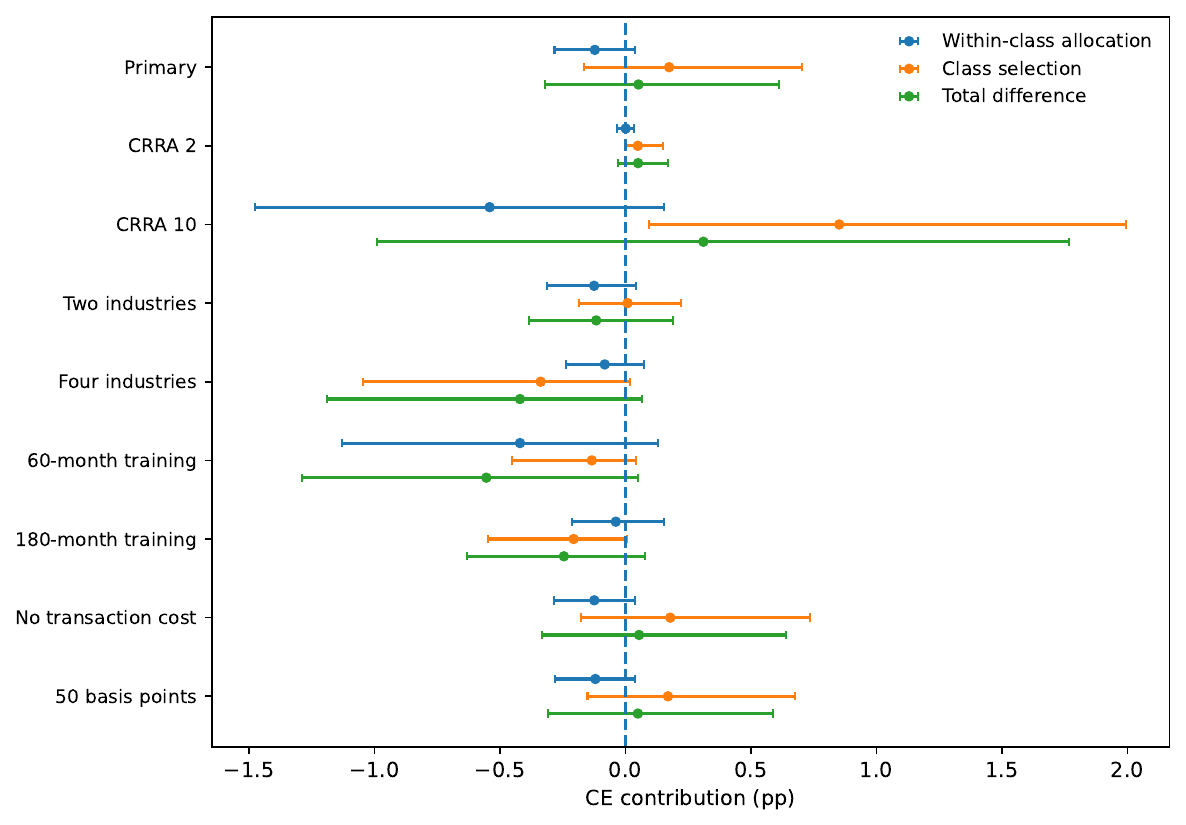}
\caption{Contributions of class selection and of allocation within a class across specifications, in annualized percentage points of certainty equivalent, with paired percentile intervals from resampling annual blocks.}
\label{fig:empirical-decomposition-robustness}
\end{figure}

The four certainty equivalents $C_{F,F}$, $C_{F,Q}$, $C_{Q,F}$, and $C_{Q,Q}$ are $1.903\%$, $2.024\%$, $1.728\%$, and $1.851\%$, and averaging the paired changes gives the contributions in Table~\ref{tab:empirical-decomposition}, whose intervals both include zero.

The two objectives select the same three-industry class in 22 of the 25 annual decisions, and classes are nearly tied in five of the full-utility decisions and eight of the second-order decisions. To assess how reproducible the selection is, we hold the portfolios fitted in each class fixed and resample the five annual blocks of the validation window 20,000 times. The mean probability of selecting the same class again is $0.318$ for full utility and $0.343$ for the second-order criterion. In years when the two objectives disagree, the smaller of these two probabilities is on average $0.108$ larger than in years when they agree, with a permutation $p$-value of $0.496$. The absolute difference between the Taylor remainders on the validation block has correlation $0.905$ with the absolute difference in out-of-sample utility.

\paragraph{Sensitivity to portfolio design.}
We examine the sensitivity of these comparisons by changing the CRRA coefficient from five to two and ten, the number of industries in a class from three to two and four, the training window from 120 months to 60 and 180, and the transaction cost from 25 basis points to zero and 50. The validation window remains 60 months, so the test period begins in January 1994 (360 months) with 60 training months and in January 2004 (240 months) with 180. Across all specifications, certainty-equivalent differences between full utility and the second-order approximation range from $-0.554$ to $0.311$ percentage points, and every bootstrap interval contains zero.

The contribution of weight estimation is negative in eight of the nine specifications and essentially zero when the CRRA coefficient is two, whereas the contribution of class selection is negative with four-industry classes and with both alternative training windows and positive in the other six specifications. Figure~\ref{fig:empirical-decomposition-robustness} shows both contributions, and the right panel of Figure~\ref{fig:empirical-results} shows the certainty-equivalent differences from equal weighting.

The minimum-variance portfolio has a positive but statistically insignificant certainty-equivalent difference from equal weighting in the primary specification, at high risk aversion, with a longer training window, and under either alternative transaction cost.

\section{Conclusion}
\label{sec:conclusion}
We construct a portfolio choice rule that attains the minimax and Bayes lower bounds for expected utility regret, including their leading constants, in a regular parametric return model with a fixed number of assets and finitely many portfolio classes. The EUR rule uses no evaluation prior, and it accounts jointly for choosing a portfolio class and estimating the weights within it. In the worst case, uncertainty about the class can cost more at the leading order than uncertainty about the weights, whereas under the Bayes criterion both costs contribute to the leading constant.

The same rule yields familiar portfolio forms for different economic reasons. When expected excess returns range over neighborhoods of zero that shrink faster than $n^{-1/4}$, estimated means and covariances suffice at the leading order of expected regret, whereas when expected excess returns are of order $n^{-1/4}$, so that the risky investment is larger, third moments change the constant of that order. In the second route, symmetry of the standardized returns and of the feasible standardized holdings under permutations that can move any asset to any other determines a risk-parity composition exactly. When the relative volatilities are unknown, the EUR rule estimates the composition as well, so its constant includes both sources of estimation loss. In both cases, the optimal composition does not require the investor's utility function, which affects only the amount invested, to first order in the expected excess returns in the first route and exactly in the second.

The utility criterion also measures the cost of imposing a simple form when the conditions that support it no longer hold. Near a symmetric model, for example, a prescribed composition gains precision but omits an investment direction, and a binding money budget can change relative holdings even when the standardized return distribution remains symmetric.

\bibliographystyle{tmlr}
\bibliography{arXiv2.bbl}

\clearpage
\appendix

\section{Bounds from Estimation Error}
\label{app:within-class}

The first proof bounds the loss from estimating the weights within a fixed class, and the second adds the error from comparing the fitted portfolios on independent observations.

\begin{proof}[Proof of Proposition~\ref{prop:within-class}]
Fix $P$ and $m$, and choose, by compactness, $w_m^\star\in\calA_m(P)$ satisfying
\begin{align}
  \norm{\widehat w_m-w_m^\star}_m
  =
  d_m\left(\widehat w_m,\calA_m(P)\right).
  \label{eq:closest-oracle}
\end{align}
Since $\widehat w_m$ maximizes $\widehat V_n$ on $\calW_{K,m}$, we have
\begin{align}
  \widehat V_n(\widehat w_m)
  \geq
  \widehat V_n(w_m^\star).
  \label{eq:empirical-optimality}
\end{align}
Hence, the population utility loss satisfies
\begin{align}
  V_{P,m}^\star-V_P(\widehat w_m)
  &\leq
  \left(\widehat V_n-V_P\right)(\widehat w_m)
  -
  \left(\widehat V_n-V_P\right)(w_m^\star)
  \nonumber\\
  &=
  \int_0^1
  \left\langle
    \nabla\widehat V_n(w_m^\star+t(\widehat w_m-w_m^\star))
    -
    \nabla V_P(w_m^\star+t(\widehat w_m-w_m^\star)),
    \widehat w_m-w_m^\star
  \right\rangle
  \rmd t
  \nonumber\\
  &\leq
  Z_{n,m}(P)\norm{\widehat w_m-w_m^\star}_m.
  \label{eq:variational-upper-step}
\end{align}
Assumption~\ref{ass:utility-loss} and \eqref{eq:variational-upper-step} imply
\begin{align}
  c_0\norm{\widehat w_m-w_m^\star}_m^q
  \leq
  Z_{n,m}(P)\norm{\widehat w_m-w_m^\star}_m.
  \label{eq:distance-estimation-comparison}
\end{align}
If the distance is zero, both conclusions are immediate. Otherwise, dividing by it gives \eqref{eq:within-class-distance}, and substitution into \eqref{eq:variational-upper-step} gives \eqref{eq:within-class-bound}.
\end{proof}

\paragraph{Selection among fitted classes.}
\label{app:selection-upper}
We bound the expected regret of the split-sample rule, which estimates the weights within each class on $\mathcal D_1$ and compares the fitted portfolios on an independent sample.

\begin{proof}[Proof of Theorem~\ref{thm:selection-allocation-upper}]
Condition on the training sample $\mathcal D_1$ and write
\begin{align}
  Q_m
  &=
  \E_P\left[\ell_{\widetilde w_m}(R)\mid\mathcal D_1\right],
  \nonumber\\
  \widehat Q_m
  &=
  \frac1{n_2}\sum_{i=1}^{n_2}\ell_{\widetilde w_m}(R_i').
  \label{eq:fitted-model-values}
\end{align}
Choose $m^\star\in\mathfrak M_P^\star$ and $\widetilde m\in\argmax_{m\in\mathfrak M_K}Q_m$. The definition of $\widehat m$ gives $\widehat Q_{\widehat m}\geq\widehat Q_{\widetilde m}$, so we have
\begin{align}
  V_P^\star-Q_{\widehat m}
  &\leq
  V_{P,m^\star}^\star-Q_{m^\star}
  +
  Q_{\widetilde m}-Q_{\widehat m}
  \nonumber\\
  &\leq
  V_{P,m^\star}^\star-Q_{m^\star}
  +
  2\sup_{m\in\mathfrak M_K}
  \left|
    (\widehat Q_m-Q_m)-(\widehat Q_{m_0}-Q_{m_0})
  \right|.
  \label{eq:selection-basic-inequality}
\end{align}
Subtracting the same reference term for every $m$ does not change the maximizing class. Conditional symmetrization then gives
\begin{align}
  \E_P\left[
    \sup_{m\in\mathfrak M_K}
    \left|
      (\widehat Q_m-Q_m)-(\widehat Q_{m_0}-Q_{m_0})
    \right|
    \mathrel{\Big|}\mathcal D_1
  \right]
  \leq
  2\mathfrak C_{\mathrm{sel},n_2}(P\mid\mathcal D_1).
  \label{eq:conditional-symmetrization}
\end{align}
Taking expectations in \eqref{eq:selection-basic-inequality} and minimizing over $m^\star\in\mathfrak M_P^\star$ proves the first inequality in \eqref{eq:selection-allocation-upper}, and Proposition~\ref{prop:within-class} applied to the training sample proves the second.
\end{proof}

\section{Lower Bounds from Finite-Support Return Models}\label{app:minimax-lower}
\label{sec:lower}

This appendix proves lower bounds under the utility-loss condition of Section~\ref{sec:complexity}. Each bound uses bounded return distributions with finite support under which every portfolio choice rule incurs a specified utility loss: the first construction varies the location of an optimal allocation, and the second varies the class that contains it.

\paragraph{Restriction to a submodel.}

A lower bound on a subset of the return model applies to the full model. For every $\calP_{K,0}\subset\calP_K$,
\begin{align}
  \mathcal V_{n,K}^{\mathrm{mm}}
  \geq
  \inf_{\delta_n\in\mathfrak D_{n,K}}
  \sup_{P\in\calP_{K,0}}
  \calR_n(\delta_n,P).
  \label{eq:submodel-lower}
\end{align}
The alternatives below may change with $n$, but each belongs to the fixed return model in \eqref{eq:minimax}, so \eqref{eq:submodel-lower} applies at every $n$.

\paragraph{Allocation lower bound.}

We vary the population-optimal weights over an $r$-dimensional portfolio class, and the Kullback--Leibler divergence $\KL(P,Q)$ from $P$ to $Q$ determines how hard a wrong choice is to avoid.

The following assumption specifies this submodel, with $K$, $r$, the feasible classes, and all constants fixed as $n\to\infty$.

\begin{assumption}[Regular $r$-dimensional allocation submodel]
\label{ass:allocation-lower}
There are a distribution $P_0\in\calP_K$, a class $m_0\in\mathfrak M_K$, a point $w_0\in\operatorname{relint}(\calW_{K,m_0})$, a linear isomorphism $B:\bbR^r\to\operatorname{lin}(\calW_{K,m_0}-w_0)$, and a bounded finite-support family $\{P_t:\norm{t}_2\leq\eta\}\subset\calP_K$ with the following properties.
\begin{enumerate}[label=(\roman*),leftmargin=2.1em]
  \item There is $\eta_x>0$ such that $w_0+Bx\in\calW_{K,m_0}$ whenever $\norm{x}_2\leq\eta_x$.
  \item The following expansion holds uniformly as $(t,x)\to(0,0)$:
  \begin{align}
    V_{P_t}(w_0+Bx)-V_{P_t}(w_0)
    =
    \kappa\ip{t}{x}
    -
    \frac{\lambda}{q}\norm{x}_2^q
    +
    o\left(\norm{t}_2\norm{x}_2+\norm{x}_2^q\right),
    \label{eq:radial-local-objective}
  \end{align}
  where $\kappa>0$, $\lambda>0$, and $q>1$.
  \item For all sufficiently small nonzero $t$, every oracle under $P_t$ lies in the coordinate neighborhood in part (i). There are constants $C_{\mathrm{in}}<\infty$ and $c_{\mathrm{out}}>0$ such that every decision outside $\{w_0+Bx:\norm{x}_2\leq C_{\mathrm{in}}\norm{t}_2^{1/(q-1)}\}$ has regret at least $c_{\mathrm{out}}\norm{t}_2^{q/(q-1)}$.
  \item For all sufficiently small $t$ and $t'$, we have
  \begin{align}
    \KL(P_t,P_{t'})
    \leq
    C\norm{t-t'}_2^2.
    \label{eq:radial-kl}
  \end{align}
\end{enumerate}
\end{assumption}

\begin{theorem}[Allocation lower bound]
\label{thm:allocation-lower}
Under Assumption~\ref{ass:allocation-lower}, there is a constant $c>0$ such that
\begin{align}
  \mathcal V_{n,K}^{\mathrm{mm}}
  \geq
  c\left(\frac{r}{n}\right)^{q/(2(q-1))}
  \label{eq:allocation-lower}
\end{align}
for all sufficiently large $n$.
\end{theorem}

Here $r$ is the affine dimension of $\calW_{K,m_0}$. For $q=2$, Theorem~\ref{thm:allocation-lower} gives the rate $r/n$, and, for fixed $r$, the power of $n$ in \eqref{eq:allocation-lower} matches that of the upper bound in Section~\ref{sec:complexity}.

\paragraph{Lower bound for class selection.}

The next construction changes which class contains the optimal portfolio while keeping the Kullback--Leibler divergences small, so a rule that chooses a wrong class loses a specified amount of expected utility however accurately it estimates the weights.

\begin{assumption}[Separated near-tie packing]
\label{ass:selection-packing}
There are constants $c_1>0$, $c_2>0$, and $\delta_0>0$ with the following property. For every $0<\delta\leq\delta_0$ at which the lower bound is applied, there are an integer $M\geq2$, pairwise disjoint measurable sets $\mathcal U_1,\ldots,\mathcal U_M\subset\calW_K$, and bounded distributions with finite support $P_0,P_1,\ldots,P_M\in\calP_K$ such that
\begin{enumerate}[label=(\roman*),leftmargin=2.1em]
  \item every oracle portfolio under $P_j$ belongs to $\mathcal U_j$;
  \item every $w\notin \mathcal U_j$ satisfies $r_{P_j}(w)\geq c_1\delta$;
  \item the average Kullback--Leibler divergence of one observation satisfies
  \begin{align}
    \frac1M\sum_{j=1}^M\KL(P_j,P_0)
    \leq
    c_2\delta^2.
    \label{eq:packing-kl}
  \end{align}
\end{enumerate}
The integer $M$ may depend on $\delta$ and hence on $n$.
\end{assumption}

\begin{theorem}[Class-selection lower bound]
\label{thm:selection-lower}
Under Assumption~\ref{ass:selection-packing}, there are constants $c>0$ and $c'>0$, depending only on $c_1$ and $c_2$, such that
\begin{align}
  \mathcal V_{n,K}^{\mathrm{mm}}
  \geq
  c\sqrt{\frac{\log M}{n}}
  \label{eq:selection-lower}
\end{align}
provided that $\log M\leq c'n$ and the packing is available at $\delta\asymp\sqrt{\log M/n}$.
\end{theorem}

Here $M$ counts return distributions in the construction, not portfolios or classes, so the result also applies when the decision set is infinite.

\begin{proof}[Proof of Theorem~\ref{thm:allocation-lower}]
When $r$ is sufficiently large, choose a packing $u_1,\ldots,u_M$ of the Euclidean unit sphere in $\bbR^r$ with pairwise distance at least $c_u>0$, $\log M\geq c_Mr$, and $\log M>2\log2$. Bounded values of $r$ are treated at the end. Set
\begin{align}
  t_j
  =
  \delta u_j,
  \qquad
  \delta
  =
  c_\delta\sqrt{\frac rn},
  \label{eq:allocation-packing-points}
\end{align}
where $c_\delta>0$ is small enough that every $t_j$ lies in the submodel of Assumption~\ref{ass:allocation-lower}.

The leading objective in \eqref{eq:radial-local-objective} has the unique maximizer
\begin{align}
  x^\star(t)
  =
  \left(\frac\kappa\lambda\right)^{1/(q-1)}
  \norm{t}_2^{1/(q-1)-1}t.
  \label{eq:local-oracle-coordinate}
\end{align}
Write $x=\delta^{1/(q-1)}y$ and $t=\delta u$. Dividing \eqref{eq:radial-local-objective} by $\delta^{q/(q-1)}$ gives, uniformly for $u$ in the finite packing and $y$ in a fixed bounded set,
\begin{align}
  \delta^{-q/(q-1)}
  \left(
    V_{P_{\delta u}}(w_0+B\delta^{1/(q-1)}y)
    -
    V_{P_{\delta u}}(w_0)
  \right)
  =
  \kappa\ip{u}{y}
  -
  \frac\lambda q\norm{y}_2^q
  +
  o(1).
  \label{eq:scaled-local-objective}
\end{align}
Because the $u_j$ are separated, so are the maximizers of the limiting functions, and we choose disjoint neighborhoods $\mathcal U_j$ of $w_0+Bx^\star(t_j)$ with radii proportional to $\delta^{1/(q-1)}$. Because the packing is finite, the limiting functions are strictly concave, the remainder in \eqref{eq:scaled-local-objective} is uniform, and part (iii) of Assumption~\ref{ass:allocation-lower} holds, it follows that, for all sufficiently large $n$,
\begin{align}
  \inf_{w\in\calW_K\setminus \mathcal U_j}r_{P_{t_j}}(w)
  \geq
  c_r\delta^{q/(q-1)}
  \label{eq:allocation-decision-separation}
\end{align}
for a constant $c_r>0$ independent of $j$ and $n$.

By \eqref{eq:radial-kl}, we have
\begin{align}
  \KL(P_{t_j}^{\otimes n},P_0^{\otimes n})
  \leq
  Cn\delta^2
  \leq
  Cc_\delta^2r.
  \label{eq:allocation-product-kl}
\end{align}
Choose $c_\delta$ so that the last quantity is at most a sufficiently small multiple of $\log M$. Any rule induces an estimator of $j$: it returns the index of the set $\mathcal U_j$ that contains the rule's output, or an arbitrary index if no such set exists. By Fano's inequality and \eqref{eq:allocation-product-kl}, the average probability of an incorrect index is bounded away from zero, and \eqref{eq:allocation-decision-separation} gives
\begin{align}
  \sup_{1\leq j\leq M}
  \E_{P_{t_j}}[r_{P_{t_j}}(\widehat w_n)]
  \geq
  c\delta^{q/(q-1)}
  =
  c\left(\frac rn\right)^{q/(2(q-1))}.
  \label{eq:allocation-fano-conclusion}
\end{align}
For bounded values of $r\geq1$, use $t_+=\delta e_1$ and $t_-=-\delta e_1$ with $\delta=c_\delta/\sqrt n$, which lie in the submodel for all sufficiently large $n$. The same calculation gives disjoint neighborhoods of the two optimal portfolios, and choosing the wrong one incurs regret of order $\delta^{q/(q-1)}$. By \eqref{eq:radial-kl}, $\KL(P_{t_+}^{\otimes n},P_{t_-}^{\otimes n})\leq4C c_\delta^2$, and Pinsker's inequality bounds the total variation distance by $\sqrt{2C}\,c_\delta$, so every binary decision has average error at least $(1-\sqrt{2C}\,c_\delta)/2>0$ for small $c_\delta$. Since $r$ is bounded, the rate $n^{-q/(2(q-1))}$ is the stated rate up to a constant.
\end{proof}

\paragraph{Proof of the selection lower bound.}
We use Fano's inequality when $M\geq3$ and a two-point comparison when $M=2$.

\begin{proof}[Proof of Theorem~\ref{thm:selection-lower}]
Let $J$ be uniform on $\{1,\ldots,M\}$ and let the observations have distribution $P_J^{\otimes n}$. A rule induces an estimator $\widehat J$, equal to $j$ when the selected portfolio lies in $\mathcal U_j$ and arbitrary otherwise. Part (ii) of Assumption~\ref{ass:selection-packing} gives
\begin{align}
  \E_{P_j}\left[r_{P_j}(\widehat w_n)\right]
  \geq
  c_1\delta\Pr_{P_j}(\widehat J\neq j).
  \label{eq:selection-regret-error}
\end{align}
For $M\geq3$, Fano's inequality and \eqref{eq:packing-kl} give
\begin{align}
  \frac1M\sum_{j=1}^M\Pr_{P_j}(\widehat J\neq j)
  \geq
  1-
  \frac{nc_2\delta^2+\log2}{\log M}.
  \label{eq:selection-fano}
\end{align}
With $\delta=c_\delta\sqrt{\log M/n}$ and small $c_\delta$, the right-hand side of \eqref{eq:selection-fano} is uniformly positive because $\log2/\log M\leq\log2/\log3<1$.

For $M=2$, the triangle inequality through $P_0^{\otimes n}$ and Pinsker's inequality bound the total variation distance: $\operatorname{TV}(P_1^{\otimes n},P_2^{\otimes n})\leq\sqrt{n\operatorname{KL}(P_1,P_0)/2}+\sqrt{n\operatorname{KL}(P_2,P_0)/2}\leq\sqrt{2nc_2\delta^2}$. The average binary error is thus at least $(1-\sqrt{2nc_2\delta^2})/2$, which is positive for $\delta=c_\delta\sqrt{\log2/n}$ and small $c_\delta$. Combining either error bound with \eqref{eq:selection-regret-error} proves \eqref{eq:selection-lower}.
\end{proof}

\section{Efficient Estimation of Expected Utility and Portfolio Weights}\label{app:influence}
\label{sec:information}

The constants in Sections~\ref{sec:minimax} and~\ref{sec:bayes} are expressed through efficient estimation of the optimized expected utilities and weights, so we derive the corresponding influence functions here, together with the constant used in the single-class simulation.

\paragraph{Expected utility at specified portfolios.}

Restrictions on the return distribution can reduce the variance of an efficient estimate of expected utility and so sharpen portfolio comparisons. For a fixed portfolio $w$, let
\begin{align}
  Y_w(R)
  =
  U\left(x_0+w^\top R\right),
  \qquad
  V_P(w)=\E_P[Y_w(R)].
  \label{eq:portfolio-utility}
\end{align}
In an unrestricted dominated nonparametric model, the efficient influence function of $V_P(w)$ is
\begin{align}
  \phi_w^{\mathrm{np}}(R)
  =
  Y_w(R)-V_P(w).
  \label{eq:np-value-if}
\end{align}
For a known regular parametric model with score vector $S_\theta$ and Fisher information $I_\theta$, the efficient influence function is
\begin{align}
  \phi_w^{\mathrm{par}}(R)
  =
  \E_\theta\left[
    \phi_w^{\mathrm{np}}(R)S_\theta(R)^\top
  \right]
  I_\theta^{-1}S_\theta(R).
  \label{eq:par-value-if}
\end{align}
The function in \eqref{eq:par-value-if} is the orthogonal projection of $\phi_w^{\mathrm{np}}$ onto the parametric score space, which gives the following covariance comparison.

\begin{proposition}[Value of parametric information]
\label{prop:parametric-information}
For every finite collection of differences between portfolio expected utilities, the parametric efficient covariance is weakly smaller than the nonparametric efficient covariance in positive-semidefinite order. The covariance difference equals the covariance of the projection residual. Equality for a scalar utility difference holds if and only if that residual has zero variance.
\end{proposition}

With $K$ and the parameter dimension fixed, correctly specified parametric information therefore changes the constants of the leading regret terms but not the powers of $n$.

\paragraph{Efficient estimation at a unique oracle.}

When expected utility has a unique regular maximizer, the estimation target is the portfolio itself, whose influence function follows from the first-order condition. Suppose that $\calW_K$ is locally an affine set at a unique oracle $w_P^\star$, and let the columns of $N\in\bbR^{K\times r}$ form an orthonormal basis for the feasible tangent space. The first-order condition is
\begin{align}
  N^\top
  \E_P\left[
    U'\left(x_0+w_P^{\star\top}R\right)R
  \right]
  =0.
  \label{eq:foc}
\end{align}
Define
\begin{align}
  J_P
  =
  -N^\top
  \E_P\left[
    U''\left(x_0+w_P^{\star\top}R\right)RR^\top
  \right]
  N.
  \label{eq:J}
\end{align}
Assume that $J_P$ is positive definite. In the unrestricted nonparametric model, the efficient influence function of the tangent coordinate $\xi_P$ in $w=w_P^\star+N\xi$ is
\begin{align}
  \psi_P^{\mathrm{np}}(R)
  =
  J_P^{-1}N^\top
  U'\left(x_0+w_P^{\star\top}R\right)R.
  \label{eq:weight-if}
\end{align}
The efficient influence function of the full weight vector is $N\psi_P^{\mathrm{np}}$. In a parametric model, the efficient influence function $\psi_P^{\mathrm{eff}}$ of the tangent coordinate, with covariance matrix $V_P^{\mathrm{eff}}$, is the projection of \eqref{eq:weight-if} onto the score space.

The following theorem translates the covariance of an efficient estimate into expected regret, with $P$, $K$, the tangent dimension $r$, and $U$ fixed as $n\to\infty$.

\begin{theorem}[Exact constant at a regular unique oracle]
\label{thm:unique-oracle-constant}
Suppose that Assumption~\ref{ass:utility} holds through the second derivative, $J_P$ is positive definite, and the oracle map is pathwise differentiable. Let $\widehat w_n=w_P^\star+N\widehat\xi_n$ with probability tending to one, and suppose that
\begin{align}
  \sqrt n\,\widehat\xi_n
  =
  \frac1{\sqrt n}\sum_{i=1}^n\psi_P^{\mathrm{eff}}(R_i)
  +o_P(1),
  \qquad
  \Var_P\left(\psi_P^{\mathrm{eff}}(R)\right)
  =
  V_P^{\mathrm{eff}}.
  \label{eq:unique-oracle-linear-representation}
\end{align}
Assume that the sequence $n\norm{\widehat\xi_n}_2^2$ is uniformly integrable and that the quadratic expansion is integrable in the sense that
\begin{align}
  \E_P\left[
    n\left|
      r_P(\widehat w_n)
      -
      \frac12\widehat\xi_n^\top J_P\widehat\xi_n
    \right|
  \right]
  \to0.
  \label{eq:integrable-quadratic-expansion}
\end{align}
Then, we have
\begin{align}
  \lim_{n\to\infty}
  n\E_P\left[r_P(\widehat w_n)\right]
  =
  \frac12\operatorname{tr}\left(J_PV_P^{\mathrm{eff}}\right).
  \label{eq:quadratic-constant}
\end{align}
\end{theorem}

In \eqref{eq:quadratic-constant}, $J_P$ translates a weight error into a loss of expected utility, and $V_P^{\mathrm{eff}}$ is the covariance of an efficient estimate of those weights.

The constants that govern the comparison of classes also require the influence function of the optimized value within a class.

\begin{lemma}[Influence function for an optimized value within a portfolio class]
\label{lem:optimized-value-if}
Suppose that $w_m(P)$ is the unique regular optimizer in $\calW_{K,m}$. In the unrestricted dominated nonparametric model, the efficient influence function of $V_{P,m}^\star$ is
\begin{align}
  \phi_m(R)
  =
  U\left(x_0+w_m(P)^\top R\right)-V_{P,m}^\star.
  \label{eq:optimized-value-if}
\end{align}
\end{lemma}

\begin{proof}
Let $P_t$ be a differentiable path through $P$ with score $s$, let
\begin{align}
  f(t,w)
  =
  \E_{P_t}\left[U\left(x_0+w^\top R\right)\right],
  \label{eq:path-objective}
\end{align}
and write $w_t=w_m(P_t)$. Optimality gives
\begin{align}
  f(t,w_m(P))-f(0,w_m(P))
  \leq
  V_{P_t,m}^\star-V_{P,m}^\star
  \leq
  f(t,w_t)-f(0,w_t).
  \label{eq:optimized-value-squeeze}
\end{align}
Regularity implies $w_t\to w_m(P)$. Moreover, the two sides of \eqref{eq:optimized-value-squeeze} have the same derivative at $t=0$, because at $w_m(P)$ the derivative of $w\mapsto V_P(w)$ is zero in every direction of the feasible tangent space. It follows that
\begin{align}
  \frac{\rmd}{\rmd t}V_{P_t,m}^\star\bigg|_{t=0}
  =
  \E_P\left[
    \left(
      U\left(x_0+w_m(P)^\top R\right)-V_{P,m}^\star
    \right)s(R)
  \right],
  \label{eq:optimized-value-path-derivative}
\end{align}
which proves \eqref{eq:optimized-value-if}.
\end{proof}

\begin{proof}[Proof of Proposition~\ref{prop:parametric-information}]
Let $\phi^{\mathrm{np}}$ be the influence function of a linear combination of the specified expected-utility differences, and let $\calS_\theta$ be the closed parametric score space. The parametric efficient influence function is the orthogonal projection
\begin{align}
  \phi^{\mathrm{par}}
  =
  \operatorname{proj}_{\calS_\theta}\phi^{\mathrm{np}}.
  \label{eq:parametric-projection}
\end{align}
Orthogonality of the projection residual to $\phi^{\mathrm{par}}$ yields the variance decomposition
\begin{align}
  \Var(\phi^{\mathrm{np}})
  =
  \Var(\phi^{\mathrm{par}})
  +
  \Var(\phi^{\mathrm{np}}-\phi^{\mathrm{par}}).
  \label{eq:projection-variance-decomposition}
\end{align}
Applying this identity to every linear combination gives the positive-semidefinite inequality, and for a scalar difference it gives the equality condition.
\end{proof}

\begin{proof}[Derivation of \eqref{eq:weight-if}]
Write $w^\star=w_P^\star$ and, for nearby $Q$, $w_Q^\star=w^\star+N\xi(Q)$. Define
\begin{align}
  F(Q,\xi)
  =
  N^\top\E_Q\left[
    U'\left(x_0+(w^\star+N\xi)^\top R\right)R
  \right].
  \label{eq:tangent-first-order-map}
\end{align}
Then $F(P,0)=0$, and the derivative of $F$ in $\xi$ at $(P,0)$ is $-J_P$. Along a path through $P$ with score $s$, the path derivative with $\xi=0$ held fixed is
\begin{align}
  N^\top\E_P\left[
    U'\left(x_0+w_P^{\star\top}R\right)R s(R)
  \right].
  \label{eq:tangent-path-derivative}
\end{align}
Differentiating $F(P_t,\xi(P_t))=0$ at $t=0$ with the implicit function theorem shows that the path derivative of the oracle coordinate is the expectation of \eqref{eq:weight-if} times $s(R)$.
\end{proof}

\begin{proof}[Proof of Theorem~\ref{thm:unique-oracle-constant}]
The asymptotic linear representation in \eqref{eq:unique-oracle-linear-representation} and the multivariate central limit theorem give
\begin{align}
  \sqrt n\,\widehat\xi_n
  \xrightarrow{d}
  Z_P,
  \qquad
  Z_P\sim N\left(0,V_P^{\mathrm{eff}}\right).
  \label{eq:weight-asymptotic-normality}
\end{align}
Uniform integrability of $n\norm{\widehat\xi_n}_2^2$ carries over to $n\widehat\xi_n^\top J_P\widehat\xi_n$, so \eqref{eq:weight-asymptotic-normality} implies
\begin{align}
  \lim_{n\to\infty}
  \E_P\left[
    n\widehat\xi_n^\top J_P\widehat\xi_n
  \right]
  =
  \E\left[Z_P^\top J_PZ_P\right]
  =
  \operatorname{tr}\left(J_PV_P^{\mathrm{eff}}\right).
  \label{eq:quadratic-form-limit}
\end{align}
The integrable quadratic expansion in \eqref{eq:integrable-quadratic-expansion} then gives
\begin{align}
  n\E_P\left[r_P(\widehat w_n)\right]
  =
  \frac12
  \E_P\left[
    n\widehat\xi_n^\top J_P\widehat\xi_n
  \right]
  +o(1),
  \label{eq:expected-regret-quadratic}
\end{align}
which proves \eqref{eq:quadratic-constant}.
\end{proof}

\section{Proof of the Mean--Variance Result}
\label{app:local-mv}
We prove Proposition~\ref{prop:mv-localization} and Theorem~\ref{thm:local-mean-variance} under Assumption~\ref{ass:mv-local-utility} with the limits \eqref{eq:mv-joint-limit}. Norms are Euclidean or induced matrix norms, and constants are independent of $\mu\in\Theta_n$ for sufficiently large $n$.

\subsection{Localization of the Population Optimum}
The bounded support in \eqref{eq:mv-mean-family} makes the mean map and covariance smooth near zero. The score with respect to the mean parameter is $S_\mu(R)=\Sigma_\mu^{-1}(R-\mu)$, and the Fisher information is $I_\mu=\Sigma_\mu^{-1}$.

At $\mu=0$, concavity gives $\E_0[U(x_0+w^\top R)]\leq U(x_0)$. The inequality is strict for $w\ne0$ because $U''(x_0)<0$ puts the supporting tangent at $x_0$ strictly above $U(x_0+x)$ for nonzero attainable $x$ and $\Sigma_0\succ0$ makes $w^\top R$ nonzero with positive probability. Hence, $w=0$ is the unique maximizer at $\mu=0$.

Because expected utility is jointly continuous in $(\mu,w)$ on the product of a fixed compact neighborhood of zero and $\calW_K$, uniqueness and compactness imply that all maximizers converge to zero as $\mu\to0$. The negative Hessian converges to $J_0=b\Sigma_0\succ0$ as $(\mu,w)\to(0,0)$, so it is positive definite on a fixed small ball around zero, and every maximizer is interior for sufficiently small $\mu$. The implicit function theorem applied to $\nabla_w V_{P_\mu}(w)=0$ gives a unique $C^2$ map $w^\star(\mu)$ with $w^\star(0)=0$ and derivative $\dot w^\star(0)=A^{-1}\Sigma_0^{-1}$. In particular, $\norm{w^\star(\mu)}_2\leq C\norm{\mu}_2$.

Taylor expansion of $U'$ at $x_0$ gives
\begin{align}
  \nabla_wV_{P_\mu}(w)
  =u_1\mu-b(\Sigma_\mu+\mu\mu^\top)w+O(\norm{w}_2^2).
  \label{eq:mv-gradient-expansion}
\end{align}
The constant in the remainder is uniform because the returns and $U'''$ are bounded on this neighborhood. Substituting $w=w^\star(\mu)=O(\norm{\mu}_2)$ into \eqref{eq:mv-gradient-expansion} and solving for $w$ proves \eqref{eq:mv-oracle-localization}. Expanding $U$ itself through its second derivative proves \eqref{eq:mv-local-objective}.

\subsection{Expected Regret of the Two Estimated Portfolios}
We first bound the error of the sample mean--variance portfolio. Write $\Delta_n=\bar R_n-\mu$ and $E_n=\widehat\Sigma_n-\Sigma_\mu$. Bounded returns give, for each fixed integer $r\geq2$,
\begin{align}
  \sup_{\mu\in\Theta_n}\E_\mu[\norm{\Delta_n}_2^r+\norm{E_n}^r]
  \leq C_r n^{-r/2}.
  \label{eq:mv-sample-moments}
\end{align}
Together with exponential bounds on deviations of a fixed size and the uniform positive definiteness of $\Sigma_\mu$ near zero, these bounds give an event $\mathcal F_n$ with $\sup_{\mu\in\Theta_n}\Pr_\mu(\mathcal F_n^c)\leq Ce^{-cn}$ on which $\widehat\Sigma_n\succ0$ and $A^{-1}\widehat\Sigma_n^{-1}\bar R_n$ lies in the interior of $\calW_K$ and hence equals $\widehat w_n^{\mathrm{MV}}$.

The inverse-matrix identity and \eqref{eq:mv-oracle-localization} yield
\begin{align}
  \widehat w_n^{\mathrm{MV}}-w^\star(\mu)
  &=A^{-1}\Sigma_\mu^{-1}\Delta_n+e_{n,\mu},
  \nonumber\\
  \sup_{\mu\in\Theta_n}\E_\mu[\norm{e_{n,\mu}}_2^2]
  &\leq C\left(r_{\mu,n}^4+\frac{r_{\mu,n}^2}{n}+\frac1{n^2}+e^{-cn}\right).
  \label{eq:mv-estimator-expansion}
\end{align}
On $\mathcal F_n$, the remaining terms are bounded by $C(\norm{\mu}_2^2+\norm{E_n}\norm{\mu}_2+\norm{E_n}\norm{\Delta_n}_2)$, whose second moments are bounded by \eqref{eq:mv-sample-moments}. On $\mathcal F_n^c$, compactness bounds the portfolio error and bounded returns bound the linear term, so \eqref{eq:mv-estimator-expansion} holds without changing the rule on that event.

The assumptions on $r_{\mu,n}$ imply $n\sup_{\mu\in\Theta_n}\E_\mu[\norm{e_{n,\mu}}_2^2]\to0$. A third-moment version of the same calculation gives
\begin{align}
  \sup_{\mu\in\Theta_n}
  n\E_\mu[\norm{\widehat w_n^{\mathrm{MV}}-w^\star(\mu)}_2^3]
  \longrightarrow0.
  \label{eq:mv-third-moment}
\end{align}
Since the Hessian at the oracle converges to $-J_0=-b\Sigma_0$ as $\mu\to0$, a Taylor expansion around the interior optimum, together with \eqref{eq:mv-third-moment} and the exponentially small probability of the exceptional event, gives
\begin{align}
  \sup_{\mu\in\Theta_n}\left|
    n\calR_n(\delta_n^{\mathrm{MV}},P_\mu)
    -\frac12\operatorname{tr}\left(J_0A^{-2}\Sigma_\mu^{-1}\right)
  \right|\longrightarrow0.
  \label{eq:mv-uniform-expected-regret}
\end{align}
Because $\Sigma_\mu\to\Sigma_0$ uniformly on $\Theta_n$, the trace converges to $Ku_1^2/(2b)=C_0$. This proves the upper bounds in \eqref{eq:local-mv-minimax} and \eqref{eq:local-mv-bayes}.

For the one-class EUR rule, the maximum likelihood estimate is computed from all $n$ observations, and the resulting estimate of the mean equals $\bar R_n$ whenever $\bar R_n$ lies in the interior of the fixed estimation set, because the likelihood score vanishes exactly when $\nabla\psi(\eta)=\bar R_n$ and strict convexity of $\psi$ makes the maximum unique. The complementary event has probability at most $Ce^{-cn}$. Expanding $w^\star(\widehat\mu(D))$ gives
\begin{align}
  \widehat w_n^\star-w^\star(\mu)
  &=A^{-1}\Sigma_\mu^{-1}(\bar R_n-\mu)+\widetilde e_{n,\mu},
  \nonumber\\
  \sup_{\mu\in\Theta_n}\E_\mu[\norm{\widetilde e_{n,\mu}}_2^2]
  &\leq C\left(\frac{r_{\mu,n}^2}n+\frac1{n^2}+e^{-cn}\right).
  \label{eq:mv-eurr-expansion}
\end{align}
Here $\dot w^\star(\mu)=A^{-1}\Sigma_\mu^{-1}+O(\norm{\mu}_2)$ by $C^2$ regularity, so both rules have the same linear term in the same sample mean. An estimate computed from the second sample alone would use the mean $\bar R_{N_n}^{(1)}=N_n^{-1}\sum_{R_i\in D_1}R_i$, whose difference from $\bar R_n$ has the second-moment matrix
\begin{align}
  \E_\mu\left[
    (\bar R_{N_n}^{(1)}-\bar R_n)(\bar R_{N_n}^{(1)}-\bar R_n)^\top
  \right]
  =\left(\frac1{N_n}-\frac1n\right)\Sigma_\mu,
  \label{eq:mv-sample-split-difference}
\end{align}
which vanishes at rate $k_n/n$ relative to $\Sigma_\mu/n$. Equations \eqref{eq:mv-estimator-expansion}--\eqref{eq:mv-eurr-expansion} prove \eqref{eq:eurr-mv-equivalence}. The same quadratic-regret calculation gives the constant $C_0$ for the EUR rule, both uniformly over $\Theta_n$ and after integration under $H_n$.

\subsection{Lower Bound for Arbitrary Rules}
We prove that no rule improves the constant. After reducing arbitrary feasible decisions to a neighborhood of the cash portfolio $w=0$, we apply an integration-by-parts inequality of the type developed by \citet{Gill1995vantrees}.

Set $r_n=\sqrt{r_{\mu,n}}$. For sufficiently large $n$, the closed ball of radius $r_n$ lies in $\calW_K$. For a unit vector $v$ and $\mu\in\Theta_n$, the directional derivative at $r_nv$ satisfies
\begin{align}
  \frac{\rmd}{\rmd t}V_{P_\mu}(tv)\bigg|_{t=r_n}
  \leq u_1r_{\mu,n}-c r_n<0.
  \label{eq:mv-radial-derivative}
\end{align}
The bound holds because the negative Hessian is positive definite on a fixed small ball, and concavity along each ray then shows that replacing $w$ outside this ball by $r_nw/\norm{w}_2$ weakly increases expected utility for every $\mu\in\Theta_n$. Define $T_n(0)=0$ and $T_n(w)=w\min\{1,r_n/\norm{w}_2\}$ for $w\ne0$. Because the projection $T_n$ depends on neither $\mu$ nor $H_n$, it can be applied to any randomized rule, so we may restrict all output portfolios to this ball.

Since $w^\star(\mu)=O(r_{\mu,n})$, continuity of the Hessian gives a deterministic $\epsilon_n\to0$ such that
\begin{align}
  r_{P_\mu}(w)
  \geq \frac{1-\epsilon_n}{2}
     (w-w^\star(\mu))^\top J_0(w-w^\star(\mu))
  \label{eq:mv-local-quadratic-lower}
\end{align}
for every $\mu\in\Theta_n$ and $\norm{w}_2\leq r_n$.

Let $\E$ now denote expectation under the joint distribution that first draws $\mu\sim H_n$ and then the sample from $P_\mu^{\otimes n}$, including any randomization of the rule. Its score with respect to $\mu$ is
\begin{align}
  S_{n,H_n}=\sum_{i=1}^nS_\mu(R_i)+\nabla\log h_n(\mu).
  \label{eq:mv-joint-score}
\end{align}
For any bounded measurable estimator $\widehat w$, integration by parts gives
\begin{align}
  \E[(\widehat w-w^\star(\mu))S_{n,H_n}^\top]
  &=\overline D_n, &\overline D_n&=\int_{\Theta_n}\dot w^\star(\mu)h_n(\mu)\,\rmd\mu.
  \label{eq:mv-score-identity}
\end{align}
The boundary term is zero because $h_n$ vanishes on the boundary of $\Theta_n$. Since the identity differentiates only the likelihood and $w^\star$, and not the estimator, it applies to every measurable rule.

Since likelihood scores have mean zero conditional on $\mu$ and the observations are independent, we have
\begin{align}
  \mathcal I_n:=\E[S_{n,H_n}S_{n,H_n}^\top]
  =n\int_{\Theta_n}\Sigma_\mu^{-1}h_n(\mu)\,\rmd\mu+r_{\mu,n}^{-2}\mathcal I(f).
  \label{eq:mv-joint-information}
\end{align}
Subtracting $\overline D_n\mathcal I_n^{-1}S_{n,H_n}$ from the estimation error and taking its second-moment matrix proves
\begin{align}
  \E[(\widehat w-w^\star(\mu))(\widehat w-w^\star(\mu))^\top]
  \succeq \overline D_n\mathcal I_n^{-1}\overline D_n^\top.
  \label{eq:mv-matrix-lower}
\end{align}
The condition $nr_{\mu,n}^2\to\infty$ makes the prior contribution in $\mathcal I_n/n$ vanish, so $\mathcal I_n/n\to\Sigma_0^{-1}$ and $\overline D_n\to A^{-1}\Sigma_0^{-1}$. Combining \eqref{eq:mv-local-quadratic-lower} and \eqref{eq:mv-matrix-lower} gives, for every rule,
\begin{align}
  \liminf_{n\to\infty}n\int_{\Theta_n}
       \calR_n(\delta_n,P_\mu)\,\rmd H_n(\mu)
  \geq\frac12\operatorname{tr}\left(J_0 A^{-2}\Sigma_0^{-1}\right)
  =C_0.
  \label{eq:mv-bayes-lower-proof}
\end{align}
The finite-$n$ bound before the limit, $(1-\epsilon_n)\operatorname{tr}(J_0\overline D_n\mathcal I_n^{-1}\overline D_n^\top)/2$, does not depend on the rule, so it bounds the infimum at each $n$. Moreover, the worst-case expected regret is at least the integrated expected regret. With \eqref{eq:mv-uniform-expected-regret}, this proves both equalities in Theorem~\ref{thm:local-mean-variance}.

\section{Proof of the Third-Moment Result}
\label{app:critical-moments}
We first establish \eqref{eq:third-moment-oracle}. Put $u_3=U'''(x_0)$ and $w_0(\mu)=A^{-1}\Sigma_\mu^{-1}\mu$. A Taylor expansion of $U'$ at $x_0$ gives
\begin{align}
  \nabla_wV_{P_\mu}(w)
  =u_1\mu-b(\Sigma_\mu+\mu\mu^\top)w
    +\frac{u_3}{2}\E_\mu[R(w^\top R)^2]+O(\norm{w}_2^3).
  \label{eq:critical-score-expansion}
\end{align}
The remainder is bounded on a neighborhood of zero because the returns are bounded and $U''''$ is continuous. Proposition~\ref{prop:mv-localization} gives $w^\star(\mu)=w_0(\mu)+O(\norm{\mu}_2^2)$, and after substitution into \eqref{eq:critical-score-expansion}, the term $\mu\mu^\top w^\star(\mu)$ is $O(\norm{\mu}_2^3)$. Moreover, replacing $P_\mu$ and $\Sigma_\mu$ by $P_0$ and $\Sigma_0$ inside the quadratic term changes that term by $O(\norm{\mu}_2^3)$, because the exponential family is smooth and the support is bounded. Solving the first-order condition for $w^\star(\mu)-w_0(\mu)$ proves \eqref{eq:third-moment-oracle} with the vector $\beta$ in \eqref{eq:third-moment-displacement}.

For $\mu=n^{-1/4}v$ and $\norm{v}_2\leq1$, the sample-moment bounds in \eqref{eq:mv-sample-moments} and the inverse-matrix expansion give
\begin{align}
  \sqrt n\left(\widehat w_n^{\mathrm{MV}}-w^\star(\mu)\right)
  &=A^{-1}\Sigma_0^{-1}\sqrt n(\bar R_n-\mu)-\beta(v)+e_{n,v},
  &\sup_{\norm{v}_2\leq1}\E_\mu[\norm{e_{n,v}}_2^2]&\longrightarrow0.
  \label{eq:critical-mv-estimation}
\end{align}
After multiplication by $\sqrt n$, the deterministic remainder is $O(\sqrt n\norm{\mu}_2^3)=O(n^{-1/4})$, and the random contributions from estimating the covariance have second moment at most $C(\norm{\mu}_2^2+n^{-1})$ after the same normalization. The events on which the unconstrained optimizer is infeasible have exponentially small probability, and they contribute $o(1)$ because the chosen portfolios lie in the compact feasible set. The corresponding fourth-moment bounds imply $n\E_\mu[\norm{\widehat w_n^{\mathrm{MV}}-w^\star(\mu)}_2^3]\to0$, uniformly over $\norm{v}_2\leq1$.

The linear random term in \eqref{eq:critical-mv-estimation} has expectation zero and covariance tending to $A^{-2}\Sigma_0^{-1}$, and the negative Hessian at the oracle tends to $J_0=b\Sigma_0$. Expanding regret at the interior oracle and using the preceding moment bounds therefore gives
\begin{align}
  n\calR_n(\delta_n^{\mathrm{MV}},P_\mu)
  &=\frac12\operatorname{tr}(J_0A^{-2}\Sigma_0^{-1})
    +\frac12 \beta(v)^\top J_0\beta(v)+o(1)
   =C_0+M_3(v)+o(1),
  \label{eq:critical-regret-calculation}
\end{align}
uniformly over $v$ in the unit ball, which proves \eqref{eq:third-moment-mv-regret}. For the EUR rule, \eqref{eq:mv-eurr-expansion} has no deterministic $\beta(v)$ term because the estimated optimizer maximizes the full expected utility under the fitted distribution, and $n$ times the second moment of its remainder still vanishes when $r_{\mu,n}=n^{-1/4}$. The same quadratic calculation proves \eqref{eq:third-moment-eurr-regret}.

The lower-bound argument in Appendix~\ref{app:local-mv} uses only $r_{\mu,n}\to0$ and $nr_{\mu,n}^2\to\infty$, so \eqref{eq:mv-radial-derivative}, \eqref{eq:mv-local-quadratic-lower}, and \eqref{eq:mv-matrix-lower} apply when $r_{\mu,n}=n^{-1/4}$, and the prior contribution divided by $n$ equals $n^{-1/2}\mathcal I(f)\to0$. The argument therefore gives the lower bound $C_0$ for both infima under the present $\Theta_n$ and $H_n$, attained by the EUR rule. Taking a supremum and integrating the uniform expansion \eqref{eq:third-moment-mv-regret} gives the two additional constants claimed in Theorem~\ref{thm:third-moment-regret}.

\section{Moment Separation}\label{app:moments}
The finite-support construction below produces two distributions whose moments up to a given order coincide but whose expected utilities differ. Applying the construction to two assets shows that the optimal portfolio reverses between the two resulting models.

\begin{proof}[Proof of Proposition~\ref{prop:moment-separation}]
Because $U$ is not a polynomial of degree at most $L$, we can choose $L+2$ distinct points $y_0,\ldots,y_{L+1}$ in $J$ with a nonzero divided difference of order $L+1$. Let $A$ be the $(L+1)\times(L+2)$ matrix with entries $A_{k,j}=y_j^k$ for $k=0,\ldots,L$, and choose a nonzero vector $\xi=(\xi_j)_{j=0}^{L+1}$ in its null space. The nonzero divided difference ensures that
\begin{align}
  \sum_{j=0}^{L+1}\xi_jU(y_j)
  \neq
  0.
  \label{eq:utility-null-direction}
\end{align}
Choose strictly positive $p_j$ with $\sum_{j=0}^{L+1}p_j=1$, and then take $\eta>0$ small enough that every $p_j\pm\eta \xi_j$ is nonnegative. The row of $A$ with $k=0$ implies $\sum_{j=0}^{L+1}\xi_j=0$, so both sets of perturbed masses sum to one. The remaining rows give the equal moments in \eqref{eq:matching-moments}, whereas \eqref{eq:utility-null-direction} makes their expected utilities different.
\end{proof}

\begin{proof}[Proof of Corollary~\ref{cor:moment-insufficiency}]
Take $Q_+$ and $Q_-$ from Proposition~\ref{prop:moment-separation} with $L=2$, interpret $X-x_0$ as an asset return, and form two models with independent asset returns, assigning $Q_+$ to asset 1 and $Q_-$ to asset 2 in the first and reversing them in the second. Their mean vectors and covariance matrices coincide, while the expected utility ranking of the pure-asset portfolios is reversed.

By continuity, the disjoint neighborhoods $\calW_{2,1}$ and $\calW_{2,2}$ of the pure-asset portfolios in \eqref{eq:moment-separation-feasible-set} can be chosen so that the largest expected utility in the preferred neighborhood exceeds that in the other by at least $c>0$. Since any rule based only on the common mean and covariance selects $\calW_{2,1}$ with the same probability $\pi_1$ under both models, its regret is at least $c(1-\pi_1)$ in one model and at least $c\pi_1$ in the other, and its worst-case regret is therefore at least $c/2$.
\end{proof}

\section{The Taylor Criteria Used in the Experiments}
\label{app:taylor}
The experiments compare empirical expected utility with the criteria obtained by replacing $U$ with its Taylor polynomials at initial wealth. For an integer $L\geq1$, define
\begin{align}
  U_L(x_0+x)&=\sum_{k=0}^{L}\frac{U^{(k)}(x_0)}{k!}x^k,
  &V_{P,L}(w)&=\E_P[U_L(x_0+w^\top R)].
  \label{eq:taylor-criterion}
\end{align}
The empirical criterion replaces the expectation by the sample average. When $L=2$, the population criterion $V_{P,L}(w)$ equals $U(x_0)+U'(x_0)w^\top\mu+U''(x_0)\E_P[(w^\top R)^2]/2$. The second-order comparison in Section~\ref{sec:experiments} thus includes the squared mean, which the covariance-based criterion in \eqref{eq:sample-mean-variance-portfolio} does not.

Let $J$ be the interval of attainable wealth, and suppose that $U$ has $L+1$ continuous derivatives on $J$ and write $M_{U,L+1}=\sup_{z\in J}|U^{(L+1)}(z)|$. If this bound is finite, Taylor's theorem gives
\begin{align}
  |U(x_0+x)-U_L(x_0+x)|
  \leq\frac{M_{U,L+1}}{(L+1)!}|x|^{L+1}.
  \label{eq:taylor-remainder-bound}
\end{align}
This bound concerns approximation of the objective. Theorem~\ref{thm:local-mean-variance}, on the other hand, concerns concentration of optimal and estimated portfolios and the resulting loss in expected utility.

\section{Proofs for Allocations Based on an Observed State}
\label{app:sieve}

We first prove the validation inequality used to choose an approximation class and then derive the conditional estimation rate.

\begin{proof}[Proof of Proposition~\ref{prop:sieve-oracle}]
Fix $m\leq\overline m_n$. On the event $\calE_n$, the definition of $\widehat m$ and \eqref{eq:sieve-validation-event} give
\begin{align}
  \mathcal V_P(\widetilde\pi_{\widehat m})
  &\geq
  \widehat{\mathcal V}_{\widehat m}^{\mathrm{val}}-p_{n_2,\widehat m}
  \nonumber\\
  &\geq
  \widehat{\mathcal V}_{m}^{\mathrm{val}}-2p_{n_2,m}+p_{n_2,\widehat m}
  \nonumber\\
  &\geq
  \mathcal V_P(\widetilde\pi_m)-3p_{n_2,m}.
  \label{eq:sieve-selection-chain}
\end{align}
Subtracting \eqref{eq:sieve-selection-chain} from $\mathcal V_P(\pi_P^\star)$ and adding and subtracting $\mathcal V_P(\pi_{P,m}^\star)$ gives
\begin{align}
  r_P(\widetilde\pi_{\widehat m})
  \leq
  A_P(m)
  +
  \mathcal V_P(\pi_{P,m}^\star)-\mathcal V_P(\widetilde\pi_m)
  +
  3p_{n_2,m}.
  \label{eq:sieve-proof-fixed-model}
\end{align}
Taking the infimum over $1\leq m\leq\overline m_n$ proves \eqref{eq:sieve-oracle-inequality}. Because utility is bounded, regret is at most $\overline U-\underline U$, so splitting the expectation over $\calE_n$ and its complement proves the statement in expectation.
\end{proof}

\paragraph{Regret from estimating conditional expected utility.}
\label{app:dynamic}
We bound regret through the integrated error of the estimated conditional gradient and obtain the matching lower bound from sign perturbations on cells of the state space.

\begin{proof}[Proof of Theorem~\ref{thm:state-rate}]
Fix an admissible bandwidth $h$ and define
\begin{align}
  \widehat\pi_h(s)
  \in
  \operatorname*{arg\,max}_{w\in\calW_K}\widehat Q_h(s,w).
  \label{eq:conditional-estimated-maximizer}
\end{align}
Use the tie-breaking rule in Section~\ref{sec:model} to choose a measurable maximizer. Holding the training sample fixed, write
\begin{align}
  \varepsilon_h(s)
  =
  \sup_{w\in\calW_K}
  \norm{\widehat g_h(s,w)-g_P(s,w)}_*.
  \label{eq:conditional-directional-estimation-error}
\end{align}
Let $A_h=(\int\varepsilon_h(s)^{p_q}P_S(\rmd s))^{1/p_q}$ be its integrated norm.
For almost every $s$, $\pi_P^\star(s)$ maximizes the conditional criterion and the segment from $\pi_P^\star(s)$ to $\widehat\pi_h(s)$ lies in one portfolio class, so \eqref{eq:conditional-directional-estimation-error} and the argument for Proposition~\ref{prop:within-class} give
\begin{align}
  Q_P(s,\pi_P^\star(s))-Q_P(s,\widehat\pi_h(s))
  \leq
  \varepsilon_h(s)
  \norm{\widehat\pi_h(s)-\pi_P^\star(s)}.
  \label{eq:conditional-variational-step}
\end{align}
Integrating \eqref{eq:conditional-variational-step} and applying H\"older's inequality gives
\begin{align}
  r_P(\widehat\pi_h)
  \leq A_h
  \left(\int
    \norm{\widehat\pi_h(s)-\pi_P^\star(s)}^qP_S(\rmd s)
  \right)^{1/q}.
  \label{eq:state-variational-step}
\end{align}
The utility-loss inequality in Assumption~\ref{ass:state-smoothness} implies
\begin{align}
  c_0\int
    \norm{\widehat\pi_h(s)-\pi_P^\star(s)}^qP_S(\rmd s)
  \leq A_h
  \left(\int
    \norm{\widehat\pi_h(s)-\pi_P^\star(s)}^qP_S(\rmd s)
  \right)^{1/q}.
  \label{eq:state-distance-estimation}
\end{align}
The bound is trivial if the integral on the left side of \eqref{eq:state-distance-estimation} is zero. If it is positive, dividing \eqref{eq:state-distance-estimation} by the $1/q$ power of this integral and substituting the result into \eqref{eq:state-variational-step} gives
\begin{align}
  r_P(\widehat\pi_h)
  \leq
  c_0^{-1/(q-1)}A_h^{q/(q-1)}.
  \label{eq:state-regret-estimation}
\end{align}
Since $p_q=q/(q-1)$ and $r\geq p_q$ in Assumption~\ref{ass:state-gradient-estimation}, Lyapunov's inequality and \eqref{eq:state-gradient-bound} imply
\begin{align}
  \sup_{P\in\calP^{\mathrm{state}}}
  \E_P\left[r_P(\widehat\pi_h)\right]
  \lesssim
  \left(
    h^\beta+
    \left(n_{\mathrm{eff}}^{\mathrm{up}}h^p\right)^{-1/2}
  \right)^{q/(q-1)}.
  \label{eq:state-proof-upper}
\end{align}
Taking the infimum over $h\in\mathcal H_n$ proves \eqref{eq:state-upper-rate}, and balancing the two terms gives \eqref{eq:nonparametric-gradient-rate}.

For the lower bound, partition a fixed interior subset of $[0,1]^p$ into $M_h\asymp h^{-p}$ cells of side proportional to $h$ and apply on each the signed probability-mass perturbation of Assumption~\ref{ass:state-lower}, with signs in $\{-1,1\}^{M_h}$ and size $\gamma_h=c_\gamma h^\beta$, where $c_\gamma>0$ keeps every perturbed conditional distribution in the model class. Two sign vectors differing on one cell give distributions of the observed time series with Kullback--Leibler divergence at most
\begin{align}
  C n_{\mathrm{eff}}^{\mathrm{low}}\gamma_h^2h^p.
  \label{eq:state-one-cell-information}
\end{align}
If $h$ is chosen so that \eqref{eq:state-one-cell-information} is at most a small constant, Assouad's lemma shows that the expected number of wrongly recovered signs is at least a multiple of $M_h$. Since each wrong sign costs regret of order $\gamma_h^{q/(q-1)}h^p$ and $M_hh^p\asymp1$, every rule then has worst-case regret of order at least $\gamma_h^{q/(q-1)}$. This choice of $h$ satisfies
\begin{align}
  n_{\mathrm{eff}}^{\mathrm{low}}h^{2\beta+p}
  \asymp
  1.
  \label{eq:state-information-balance}
\end{align}
Thus, the lower bound has order
\begin{align}
  \left(n_{\mathrm{eff}}^{\mathrm{low}}\right)^{-\beta q/((2\beta+p)(q-1))}.
  \label{eq:state-proof-lower}
\end{align}
When $n_{\mathrm{eff}}^{\mathrm{low}}\asymp n_{\mathrm{eff}}^{\mathrm{up}}$, \eqref{eq:state-proof-upper} and \eqref{eq:state-proof-lower} have the same order, proving \eqref{eq:nonparametric-state-rate}.
\end{proof}

\section{Proofs of Minimax and Bayes Optimality of the EUR Rule}
\label{app:likelihood-global}

Throughout this appendix, the number of classes and the parameter dimension are fixed, and constants are uniform over $\theta\in\Theta$ and may change between displays. Because the compact estimation set is larger than the evaluation set, a maximum likelihood estimate near the true parameter lies in the interior of the estimation set even when $\theta\in\partial\Theta$.

\subsection{Likelihood Concentration and Expected Allocation Regret}

The density bounds give uniform bounds for the log likelihood and its first four parameter derivatives, and compactness and identifiability give a positive Kullback--Leibler gap outside any fixed neighborhood of $\theta$. In a sample of size $N$, concentration on a fine finite net of the estimation set and the local Lipschitz bound for the log likelihood therefore keep every global likelihood maximizer in a fixed small neighborhood of $\theta$, except on an event of probability at most $C\exp(-cN)$.

Write $\bar S_N=N^{-1}\sum_i S_\theta(R_i)$ for the average score and $\widehat I_N=-N^{-1}\sum_i\nabla^2\log p_\theta(R_i)$ for the observed information. As averages of bounded centered terms, $\bar S_N$ and $\widehat I_N-I_\theta$ have sub-Gaussian tails and $r$th moments of order $N^{-r/2}$. On an event with probability at least $1-C\exp(-cN)$, the score equation and uniform negative curvature give a unique likelihood maximizer and the expansion
\begin{align}
  \widehat\theta_N-\theta
  &=I_\theta^{-1}\bar S_N+\rho_N,
  &
  \norm{\rho_N}_2
  &\leq C\left(\norm{\bar S_N}_2^2+\norm{\widehat I_N-I_\theta}\norm{\bar S_N}_2\right).
  \label{eq:likelihood-score-expansion}
\end{align}
Since the estimator's range is compact, the complement of the event is negligible and, for every fixed $r\geq1$,
\begin{align}
  \sup_{\theta\in\Theta}\E_\theta[\norm{\rho_N}_2^r]&\leq C_rN^{-r},
  &
  \sup_{\theta\in\Theta}\Pr_\theta(\norm{\widehat\theta_N-\theta}_2>t)
  &\leq C\exp(-cNt^2).
  \label{eq:likelihood-parameter-bounds}
\end{align}
For $t$ bounded away from zero, the second inequality follows from global likelihood separation.

Strong concavity in tangent coordinates, interiority of the optimizers, and the implicit function theorem imply that the maps $\theta\mapsto w_m(\theta)$ have uniformly bounded derivatives up to the orders used below. When the optimized expected utility is differentiated as in Lemma~\ref{lem:optimized-value-if}, the estimated weights $\widehat w_m$ in tangent coordinates and the estimated optimized expected utilities $\widehat B_m$ satisfy
\begin{align}
  N_m^\top(\widehat w_m-w_m(\theta))
  &=D_m(\theta)I_\theta^{-1}\bar S_N+\rho_{m,N}^{w},
  \nonumber\\
  \widehat B_m-B_m(\theta)
  &=\nabla B_m(\theta)^\top I_\theta^{-1}\bar S_N+\rho_{m,N}^{B},
  \label{eq:likelihood-weight-value-expansions}
\end{align}
where the $r$th moments of both remainders are $O(N^{-r})$ uniformly. A second-order expansion of the utility loss within class $m$ gives, with $r_m(w,\theta)=B_m(\theta)-V_{P_\theta}(w)$,
\begin{align}
  \sup_{\theta\in\Theta}
    \abs{N\E_\theta[r_m(\widehat w_m,\theta)]-c_m(\theta)}&\longrightarrow0,
  \nonumber\\
  \sup_{N\geq1,\theta\in\Theta}
    \E_\theta[(N r_m(\widehat w_m,\theta))^2]&<\infty.
  \label{eq:likelihood-allocation-bounds}
\end{align}
The leading tangent term has covariance $D_m I_\theta^{-1}D_m^\top/N$, which proves the allocation formula in \eqref{eq:likelihood-constants}.

For every utility difference $d_{m\ell}=B_m-B_\ell$, the same argument yields
\begin{align}
  \sup_{\theta\in\Theta}
  \Pr_\theta(\abs{\widehat d_{m\ell}-d_{m\ell}(\theta)}>t)
  \leq C\exp(-cNt^2)
  \label{eq:likelihood-contrast-tail}
\end{align}
throughout its bounded range. Near a regular tie between two optimal classes, $s_{m\ell}^2$ is bounded away from zero, the bounded score summands in \eqref{eq:likelihood-weight-value-expansions} satisfy the Berry--Esseen inequality uniformly, and the normalized remainder has expected absolute value $O(N^{-1/2})$. Treating separately the event that this remainder exceeds $N^{-1/4}$ and using the boundedness of the normal density, we obtain a normal approximation for the normalized utility difference that is uniform in $\theta$ and also bounds the probability that the two estimated values tie exactly.

\subsection{The Finite Gaussian Comparison}
We first prove the closed form of the two-class value.

\begin{proof}[Proof of Lemma~\ref{lem:two-class-gaussian-value}]
Write $\sigma=\sqrt\Omega$ and let $u^\star$ attain $\kappa$. The sign rule, which selects the class with the larger observation, has expected gap loss $\abs{\mu}\Phi(-\abs{\mu}/\sigma)$ under mean $\mu$, with supremum $\kappa\sigma$ over $\mu\in\bbR$, so $\gamma_C(\Omega)\leq\kappa\sigma$. Conversely, under the prior with mass $1/2$ at each of $\mu=\pm u^\star\sigma$, both points give the gap $u^\star\sigma$, so a rule's Bayes risk is $u^\star\sigma$ times its probability of selecting the class with the smaller mean. The Bayes rule for this prior minimizes that probability by selecting on the sign of the observation, with error probability $\Phi(-u^\star)$. Hence every rule has Bayes risk, and so supremum risk, at least $u^\star\sigma\Phi(-u^\star)=\kappa\sigma$.
\end{proof}

For three or more classes, we first bound the discretization error for a candidate set $C$ and a covariance satisfying $\tau_n^2 I\preceq\Omega\preceq\Lambda I$. Selecting the class with the largest observation has loss at most $2\norm{G}_\infty$, so $\gamma_C(\Omega)$ is uniformly bounded. Now take a nearly minimax rule for the Gaussian comparison and replace its action probabilities by their averages over each observation cell. Because a Gaussian density differs in $L^1$ from its cellwise average by at most the cell diameter times the $L^1$ norm of its gradient, which is at most $C/\tau_n$, this averaging changes the expected loss at every mean in the box $[-\rho_n,\rho_n]^q$ by at most $C\rho_n\epsilon_n/\tau_n$.

Every mean coordinate lies at least $t_n$ inside the box of observations, so an observation falls outside it with probability at most $C\exp(-t_n^2/(2\Lambda))$, and selecting the reference class $m_0$ there costs at most $C\rho_n\exp(-t_n^2/(2\Lambda))$. Moving a mean to its nearest grid point changes the loss by at most twice the displacement and the Gaussian distribution by at most $C\epsilon_n/\tau_n$ in total variation. Hence the solution of \eqref{eq:likelihood-gaussian-program} satisfies
\begin{align}
  \sup_{\mu\in[-\rho_n,\rho_n]^q}
  \E[L_C(\mu,\delta_n^{G}(\mu+G))]
  &\leq\gamma_C(\Omega)+e_n,
  \nonumber\\
  e_n&\leq C\epsilon_n(1+\rho_n/\tau_n)
           +C\rho_n\exp(-t_n^2/(2\Lambda))=o(1).
  \label{eq:likelihood-program-error}
\end{align}
Here $\delta_n^G$ is the rule defined on the cells before the threshold $\log(n+1)$ is applied. Because the set of feasible probability arrays is compact, the optimum is attained, and minimizing successively over the coordinates selects the lexicographic solution, which is measurable because the Gaussian probabilities of the cells are continuous in the covariance.

Put $\ell_n=\log(n+1)$ and append the noise coordinate $z_{m_0}=0$ of the reference class. For any observation $y=\mu+z$, the final action $\delta_n^{T}(y)$ of the thresholded rule satisfies
\begin{align}
  L_C(\mu,\delta_n^{T}(\mu+z))&\leq\ell_n+2\norm{z}_\infty.
  \label{eq:likelihood-threshold-loss}
\end{align}
Replacing a proposed action by the observed maximizer increases its true loss by at most $(2\norm{z}_\infty-\ell_n)_+$, since replacement requires an observed advantage above $\ell_n$ and the true disadvantage is at most the noise difference minus that advantage. Under any covariance bounded by $\Lambda I$, the expected increase is at most $C\exp(-c\ell_n^2)$, so \eqref{eq:likelihood-program-error} remains valid for the thresholded rule.

\subsection{Sampling Approximation with an Estimated Covariance}

For fixed $C$, let $G_C(\theta)=\nabla d_C(\theta)$ and $\xi_{C,\theta}=G_C(\theta)I_\theta^{-1}S_\theta(R)$, a bounded centered vector with covariance $\Omega_C(\theta)$. The utility differences computed from the second sample satisfy
\begin{align}
  \sqrt{N_n}(d_C(\widehat\theta(D_1))-d_C(\theta))
  =N_n^{-1/2}\sum_{i\in D_1}\xi_{C,\theta}(R_i)+r_{C,n},
  \qquad \sup_{\theta\in\Theta}\E_\theta[\norm{r_{C,n}}_2]\leq C N_n^{-1/2}.
  \label{eq:likelihood-contrast-linear}
\end{align}
Adding the independent vector $\tau_n Z_C$ yields a normal approximation in total variation. Indeed, because the convolution of a bounded test function with the $N(0,\tau_n^2 I)$ density has third derivatives of order $\tau_n^{-3}$, replacing the bounded score summands one at a time by Gaussian summands with the same covariance, as in Lindeberg's argument, leaves an error of at most $C N_n^{-1/2}\tau_n^{-3}$. The same smoothing and \eqref{eq:likelihood-contrast-linear} bound the contribution of the remainder $r_{C,n}$ by $C N_n^{-1/2}\tau_n^{-1}$.

Since the map $\theta\mapsto\Omega_C(\theta)$ is Lipschitz near the evaluation set and both regularized covariances have eigenvalues at least $\tau_n^2$, a relative entropy bound shows that their Gaussian distributions differ in total variation by at most
$C\min\{1,\norm{\Omega_C(\widehat\theta(D_0))-\Omega_C(\theta)}/\tau_n^2\}$.
Taking expectation over the preliminary sample and applying \eqref{eq:likelihood-parameter-bounds} bounds the total error by
\begin{align}
  \omega_n=C\left(N_n^{-1/2}(\tau_n^{-3}+\tau_n^{-1})
                       +k_n^{-1/2}\tau_n^{-2}\right),
  \qquad \omega_n\log(n+1)\longrightarrow0.
  \label{eq:likelihood-smoothed-distance}
\end{align}
Conditional on the preliminary sample, the second sample and the added Gaussian vector are independent of the estimated covariance and of the cell rule computed from it, so the bound applies to this random rule, uniformly over the finitely many candidate sets.

To pass from total variation to expected loss, we restrict attention to centered errors below a large multiple of $\log(n+1)$, where \eqref{eq:likelihood-threshold-loss} bounds the loss by $C\log(n+1)$. Outside this range, the sub-Gaussian tails in \eqref{eq:likelihood-contrast-tail} and the bounded covariances of the added noise and of the Gaussian comparison make the contribution uniformly negligible both in probability and in first moment. Hence, whenever the candidate set contains an optimal class, \eqref{eq:likelihood-program-error} and \eqref{eq:likelihood-smoothed-distance} bound the normalized expected selection regret of the EUR rule by the expected Gaussian minimax constant of the candidate set plus $o(1)$, uniformly over the box of means. In particular, the expected selection regret is uniformly $O(n^{-1/2})$. With two candidates, the scalar normal approximation and \eqref{eq:likelihood-contrast-tail} give the same bound for the comparison of the two estimated values, and with one candidate the selection loss is zero.

\subsection{Gaussian Comparison and the Worst-Case Bound}

We record three properties of \eqref{eq:likelihood-gaussian-value}. First, adding independent Gaussian noise cannot improve the minimax value because a rule for the noisier experiment can be applied to the original observations after adding that noise. This monotonicity, together with the identity $\gamma_C(t\Omega)=\sqrt t\,\gamma_C(\Omega)$ for $t>0$, which follows by rescaling the means, implies continuity of $\gamma_C$: if $\Omega_j\to\Omega\succ0$, then $\Omega_j$ lies between $(1-\varepsilon_j)\Omega$ and $(1+\varepsilon_j)\Omega$ with $\varepsilon_j\to0$.

Second, the value is the supremum of Bayes values over finite priors on the means. Indeed, relative to a positive reference density on $\bbR^{\abs{C}-1}$, the randomized kernels with finitely many actions form a weak-* compact convex subset of a finite product of $L^\infty$ spaces, on which the expected loss at each mean is continuous and linear because the likelihood ratios times the losses are integrable. For any number below the minimax value, compactness yields finitely many of the expected-loss constraints that no kernel satisfies simultaneously, and minimax duality for finitely many parameters then gives a finite prior whose Bayes value exceeds that number. The reverse inequality is immediate.

Third, adding competing alternatives cannot decrease the value for the marginal experiment on the utility differences. Take $C\subset C'$ with a common reference class in $C$, fix a finite prior on the $C$ means, and draw the additional means independently and uniformly on $[-T^2-T,-T^2+T]$. For large $T$, they lie below the value of every class in $C$, so the reference class dominates any selection among them. Conditional on the Gaussian errors, shifting a uniform box by its additional error coordinates changes its distribution by $O(\norm{G_{C'\setminus C}}_1/T)$ on average, so the observations converge in total variation to the $C$ observations plus independent uniform coordinates. Since losses are bounded under the finite prior, the Bayes value of the enlarged experiment converges to that of $C$, and taking the supremum over finite priors proves
\begin{align}
  \gamma_C(\Omega_C)\leq\gamma_{C'}(\Omega_{C'}),
  \label{eq:likelihood-gaussian-subset}
\end{align}
where the covariances come from the same Gaussian utility differences. Changing the reference class is an invertible change of observation coordinates and therefore does not affect this inequality.

\paragraph{The preliminary comparison.}
Let $\mathcal E_n$ be the event that every optimized value estimated from the preliminary sample has error at most $\sigma_Ub_n/4$, a quarter of the screening threshold in \eqref{eq:likelihood-retained-classes}. Equations~\eqref{eq:likelihood-parameter-bounds} and \eqref{eq:likelihood-contrast-tail}, applied also to individual values, imply
\begin{align}
  \sup_{\theta\in\Theta}\Pr_\theta(\mathcal E_n^c)
  \leq C\exp(-c k_n b_n^2),
  \label{eq:likelihood-pilot-bound}
\end{align}
with $\sigma_U^2$ absorbed into $c$. The exponents in Algorithm~\ref{rule:eurr} give $k_nb_n^2\asymp k_n^{1/4}$, so $k_nb_n^2/\log n\to\infty$. Consequently, multiplying the probability of $\mathcal E_n^c$ by $n$ still yields a quantity that tends to zero. On $\mathcal E_n$, the candidate set contains every optimal class and every class with gap at most $\sigma_Ub_n/2$, and every candidate has gap at most $3\sigma_Ub_n/2$.

\paragraph{The supremum over the evaluation set.}
To control the supremum over $\Theta$, take $\theta_n\in\Theta$ and a subsequence converging to $\theta_0\in\Theta$. On $\mathcal E_n$, every candidate eventually belongs to $A(\theta_0)$ because all other population gaps have positive limits, and $\widehat\Omega_C\to\Omega_C(\theta_0)$, which is positive definite by full row rank. Continuity, \eqref{eq:likelihood-gaussian-subset}, and the sampling comparison bound the limit superior of the normalized expected selection regret by $\gamma_{A(\theta_0)}(\Omega_{A(\theta_0)}(\theta_0))$, and exceptional events of vanishing probability do not matter because the Gaussian values are uniformly bounded. With two candidates, the comparison of the two estimated values gives the same bound with value $\kappa s_{m\ell}$, where \eqref{eq:likelihood-contrast-tail} controls the gaps that lie outside a bounded interval after normalization. At a unique optimum, the candidate set is that class alone with exponentially high probability. With the $O(n^{-1})$ allocation regret from \eqref{eq:likelihood-allocation-bounds}, this proves the upper bound in \eqref{eq:likelihood-minimax}. Without ties in $\Theta$, compactness gives a uniformly positive gap to every inferior class, and the same argument gives the upper bound in \eqref{eq:likelihood-unique-minimax}.

\subsection{A Minimax Lower Bound in the Original Return Model}

Fix an interior tie point $\theta_0$ and let $C=A(\theta_0)$, $G=\nabla d_C(\theta_0)$, $I=I_{\theta_0}$, and $\Omega=G I^{-1}G^\top$. For a mean $\mu$ of the utility differences, $v(\mu)=I^{-1}G^\top\Omega^{-1}\mu$ satisfies $Gv(\mu)=\mu$, and for any finite collection of means, the local parameters $\theta_{n,\mu}=\theta_0+v(\mu)/\sqrt n$ lie in $\Theta$ for large $n$ and have utility differences $\mu/\sqrt n+O(n^{-1})$.

The bounded likelihood derivatives give
\begin{align}
  \log\frac{\rmd P_{\theta_{n,\mu}}^{\otimes n}}
                 {\rmd P_{\theta_0}^{\otimes n}}
  =\mu^\top\Omega^{-1}Z_n-\frac12\mu^\top\Omega^{-1}\mu+o_P(1),
  \qquad
  Z_n=G I^{-1}n^{-1/2}\sum_i S_{\theta_0}(R_i)
       \rightsquigarrow N(0,\Omega).
  \label{eq:likelihood-efficient-local-experiment}
\end{align}
The remainders are uniform over the finite collection. Since the moment of order $1+\epsilon$ of the likelihood ratio for one observation is $1+O(n^{-1})$, independence makes the likelihood ratios uniformly integrable. The minimum expected loss under the finite prior is a $P_{\theta_0}^{\otimes n}$ expectation of the minimum of finitely many losses weighted by likelihood ratios, and hence it converges to the corresponding Gaussian expected regret under that prior.

We now extend this bound from rules that select a class to arbitrary portfolio choices. Each optimizer of a class in $C$ belongs to no other class, since otherwise it would also optimize that class, which contradicts the distinctness of the optimizers and their uniqueness within each class. By compactness, these optimizers therefore have disjoint fixed neighborhoods, each meeting only its own class. Let $j(w)$ be the index of the class whose neighborhood contains $w$, and a fixed index in $C$ if no neighborhood contains $w$. Outside all neighborhoods, regret is bounded away from zero, also under the local alternatives, whereas inside the neighborhood of class $j$ it is at least the best class value minus $B_j$. Hence, for every action and mean in the collection,
\begin{align}
  \sqrt n\,r_{P_{\theta_{n,\mu}}}(w)
  \geq L_C(\mu,j(w))-O(n^{-1/2}).
  \label{eq:likelihood-action-test-reduction}
\end{align}
Outside the neighborhoods this holds for large $n$ because the finitely many utility-gap losses are bounded, and it extends to randomized actions by taking expectations. Hence the Gaussian Bayes value under every finite prior, and thus $\gamma_C(\Omega)$, is a lower bound for the limit inferior of the normalized minimax expected regret. At ties on $\partial\Theta$, transversality supplies interior tie points with the same optimal classes that converge to them, and taking the supremum over tie points proves the lower bound in \eqref{eq:likelihood-minimax}.

\subsection{Integration under the Original Prior}

At a parameter with a unique optimal class, that class is selected with probability tending to one, and \eqref{eq:likelihood-allocation-bounds} with Cauchy--Schwarz shows that selecting it from the same second sample that estimates its weights leaves its allocation constant unchanged. The total allocation loss is at most the sum of the nonnegative class-specific losses, whose normalized expectations the same equation bounds uniformly. Since the set of ties has Lebesgue measure zero, dominated convergence proves the allocation integral in \eqref{eq:likelihood-bayes} for the bounded continuous density $h$.

\paragraph{Multiple nearly optimal classes.}
To bound regions with multiple nearly optimal classes when there are at least three classes, let $E_b$ be the set where the largest and third-largest optimized values differ by at most $3\sigma_Ub/2$. Near any limiting triple tie, Assumption~\ref{ass:regular-return-model}(iii) makes two utility differences among those classes local coordinates, and compactness gives finitely many such charts with uniformly bounded Jacobians. On $E_b$ both coordinates lie in intervals of length $Cb$, so we have
\begin{align}
  \operatorname{Leb}(E_b)\leq Cb^2
  \label{eq:likelihood-multiple-tube-volume}
\end{align}
for small $b$, after a union over the finitely many triples. If the parameter dimension is too small for a triple tie, $E_b$ is empty for small $b$.

On $\mathcal E_n$, three or more candidates require $\theta\in E_{b_n}$, so the uniform $O(n^{-1/2})$ selection bound and boundedness of $h$ give
\begin{align}
  n\int_{E_{b_n}}\E_\theta[\text{selection loss}]h(\theta)\rmd\theta
  \leq C\sqrt n\,b_n^2+o(1)\longrightarrow0.
  \label{eq:likelihood-multiple-bayes-negligible}
\end{align}
A failed preliminary comparison contributes $o(1)$ after normalization by \eqref{eq:likelihood-pilot-bound}.

\paragraph{Two leading classes.}
Outside $E_{b_n}$, a successful preliminary comparison leaves at most the two leading classes as candidates, and both are candidates when their population gap is at most $\sigma_Ub_n/2$. For a larger gap, omitting the second class changes the expected selection loss by at most $C\exp(-cN_n b_n^2)$ by \eqref{eq:likelihood-contrast-tail}, since the sole optimal candidate has zero selection loss, and this is negligible after integration and multiplication by $n$.

Partition the remaining parameters according to their two leading classes. Near a regular piece of a tie surface $\calT_{m\ell}$, where exactly $m$ and $\ell$ are optimal, apply the coarea formula with $x=d_{m\ell}(\theta)$. For $x=v/\sqrt{N_n}$ with bounded $v$, the scalar normal approximation gives error probability $\Phi(-\abs{v}/s_{m\ell}(\tau))$ at each $\tau\in\calT_{m\ell}$, and \eqref{eq:likelihood-contrast-tail} bounds the normalized integrand by $C\abs{v}\exp(-cv^2)$.

These bounds extend the integration beyond compact subsets of individual tie surfaces. Since $\nabla d_{m\ell}\neq0$ at every limit point of $\calT_{m\ell}$, finitely many charts whose coordinates are utility differences bound, uniformly, the Jacobians and the areas of the sets on which each utility difference is constant. By the rank condition, sets with more than two optimal classes have codimension at least two and meet the closure of any tie surface in a $(p-1)$-dimensional null set, as does $\partial\Theta$ by transversality. Hence the indicator of the two leading classes and of membership in $\Theta$ converges at almost every point of the tie surface. Dominated convergence in these charts, integration over $v$, and \eqref{eq:likelihood-multiple-bayes-negligible} for the excluded sets $E_{b_n}$ give for each pair the limit
\begin{align}
  &\int_{\mathcal T_{m\ell}}
  \frac{h(\tau)}{\norm{\nabla d_{m\ell}(\tau)}_2}
  \left(\int_{\bbR}\abs{v}\Phi(-\abs{v}/s_{m\ell}(\tau))\rmd v\right)
  \rmd\mathcal H^{p-1}(\tau)
  \nonumber\\
  &\qquad=\frac12\int_{\mathcal T_{m\ell}}
  \frac{s_{m\ell}^2(\tau)h(\tau)}{\norm{\nabla d_{m\ell}(\tau)}_2}
  \rmd\mathcal H^{p-1}(\tau).
  \label{eq:likelihood-coarea-limit}
\end{align}
Here we used $n/N_n\to1$ and $\int_{\bbR}\abs{v}\Phi(-\abs{v}/\sigma)\rmd v=\sigma^2/2$. Away from the optimal tie set, inferior gaps are uniformly positive on compact subsets and contribute exponentially small amounts. Summing over pairs proves the selection part of the Bayes upper bound.

\subsection{The Bayes Lower Bound and the No-Tie Minimax Value}

For the given prior and the full sample $D_n$, let $Q_n(w)=\E_H[V_{P_\vartheta}(w)\mid D_n]$ and choose a measurable maximizer over the original feasible set. It is an exact Bayes decision, because the posterior expected oracle value does not depend on the action and randomization cannot improve the maximum.

At an interior parameter with $h(\theta)>0$, expand the log likelihood about its maximizer. In the coordinate $z=\sqrt n(\vartheta-\widehat\theta_n)$, its difference from the maximum is $-z^\top I_\theta z/2+o_P(1)$ on bounded sets, where the prior is asymptotically constant. Uniform local negative curvature gives Gaussian posterior tails in $z$, global likelihood separation makes the posterior mass outside a fixed neighborhood exponentially small, and the posterior normalizing constant has order $n^{-p/2}$ times a positive local density. Hence the local posterior density and its first two centered moments converge, and in particular
\begin{align}
  \E_H[\vartheta-\widehat\theta_n\mid D_n]&=o_P(n^{-1/2}),
  &
  \E_H[\norm{\vartheta-\widehat\theta_n}_2^2\mid D_n]&=O_P(n^{-1}).
  \label{eq:likelihood-posterior-moments}
\end{align}
This is uniform over bounded $n^{-1/2}$ perturbations of the interior point. Taylor expansion of expected utility and its weight gradient, with uniformly bounded parameter derivatives, gives
\begin{align}
  \sup_{w\in\calW_K}\abs{Q_n(w)-V_{P_{\widehat\theta_n}}(w)}
  +\sup_{w\in\calW_K}\norm{\nabla_w Q_n(w)-\nabla_w V_{P_{\widehat\theta_n}}(w)}_2
  =o_P(n^{-1/2}).
  \label{eq:likelihood-posterior-utility}
\end{align}
The posterior-averaged negative Hessian has eigenvalues bounded away from zero in tangent coordinates, so the resulting class optimizers differ from $w_m(\widehat\theta_n)$ by $o_P(n^{-1/2})$.

On unique-optimum regions this gives the efficient normal limit for the exact Bayes portfolio, and Fatou's lemma for the nonnegative quadratic regret gives $c_{m(\theta)}(\theta)$ pointwise as a lower bound for the limit inferior of $n$ times its expected regret. On bounded normal-coordinate intervals about regular pieces of the tie surfaces, \eqref{eq:likelihood-posterior-utility} gives the sign-comparison limit of \eqref{eq:likelihood-coarea-limit}, and Fatou's lemma for the nonnegative selection losses, with exhaustion of the intervals and regular pieces, gives $C_{\mathrm{sel}}(H)$. Points with $h=0$ contribute zero, and by nonnegativity the sets with more than two optimal classes can be omitted. For large $n$, the shrinking neighborhoods of the tie surfaces are disjoint from any fixed compact unique-optimum region where the prior is positive, so the two lower bounds can be added before their domains are exhausted. Since the exact Bayes decision minimizes integrated expected regret at each $n$, this lower bound for $\mathcal V_{n,K}^H$ matches the upper bound and proves \eqref{eq:likelihood-bayes}.

If $\mathcal T$ is empty, the worst-case expected regret is at least its average under any fixed prior. Applying the Bayes result to smooth priors supported near interior points whose allocation constants approach $\sup_{\theta\in\Theta}c_{m(\theta)}(\theta)$, with $n\to\infty$ before the supports shrink, proves \eqref{eq:likelihood-unique-minimax} by continuity of $c_{m(\theta)}(\theta)$.

\section{Portfolio Approximation}
\label{app:portfolio-approximation}
We prove Proposition~\ref{prop:portfolio-approximation} in the orthonormal coordinates of the class. Write $x=N^\top(w-w_P^\star)$ and suppress $P$. Since the gradient in the tangent coordinates is zero at the interior optimum, strong concavity gives $\ell(w)\leq \zeta(w)^\top(-x)-\lambda\norm{x}_2^2/2\leq\norm{\zeta(w)}_2^2/(2\lambda)$. Bounded third derivatives and the eigenvalue bounds give
\begin{align}
 \zeta(w)&=-J(w)x+O(\norm{x}_2^2),
 &\ell(w)&=\frac12x^\top J(w)x+O(\norm{x}_2^3).
 \label{eq:approximation-proof-expansion}
\end{align}
Since strong concavity also gives $\norm{\zeta(w)}_2\geq\lambda\norm{x}_2$, substituting $x=-J(w)^{-1}\zeta(w)+O(\norm{\zeta(w)}_2^2)$ proves the expansion in \eqref{eq:utility-loss-gradient}.

For feasible $\widehat w$ and $\widetilde w$, the fundamental theorem of calculus along their segment and the upper Hessian bound imply
\begin{align}
 |\ell(\widetilde w)-\ell(\widehat w)|
 \leq\Lambda\norm{\widehat w-w_P^\star}_2\norm{\widetilde w-\widehat w}_2
       +\frac\Lambda2\norm{\widetilde w-\widehat w}_2^2.
 \label{eq:approximation-proof-transfer}
\end{align}
Taking expectations and applying Cauchy--Schwarz proves \eqref{eq:portfolio-approximation-transfer}.

\section{Regular Likelihoods in One Portfolio Class}
\label{app:regular-one-class}
This appendix establishes the one-class versions of Theorems~\ref{thm:likelihood-global} and~\ref{thm:likelihood-bayes} under Assumption~\ref{ass:regular-one-class}.

\subsection{Expected Regret and the Lower Bound}
The first-order condition for the interior optimum and the implicit function theorem give a twice continuously differentiable map $w_1(\theta)$. Put $D_\theta=\dot w_1(\theta)$ and $J_\theta=-\nabla^2V_{P_\theta}(w_1(\theta))$. Taylor expansion and \eqref{eq:regular-estimator-ltwo}--\eqref{eq:regular-estimator-fourth} yield
\begin{align}
 w_1(\widehat\theta_n)-w_1(\theta)
 &=D_\theta I_\theta^{-1}\frac1n\sum_{i=1}^nS_\theta(R_i)+u_n(\theta),
 &\sup_{\theta\in\Theta}n\E_\theta[\norm{u_n(\theta)}_2^2]&\longrightarrow0.
 \label{eq:one-class-weight-expansion}
\end{align}
The fourth moment of the weight error is $O(n^{-2})$. By Markov's inequality and \eqref{eq:regular-estimator-fourth}, the estimates leave the neighborhood used for the Taylor formula with probability $O(n^{-2})$, and compactness of the parameter and decision sets bounds the regret on that event. A Taylor expansion of expected utility gives
\begin{align}
 \sup_{\theta\in\Theta}
 \left|n\E_\theta[r_{P_\theta}(w_1(\widehat\theta_n))]
 -\frac12\operatorname{tr}(J_\theta D_\theta I_\theta^{-1}D_\theta^\top)\right|
 \longrightarrow0.
 \label{eq:one-class-regret-expansion}
\end{align}
The cubic loss remainder has expectation $O(n^{-3/2})$ and the $L^2$ remainder in \eqref{eq:one-class-weight-expansion} contributes $o(n^{-1})$. When only one class is available, the parameter estimate is computed from all $n$ observations, so these expansions give the upper bounds under both evaluations.

\paragraph{Lower bound over all rules.}
We obtain the integrated lower bound over all measurable rules using the matrix form of the van Trees argument \citep{Gill1995vantrees}. Fix an interior $\theta_0$ with $w_0=w_1(\theta_0)$ and $\varepsilon>0$. By continuity, every relevant Hessian is at least $(1-\varepsilon)J_{\theta_0}$ on a ball around $w_0$ for $\theta$ near $\theta_0$. Let $T$ be the projection onto a smaller ball in the $J_{\theta_0}$ norm, and shrink the parameter neighborhood until all $w_1(\theta)$ lie in the concentric ball of half its radius. Projection does not increase the $J_{\theta_0}$ distance to $w_1(\theta)$ inside the larger ball. Outside the larger ball, strong concavity gives a lower bound on regret that exceeds every projected quadratic loss when the projection radius is small. Hence, for all feasible $w$ and parameters in that neighborhood,
\begin{align}
 r_{P_\theta}(w)
 \geq\frac{1-\varepsilon}{2}
 (T(w)-w_1(\theta))^\top J_{\theta_0}(T(w)-w_1(\theta)).
 \label{eq:one-class-local-loss-lower}
\end{align}
The construction does not depend on the rule or the sample size.

Let $\varphi$ be a continuously differentiable probability density supported inside this neighborhood, zero on its boundary, with finite $\mathcal I(\varphi)=\int(\nabla\varphi)(\nabla\varphi)^\top\varphi^{-1}\,\rmd\theta$, where integrals involving $\varphi$ are taken over its support. Under the joint distribution with parameter density $\varphi$ and $n$ return observations, write $\mathbb E_\varphi$ for expectation and $S=\sum_{i=1}^n S_\theta(R_i)+\nabla\log\varphi(\theta)$. Integration by parts and the zero mean of the likelihood score give
\begin{align}
 \mathbb E_\varphi[(T(\widehat w)-w_1(\theta))S^\top]&=\overline D_\varphi,
 &\overline D_\varphi&=\int D_\theta \varphi(\theta)\,\rmd\theta,\nonumber\\
 \mathbb E_\varphi[SS^\top]&=n\overline I_\varphi+\mathcal I(\varphi),
 &\overline I_\varphi&=\int I_\theta \varphi(\theta)\,\rmd\theta.
 \label{eq:one-class-joint-score}
\end{align}
The identities hold for every bounded measurable output, including randomization independent of the parameter, because they differentiate the joint density and $w_1$, not the estimator. Matrix Cauchy--Schwarz applied to \eqref{eq:one-class-joint-score}, followed by \eqref{eq:one-class-local-loss-lower}, gives
\begin{align}
 n\int\calR_n(\delta_n,P_\theta)\varphi(\theta)\,\rmd\theta
 \geq\frac{1-\varepsilon}{2}
 \operatorname{tr}\left(J_{\theta_0}\overline D_\varphi
             (\overline I_\varphi+\mathcal I(\varphi)/n)^{-1}\overline D_\varphi^\top\right).
 \label{eq:one-class-integrated-lower}
\end{align}
The right side does not depend on the rule, and letting $n\to\infty$ and then taking densities supported close to $\theta_0$ makes it arbitrarily close to $c_1(\theta_0)$.

For the given continuous density $h$, partition compact subsets of its positive-density region in $\operatorname{int}(\Theta)$ into finitely many small disjoint cubes, and on each choose a smooth nonnegative density below $h$ with mass arbitrarily close to that of $h$ and finite prior information after normalization. Since regret is nonnegative, summing \eqref{eq:one-class-integrated-lower} over the normalized densities, weighted by their masses, bounds the integral under $h$ below. By uniform continuity of $D_\theta$, $I_\theta$, and $J_\theta$, the sum approaches $\int_\Theta c_1(\theta)h(\theta)\,\rmd\theta$ as the cubes shrink and exhaust the positive-density region, since the boundary is Lebesgue null. This proves the Bayes lower bound for every $H\in\mathfrak H$. Smooth priors concentrated near interior points whose allocation constants approach $\sup_{\theta\in\Theta}c_1(\theta)$ give the minimax lower bound, with $n\to\infty$ before the auxiliary supports shrink in both cases. With \eqref{eq:one-class-regret-expansion}, these bounds prove the two one-class conclusions.

\subsection{Verification for the Signed-Gamma Family}
For $r_a\ne0$, the joint density of the example in Section~\ref{sec:estimated-volatility} is
\begin{align}
 p_\theta(r)=\prod_{a=1}^K
 \left(\pi^{\mathbf1\{r_a>0\}}(1-\pi)^{\mathbf1\{r_a<0\}}
 \frac{|r_a|^{k_G-1}\exp(-|r_a|/s_a)}{\Gamma(k_G)s_a^{k_G}}\right).
 \label{eq:gamma-return-density}
\end{align}
The coordinate hyperplanes are null and the common support does not depend on the scales. The scores are
\begin{align}
 S_\pi(R)&=\frac{\sum_{a=1}^K(\mathbf1\{R_a>0\}-\pi)}{\pi(1-\pi)},
 &S_{\log s_a}(R)&=|R_a|/s_a-k_G.
 \label{eq:gamma-score}
\end{align}
Independence of signs and magnitudes gives the stated block-diagonal information. Score moments are bounded on the compact parameter ranges, and the gamma tails justify differentiation under the integral and differentiability in quadratic mean.

The unconstrained maximum likelihood estimates of $\pi$ and $s_a$ are the pooled sign proportion and $n^{-1}k_G^{-1}\sum_i|R_{ia}|$, and maximizing over the prescribed intervals gives \eqref{eq:gamma-likelihood-estimates}. Binomial and gamma exponential bounds make clipping exponentially unlikely for an evaluation set strictly inside those intervals. With $X_{ia}=|R_{ia}|/s_a$, the log-scale estimate satisfies
\begin{align}
 \log\widehat s_{a,n}-\log s_a
 =\frac1{nk_G}\sum_{i=1}^n(X_{ia}-k_G)+r_{a,n},
 \qquad
 \sup_{\theta\in\Theta}n\E_\theta[r_{a,n}^2]\longrightarrow0.
 \label{eq:gamma-log-scale-expansion}
\end{align}
Taylor expansion near the positive mean, with gamma fourth moments and exponentially small tails elsewhere, proves this, and the same bounds give \eqref{eq:regular-estimator-fourth} since the clipped log scales are bounded. The sign estimate has the efficient expansion directly, which verifies the estimation part of Assumption~\ref{ass:regular-one-class}.

The condition $\gamma Ls_+<1$ makes the required derivatives of expected utility finite and continuous near the compact parameter and decision sets. Since $\gamma\E_\theta[(u^\top R)^2\exp(-\gamma w^\top R)]>0$ for every nonzero $u$, continuity and compactness give uniform positive lower and finite upper Hessian bounds. The minimizer of $\ell_\pi(\alpha)$ solves $\pi(1+\gamma \alpha)^{-k_G-1}=(1-\pi)(1-\gamma \alpha)^{-k_G-1}$, which gives \eqref{eq:gamma-investment}, and the parameter bounds make $\alpha_\pi/s_a$ positive and uniformly below $L$, verifying the interior condition.

At this optimum, $\ell_\pi'(\alpha_\pi)=0$, so cross derivatives in standardized holdings vanish and the negative Hessian is $j_\pi I_K$. The derivatives of $w_a=\alpha_\pi/s_a$ with respect to $\pi$ and $\log s_{a'}$ are $\alpha_\pi'/s_a$ and $-\mathbf1\{a=a'\}\alpha_\pi/s_a$. Substituting them and \eqref{eq:gamma-score} into \eqref{eq:likelihood-constants} gives \eqref{eq:gamma-allocation-constant}, and the assertions about the fitted risk-parity allocation follow from the general EUR rule results.

\section{Portfolio Restrictions}
\label{app:portfolio-restrictions}
We prove the risk-parity representation of the EUR rule and the comparisons for a prescribed composition and for a money budget.

\subsection{Proof of Proposition~\ref{prop:eurr-risk-parity}}
For $y\in\mathcal C$, transitivity gives $|\mathcal G|^{-1}\sum_{\Pi\in\mathcal G}\Pi y=(\mathbf1^\top y/K)\mathbf1$, which lies in $\mathcal C$ by convexity and invariance. By concavity of $U$ and invariance of the distribution,
\begin{align}
  \E_\theta\left[U\left(x_0+\frac{\mathbf1^\top y}{K}T\right)\right]
  &\geq\frac1{|\mathcal G|}\sum_{\Pi\in\mathcal G}
       \E_\theta[U(x_0+(\Pi y)^\top Z)]
   =\E_\theta[U(x_0+y^\top Z)].
  \label{eq:rp-concavity}
\end{align}
Assumption~\ref{ass:regular-return-model}(ii) gives strict concavity and an interior optimizer, so the unique optimum is $\alpha_\theta\mathbf1$ in standardized coordinates, with $\alpha_\theta>0$ because $\mathcal C$ is full-dimensional and lies in the nonnegative orthant. With one class, the EUR rule evaluates this optimizer at $\widehat\theta(D)$, proving \eqref{eq:rp-eurr-identity}.

The covariance satisfies $\Pi\Sigma_\theta^Z\Pi^\top=\Sigma_\theta^Z$ for all $\Pi\in\mathcal G$, so transitivity makes its diagonal entries equal and gives $\Sigma_\theta^Z\mathbf1=\lambda_\theta\mathbf1$ for $\lambda_\theta=\mathbf1^\top\Sigma_\theta^Z\mathbf1/K>0$. For $\alpha>0$ and $w=\alpha D_s^{-1}\mathbf1$, $w_a(\Sigma_\theta w)_a=\alpha^2\lambda_\theta$ and $w^\top\Sigma_\theta w=K\alpha^2\lambda_\theta$, so \eqref{eq:rp-contributions} gives $\operatorname{RC}_{a,\theta}(w)=\sigma_\theta(w)/K$. Normalizing the risky holdings by their sum proves \eqref{eq:rp-capital-composition}.

\subsection{Proof of Proposition~\ref{prop:local-composition}}
Differentiating $\nabla_wV_{P_\theta}(w_1(\theta))=0$ at $\theta_0$ gives $\nabla_\theta\nabla_wV_{P_\theta}(\alpha_0v)|_{\theta=\theta_0}=J_{\theta_0}D_{\theta_0}$. The scalar first-order condition for $\alpha^v(\theta)$ then gives
\begin{align}
 v\nabla_\theta \alpha^v(\theta_0)^\top
 =v(v^\top J_{\theta_0}v)^{-1}v^\top J_{\theta_0}D_{\theta_0}=\Pi_vD_{\theta_0}=D_\parallel.
 \label{eq:restricted-composition-derivative}
\end{align}
Under $\theta_{n,h}=\theta_0+h/\sqrt n$, the likelihood expansion implies, uniformly for $h$ in a compact set,
\begin{align}
 \sqrt n\left(\widehat w_n^v-w_1(\theta_{n,h})\right)
 =D_\parallel\sqrt n(\widehat\theta(D)-\theta_{n,h})-D_\perp h+o_{L^2}(1).
 \label{eq:restricted-composition-expansion}
\end{align}
The centered estimator on the right has mean tending to zero and covariance tending to $I_{\theta_0}^{-1}$, and bounded fourth moments control Taylor remainders and estimates outside the local neighborhood, so quadratic expansion of regret gives \eqref{eq:composition-restricted-regret}. For the EUR rule, the expansion has $D_{\theta_0}$ in place of $D_\parallel$ and no bias term, and since $D_\parallel^\top J_{\theta_0}D_\perp=0$ its trace constant is $c_\parallel+c_\perp$, proving \eqref{eq:composition-full-regret}. For a deterministic perturbation $th$, the two population portfolios differ by $tD_\perp h+O(t^2)$, and Taylor expansion at the unrestricted optimum proves \eqref{eq:composition-population-cost}.

\subsection{Proof of Proposition~\ref{prop:rp-budget}}
When $\mathcal B\geq\mathcal B_0$, the original optimizer remains feasible and solves the restricted problem. At a symmetric allocation, invariance under the transitive group makes the gradient a multiple of $\mathbf1$. If the money budget is the only active constraint, stationarity sets this gradient equal to $\lambda_{\mathcal B}c$, which, for $\lambda_{\mathcal B}>0$, requires the entries of $c$ to be equal.

Near $\mathcal B_0$, the bordered matrix for the equations $\nabla\phi_\theta(y)=\lambda_{\mathcal B}c$ and $c^\top y=\mathcal B$ is nonsingular because $J_\phi\succ0$ and $c\ne0$. The implicit function theorem gives the first expansion in \eqref{eq:budget-allocation-expansion}, with $\lambda_{\mathcal B}=\Delta/(c^\top J_\phi^{-1}c)+O(\Delta^2)>0$ for small $\Delta>0$. Strict concavity of the scalar objective makes its optimum under the tighter budget equal to $\alpha_\theta-\Delta/(c^\top\mathbf1)$, proving the second expansion.

When both portfolios are written relative to $y_0$ with $\nabla\phi_\theta(y_0)=0$, their quadratic losses from the original optimum are $\Delta^2/(2c^\top J_\phi^{-1}c)$ and $\Delta^2\mathbf1^\top J_\phi\mathbf1/(2(c^\top\mathbf1)^2)$, up to $O(\Delta^3)$, and their difference proves \eqref{eq:budget-composition-cost}. By Cauchy--Schwarz the coefficient is nonnegative, with equality precisely when $c$ is proportional to $J_\phi\mathbf1$. Since invariance and transitivity give $J_\phi\mathbf1=\lambda_0\mathbf1$ for some $\lambda_0>0$, equality occurs only when the relative volatilities are equal.

\end{document}